\def\confversion{0}
\def\ifconf{\ifnum\confversion=1}
\def\ifnotconf{\ifnum\confversion=0}

\documentclass[12pt]{article}
\usepackage[utf8]{inputenc}
\usepackage{comment}

\def\showauthornotes{1}
\def\showkeys{0}
\def\showdraftbox{0}
\usepackage{xspace,xcolor,enumerate}
\usepackage{amsmath,amssymb}
\usepackage{amsthm}
\usepackage[toc,page]{appendix}
\usepackage{thmtools}
\usepackage{thm-restate}
\usepackage{color,graphicx}
\usepackage{boxedminipage}
\usepackage{makecell}

\ifnum\showkeys=1
\usepackage[color]{showkeys}
\fi

\definecolor{darkred}{rgb}{0.5,0,0}
\definecolor{darkgreen}{rgb}{0,0.35,0}
\definecolor{darkblue}{rgb}{0,0,0.55}

\usepackage[pdfstartview=FitH,pdfpagemode=UseNone,colorlinks,linkcolor=darkblue,filecolor=darkred,citecolor=darkgreen,urlcolor=darkred,pagebackref]{hyperref}

\usepackage[capitalise,nameinlink]{cleveref}
\usepackage[T1]{fontenc}
\usepackage{mathtools,dsfont,bbm}
\usepackage{mathpazo}
\usepackage[top=1in, bottom=1in, left=1in, right=1in]{geometry}
\usepackage{fullpage}

\ifnum\showauthornotes=1
\newcommand{\Authornote}[2]{{\sf\small\color{red}{[#1: #2]}}}
\newcommand{\Authorcomment}[2]{{\sf \small\color{gray}{[#1: #2]}}}
\newcommand{\Authorfnote}[2]{\footnote{\color{red}{#1: #2}}}
\else
\newcommand{\Authornote}[2]{}
\newcommand{\Authorcomment}[2]{}
\newcommand{\Authorfnote}[2]{}
\fi

\ifnum\showdraftbox=1
\newcommand{\draftbox}{\begin{center}
  \fbox{%
    \begin{minipage}{2in}%
      \begin{center}%
        \begin{Large}%
          \textsc{Working Draft}%
        \end{Large}\\
        Please do not distribute%
      \end{center}%
    \end{minipage}%
  }%
\end{center}
\vspace{0.2cm}}
\else
\newcommand{\draftbox}{}
\fi

\newtheorem{theorem}{Theorem}[section]

\newtheorem{definition}[theorem]{Definition}

\newtheorem{lemma}[theorem]{Lemma}
\newtheorem{remark}[theorem]{Remark}
\newtheorem{proposition}[theorem]{Proposition}
\newtheorem{corollary}[theorem]{Corollary}
\newtheorem{claim}[theorem]{Claim}
\newtheorem{fact}[theorem]{Fact}

\newtheorem{algo}[theorem]{Algorithm}

\newtheorem{bound}[theorem]{Bound}

\def\FullBox{\hbox{\vrule width 6pt height 6pt depth 0pt}}

\def\qed{\ifmmode\qquad\FullBox\else{\unskip\nobreak\hfil
\penalty50\hskip1em\null\nobreak\hfil\FullBox
\parfillskip=0pt\finalhyphendemerits=0\endgraf}\fi}

\def\qedsketch{\ifmmode\Box\else{\unskip\nobreak\hfil
\penalty50\hskip1em\null\nobreak\hfil$\Box$
\parfillskip=0pt\finalhyphendemerits=0\endgraf}\fi}

\def\to{\rightarrow}
\def\eps{\varepsilon}
\def\epsilon{\varepsilon}

\def\eps{\epsilon}

\def\phi{\varphi}
\def\cal{\mathcal}

\def\psdgeq{\succeq} 
\newcommand{\defeq}{:=}

\newcommand{\R}{{\mathbb R}}
\newcommand{\E}{{\mathbb E}}

\newcommand{\N}{{\mathbb{N}}}

\usepackage{nicefrac}

\newcommand{\abs}[1]{\ensuremath{\left\lvert #1 \right\rvert}}

\newcommand{\norm}[1]{\ensuremath{\left\lVert #1 \right\rVert}}

\newcommand{\ip}[2] {\ensuremath{\langle #1 , #2 \rangle}}

\newcommand{\one}{{\mathbf{1}}}

\newcommand{\yes}{{\sf Yes}\xspace}

\makeatletter

\def\ProbabilityRender#1#2{
  \@ifnextchar\bgroup%
  {\renderwithdist{#1}{#2}}
   {\singlervrender{#1}{#2}}
}
\def\singlervrender#1#2{%
   \ensuremath{\mathchoice
       {{#1}\left[ #2 \right]}
       {{#1}[ #2 ]}
       {{#1}[ #2 ]}
       {{#1}[ #2 ]}
   }
}
\def\renderwithdist#1#2#3{%
   \@ifnextchar\bgroup
   {\superfancyrender{#1}{#2}{#3}}
   {\ensuremath{\mathchoice
      {\underset{#2}{#1}\left[ #3 \right]}
      {{#1}_{#2}[ #3 ]}
      {{#1}_{#2}[ #3 ]}
      {{#1}_{#2}[ #3 ]}
     }
   }
}
\def\superfancyrender#1#2#3#4#5{
   \ensuremath{\mathchoice
      {\underset{#1}{{#1}}\left#4 #3 \right#5}
      {{#1}_{#2}#4 #3 #5}
      {{#1}_{#2}#4 #3 #5}
      {{#1}_{#2}#4 #3 #5}
   }
}
\makeatother

\newfont{\inhead}{eufm10 scaled\magstep1}

\newcommand{\calA}{{\cal A}}
\newcommand{\calB}{{\cal B}}
\newcommand{\calC}{{\cal C}}
\newcommand{\calD}{{\cal D}}
\newcommand{\calE}{{\cal E}}
\newcommand{\calF}{{\cal F}}
\newcommand{\calG}{{\cal G}}
\newcommand{\calH}{{\cal H}}
\newcommand{\calI}{{\cal I}}
\newcommand{\calJ}{{\cal J}}
\newcommand{\calK}{{\cal K}}
\newcommand{\calL}{{\cal L}}
\newcommand{\calM}{{\cal M}}
\newcommand{\calN}{{\cal N}}
\newcommand{\calO}{{\cal O}}
\newcommand{\calP}{{\cal P}}
\newcommand{\calQ}{{\cal Q}}
\newcommand{\calR}{{\cal R}}
\newcommand{\calS}{{\cal S}}
\newcommand{\calT}{{\cal T}}
\newcommand{\calU}{{\cal U}}
\newcommand{\calV}{{\cal V}}
\newcommand{\calW}{{\cal W}}
\newcommand{\calX}{{\cal X}}
\newcommand{\calY}{{\cal Y}}
\newcommand{\calZ}{{\cal Z}}

\newcommand{\poly}{{\mathrm{poly}}}

\DeclareMathOperator{\tr}{\operatorname {tr}}

\DeclareMathOperator{\sgn}{\operatorname{sgn}}

\newcommand{\ceil}[1]{\ensuremath{\left\lceil #1 \right\rceil}}
\newcommand{\floor}[1]{\ensuremath{\left\lfloor #1 \right\rfloor}}

\newcommand{\bigtheta}{\mathop{\Theta}}
\newcommand{\bigomega}{\mathop{\Omega}}

\newcommand{\littleomega}{\mathop{\omega}}

\newcommand{\softomega}{\mathop{\widetilde{\Omega}}}
\let\bigO=\bigoh

\let\littleo=\littleoh

\let\softO=\softoh

\newcommand{\classfont}[1]{\textsf{#1}}

\DeclareMathOperator*{\pseudoE}{\widetilde{\mathbb{E}}}
\let\pE=\pseudoE
\DeclareMathOperator*{\pseudoPr}{\widetilde{\Pr}}

\newcommand{\Erdos}{Erd\H{o}s\xspace}
\newcommand{\Renyi}{R\'enyi\xspace}

\newcommand{\Id}{\mathbf{Id}}

\newcommand{\trunc}{ \textnormal{\textsf{{truncation}}}}

\usepackage[varg]{pxfonts}

\newcommand{\T}{\intercal}

\newcommand{\NN}{\mathbb{N}}

\newcommand{\sig}{\sigma}

\newcommand{\al}{\alpha}

\newcommand{\lam}{\lambda}

\newcommand{\Lam}{\Lambda}
\newcommand{\Gam}{\Gamma}
\newcommand{\gam}{\gamma}

\newcommand{\del}{\delta}
\newcommand{\dv}{D_{V}}
\newcommand{\dsos}{D_{\textnormal{SoS}}}
\DeclareMathOperator{\Bernoulli}{Bernoulli}

\DeclareMathOperator{\Aut}{Aut}
\DeclareMathOperator{\Iso}{Iso}
\DeclareMathOperator{\mult}{mult}

\newcommand{\psdleq}{\preceq}
\newcommand{\psdl}{\prec}

\newcommand{\psdg}{\succ}

\newcommand{\sleq}{{\scriptstyle \le}}

\newcommand{\ssleq}{{\scriptscriptstyle \le}}

\newcommand{\mM}{\mathbf{M}}
\newcommand{\matD}{\mathbf{D}}
\newcommand{\matL}{\mathbf{L}}
\newcommand{\matQ}{\mathbf{Q}}
\newcommand{\matM}{\mathbf{M}}
\newcommand{\matLam}{\mathbf{\Lambda}}
\newcommand{\matId}{\mathbf{Id}}

\DeclareMathOperator{\slack}{slack}

\newcommand{\normapx}[1]{\norm{\mM^\approx_{#1}}}

\usepackage[outline]{contour} 
\contourlength{0.1pt}
\contournumber{10}
\newcommand{\xsos}{\textnormal{\textsf{\contour{black}{X}}}}

\newcommand{\idsym}{\matId_{sym}}

\ifnum\showauthornotes=1
\newcommand{\anote}[1]{{\sf\small\color{orange}{ [Aaron: #1] }}}
\newcommand{\cnote}[1]{{\sf\small\color{blue}{ [Chris: #1] }}}
\newcommand{\gnote}[1]{{\sf\small\color{red}{ [Goutham: #1] }}}
\newcommand{\jnote}[1]{{\sf\small\color{violet}{ [Jeff: #1] }}}
\newcommand{\mnote}[1]{{\sf\small\color{olive}{ [Madhur: #1] }}}
\else
\newcommand{\anote}[1]{}
\newcommand{\cnote}[1]{}
\newcommand{\gnote}[1]{}
\newcommand{\jnote}[1]{}
\newcommand{\mnote}[1]{}
\fi

\allowdisplaybreaks

\begin{document}
\title{Sum-of-Squares Lower Bounds for\\Densest $k$-Subgraph}

\author{
  Chris Jones\thanks{{\tt Bocconi University}. {\tt chris.jones@unibocconi.it}. Supported in part by the ERC under the European Union's Horizon 2020 research and innovation programme (grant agreement No. 834861).}
  \and
  Aaron Potechin\thanks{{\tt University of Chicago}. {\tt potechin@uchicago.edu}. Supported in part by NSF grant CCF-2008920.}
  \and  
  Goutham Rajendran\thanks{{\tt Carnegie Mellon University}. {\tt gouthamr@cmu.edu}. Supported in part by NSF grants CCF-1816372 and CCF-2008920.}
  \and
  Jeff Xu\thanks{{\tt Carnegie Mellon University}. {\tt jeffxusichao@cmu.edu}. Supported in part by NSF CAREER Award \#2047933 and CyLab Presidential Fellowship.}
}

\maketitle
 \thispagestyle{empty}
\draftbox

\thispagestyle{empty}

\begin{abstract}

Given a graph and an integer $k$, Densest $k$-Subgraph is the algorithmic task of finding the subgraph on $k$ vertices with the maximum number of edges. This is a fundamental problem that has been subject to intense study for decades, with applications spanning a wide variety of fields. The state-of-the-art algorithm is an $O(n^{1/4 + \eps})$-factor approximation (for any $\eps > 0$) due to Bhaskara et al. [STOC '10]. Moreover, the so-called \emph{log-density framework} predicts that this is optimal, i.e. it is impossible for an efficient algorithm to achieve an $O(n^{1/4 - \eps})$-factor approximation. In the average case, Densest $k$-Subgraph is a prototypical noisy inference task which is conjectured to exhibit a \emph{statistical-computational gap}.

In this work, we provide the strongest evidence yet of hardness for Densest $k$-Subgraph by showing matching lower bounds against the powerful Sum-of-Squares (SoS) algorithm, a meta-algorithm based on convex programming that achieves state-of-art algorithmic guarantees for many optimization and inference problems.
For $k \leq n^{\frac{1}{2}}$, 
we obtain a degree $n^{\delta}$ 
SoS lower bound for the hard regime as predicted by the log-density framework.

To show this, we utilize the modern framework for proving SoS lower bounds on average-case problems pioneered by Barak et al. [FOCS '16].
A key issue is that small denser-than-average subgraphs in the input
will greatly affect the value of the candidate pseudoexpectation operator around the subgraph.
To handle this challenge, we devise a novel matrix factorization scheme based on the \emph{positive minimum vertex separator}. We then prove an intersection tradeoff lemma to show that the error terms when using this separator are indeed small.

\end{abstract}

\newpage
\ifnotconf
\pagenumbering{roman}
\setcounter{tocdepth}{2}
	\tableofcontents
\clearpage
\fi

\pagenumbering{arabic}
\setcounter{page}{1}

\section{Introduction}

In the Densest $k$-Subgraph problem, we are given an undirected graph $G$ on $n$ vertices and an integer $k$ and we want to output the subgraph on $k$ vertices with the most edges, or in other words, the subgraph on $k$ vertices with the highest edge density.
This is a natural generalization of the $k$-clique problem \cite{karp1972reducibility} and has been subject to a long line of work for decades \cite{feige1997densest, srivastav1998finding, feige2001dense, feige2001approximation, asahiro2002complexity, feige2002relations, khot2006ruling, goldstein2009dense, BCCFV10, raghavendra2010graph, alon2011inapproximabilty, bhaskara2012polynomial, barman2015approximating, hajek2015computational, ames2015guaranteed, hajek2016achieving,
braverman2017eth, manurangsi2017almost, bombina2020convex, khanna2020planted}.
This problem has been the subject of intense study partly because of its numerous connections to other problems and fields (e.g. \cite{hajiaghayi2006prize, hajiaghayi2006minimum, kolliopoulos2007partially, pisinger2007quadratic, kortsarz2008approximating, andersen2009finding, charikar2011improved, hamada2011hospitals, li2014complexity, chuzhoy2015approximation, chekuri2015note, chen2015combining, skowron2016finding, chestnut2017hardness, tirodkar2017approximability, lee2017partitioning, chlamtac2018densest, mao2023detectionrecovery})
The best known approximation algorithm for this problem yields an approximation factor of $O(n^{1/4 + \eps})$ for any constant $\eps > 0$, due to \cite{BCCFV10}. 
On the other hand, it is conjectured that no efficient algorithm can achieve an $O(n^{1/4 - \eps})$ approximation.

Densest $k$-Subgraph is
a compelling problem
because random instances (\Erdos-\Renyi graphs) are conjectured and widely believed to be the ''hardest'' instances for algorithms.
In fact, the insight that ``worst case is average case'' was crucial to the aforementioned algorithm in \cite{BCCFV10}.
Their idea of going from average-case instances to worst-case instances was generalized into the \textit{log-density framework} (more in \cref{subsec: logdensity}), which has been further applied to various other problems \cite{chlamtac2012everywhere, chlamtavc2017minimizing, chlamtavc2017approximation}. 
Since an algorithm for random instances seems to be the crucial conceptual step needed to solve the problem on all instances, understanding these random instances is a pressing topic.

As stated in \cite{BCCFV10, bhaskara2012polynomial, braverman2017eth, manurangsi2017almost}, Densest $k$-Subgraph on a random graph 
is a landmark question in the field of average-case complexity.
Moreover, the conjectured hardness of this problem on random instances (which is the focus of our work) has been used for applications in finance \cite{arora2010computational} and cryptography \cite{applebaum2010public}.
However, evidence of hardness for Densest $k$-Subgraph stands to be improved, both in the average-case and worst-case settings.
For example, even in the worst-case setting, no work has been able to show that Densest $k$-Subgraph is hard to $n^\eps$-approximate for a fixed $\eps > 0$ using any reasonable complexity-theoretic assumption (although some works come close, see \cref{related_work}).
In the more interesting average-case setting of random graphs, relatively little progress has been made to justify hardness, let alone match the log-density framework.

In this work, we study the hardness of Densest $k$-Subgraph on random graphs through
a generic, powerful algorithm for optimization known as the Sum-of-Squares (SoS) hierarchy \cite{shor1987approach, nesterov2000squared, parrilo2000structured, grigoriev2001complexity, lasserre2001global}.
The SoS hierarchy is a family of semidefinite programming relaxations for polynomial optimization problems which implements a certain type of ``sum-of-squares reasoning''.
Arguably at the center stage of average-case complexity in recent years, SoS has proven to be a highly effective tool for combinatorial and continuous optimization.
Indeed, the SoS hierarchy is rich enough to capture the state-of-the-art convex relaxations for Sparsest Cut \cite{AroraRV04}, Max-Cut \cite{GW94}, all Max $k$-CSPs \cite{Raghavendra08}, etc. 
Sum-of-Squares has also led to new algorithms for approximating CSPs \cite{AJT19, bafna2021playing, bafna2022high} and breakthroughs in robust statistics \cite{kothari2017outlier, hopkins2018mixture, raghavendra2018high, klivans2018efficient, hopkins2020mean, bakshi2021robust, bakshi2020list}, a highlight being the resolution of longstanding open problems in Gaussian mixture learning (over a decade of work culminating in \cite{bakshi2020robustly, liu2021settling}).
Moreover, for a large class of problems, it has been shown that SoS algorithms are the most effective among all semidefinite programming relaxations \cite{lee2015lower}.
Therefore, understanding the limits of SoS algorithms is an important research endeavour and 
lower bounds against SoS serve as strong evidence for algorithmic hardness \cite{hop17, hop18, kunisky2021spectral}.

In this paper, we prove that for $k \leq n^{\frac{1}{2}}$, SoS of degree $n^{\delta}$ does not offer any significant improvement in the conjectural hard regime of random instances for Densest $k$-Subgraph as predicted by the log-density framework. This settles the open questions raised in the works \cite{bhaskara2012polynomial, rajendran2022combinatorial, chlamtavc2018sherali}.
Considering that the algorithm of Bhaskara et al. \cite{BCCFV10} matching the log-density framework is captured by SoS, our lower bound completes the picture of the performance of SoS for Densest $k$-Subgraph for $k \leq n^{\frac{1}{2}}$. 
This gives solid evidence that the conjectured approximability thresholds for Densest $k$-Subgraph are correct.

\subsection{Our contributions}\label{subsec: contributions}

We will now describe our results on SoS lower bounds for Densest $k$-Subgraph that match the predictions of the log-density framework (to be described in \cref{subsec: logdensity}).

Consider the following hypothesis testing variant of the Densest $k$-Subgraph problem. For an integer $n$ and a real $p \in [0, 1]$, let $\calG_{n, p}$ denote the \Erdos-\Renyi random distribution where a graph on $n$ vertices is sampled by choosing each edge to be present independently with probability $p$. For parameters $n, k \in \NN$ and $p, q \in [0, 1]$, we are given a graph $G$ sampled either from
\begin{enumerate}
    \item The null distribution $\calG_{n, p}$ or
    \item The alternative distribution where we first sample $G \sim \calG_{n, p}$, then a set $H \subseteq V(G)$ is chosen by including each vertex with probability $\frac{k}{n}$, and finally we replace $H$ by a sample from $\calG_{|H|, q}$.
\end{enumerate}
and our goal is to correctly identify which distribution it came from, with non-negligible probability.

The hypothesis testing question is a ``planted model'' of Densest $k$-Subgraph which is conjectured to exhibit a \emph{statistical-computational gap} \cite{brennan2020reducibility,brennan2020statistical}.
With high probability, for $q$ slightly larger than $p$, the subgraph $H$ in the alternative distribution is truly the densest subgraph of $G$ with size $k$ (hence the
null and alternative distributions are statistically distinguishable),
but it is conjecturally computationally impossible to distinguish the two cases (in the parameter regime below).

Studying algorithms for this hypothesis testing variant was crucial
to the log-density framework~\cite{BCCFV10},
which both generalizes an algorithm for the hypothesis testing variant
into a worst-case algorithm, and predicts the relationships between $n, k, p, q$ for which the hypothesis testing problem is hard.
In particular, consider the setting 
\[k = n^{\al}, \qquad\qquad p = n^{-\beta}, \qquad\qquad q = n^{-\gam}\]
for constants $\al \in (0, 1/2], \beta \in (0, 1), \gam \in (0, 1)$, a notation that we will use throughout this paper. According to the framework, it's algorithmically hard to solve the problem if
\[\gam > \al\beta\]
That is, in this regime, no polynomial-time algorithm can distinguish the two distributions with probability at least $2/3$ of success.\footnote{When $\al > \frac{1}{2}$, i.e. $k = \omega(\sqrt{n})$, spectral algorithms beat the log-density threshold \cite{BCCFV10, khanna2020planted}.
Spectral algorithms are captured by degree-2 SoS. Various works have also studied other special settings (e.g. when $q = 1$, or when $p, q$ are constants). See \cref{related_work}.}

To state our result, we recall that the SoS hierarchy is a family of convex semidefinite programming relaxations parameterized by an integer $\dsos$ called the 
\emph{degree} or \emph{level} of SoS. The relaxation gets tighter as $\dsos$ increases but the runtime also increases at the rate\footnote{In pathological cases, there may be issues with bit complexity \cite{o2017sos, RW17:sos}} of approximately $n^{O(\dsos)}$ for degree $\dsos$ SoS.
Thus, conceptually degree $O(1)$ corresponds to polynomial time, and degree $n^\delta$ to subexponential time algorithms.
In this work, we study the performance of degree $\dsos = n^{\delta}$ Sum-of-Squares on the Densest $k$-Subgraph problem for a constant $\del > 0$ and obtain strong lower bounds.

Because of the well-known duality between SoS programs and pseudo-expectation operators, to show a lower bound, it suffices to show a feasible pseudo-expectation operator $\pE$ satisfying the constraints.
For a formal definition of SoS, see \cref{subsec: sos}. We are now ready to state our result.

\begin{theorem}\label{thm:main}
    For all constants $\al \in (0, 1/2], \beta \in (0,1), \gam \in (0,1)$
    such that 
    $\gam > \alpha \beta$, there exists $\delta > 0$ such that
    with high probability over $G = (V, E) \sim \calG_{n,p}$, there exists a degree $n^{\delta}$ pseudo-expectation operator $\pE$ on SoS program variables $\{\xsos_u\}_{u \in V}$ such that
    \begin{enumerate}
        \item (Normalization) $\pE[1] = 1 \pm o(1)$.
        \item (Subgraph on $k$ vertices) $\pE[\sum_{v \in V} \xsos_v] = k(1 \pm o(1))$.
        \item (Large density) $\pE[\sum_{\{u,v\} \in E} \xsos_u \xsos_v] = \frac{k^2q}{2}(1 \pm o(1))$
        \item (Feasibility) The moment matrix $\mM$ corresponding to $\pE$ is positive semidefinite.
    \end{enumerate}
\end{theorem}

This in particular implies that, in the predicted hard regime of the log-density framework, SoS cannot be used to solve the Densest $k$-Subgraph problem as stated above.
As discussed earlier, these SoS lower bounds offer strong evidence that for $k \leq \sqrt{n}$, it is unlikely that efficient algorithms can beat the predictions of the log-density framework for Densest $k$-Subgraph.

By setting $\al = 1/2, \beta = 1/2$ and $\gam = 1/4 + \eps$, we obtain the following important corollary.

\begin{corollary}\label{cor:main}
    For any $\eps > 0$, 
    there exists a constant $\delta > 0$ such that degree-$n^{\del}$ Sum-of-Squares exhibits an integrality gap of $O(n^{1/4 - \eps})$ for the Densest $k$-Subgraph problem. 
\end{corollary}

This corollary essentially matches the best known algorithmic guarantees for the Densest $k$-subgraph problem \cite{BCCFV10}, namely an efficient $O(n^{1/4 + \eps})$-factor approximation algorithm, thereby completing the picture for Sum-of-Squares.

\subsection{The log-density framework}\label{subsec: logdensity}

For more context on our results, we give a brief description of the log-density framework \cite{BCCFV10}. See \cite[Section 1.3]{khanna2020planted} or \cite{chlamtavc2018sherali} for a more detailed treatment.

The log-density framework is a relatively recent technique that devises worst-case algorithms for problems by studying algorithms for average-case instances. It was introduced in the context of the Densest $k$-Subgraph problem and has been since utilized for many other problems such as
Lowest Degree 2-Spanner, Smallest p-Edge Subgraph (SpES) \cite{chlamtac2012everywhere}, Small
Set Bipartite Vertex Expansion (SSBVE) \cite{chlamtavc2017minimizing}, Label Cover, 2-CSPs \cite{chlamtavc2017approximation}, etc.

Formally, for a graph on $n$ vertices with average degree $d$, we define its log-density to be $\frac{\log d}{\log n}$.
Consider the hypothesis testing problem from \cref{subsec: contributions}.
The log-density framework predicts that it is possible to algorithmically distinguish the distributions and solve the hypothesis testing problem if and only if the log-density of the planted subgraph is larger than the log-density of the original graph before planting. Since the average degree of a graph sampled from $G_{n, p}$ is $\approx np$,
this framework predicts that the distributions are distinguishable if for some constant $\eps > 0$,
\[\frac{\log (kq)}{\log k} \ge \frac{\log (np)}{\log n} + \eps \qquad\Longleftrightarrow\qquad \gam \le \al\beta - \eps'\]
for some constant $\eps' > 0$. 

Moreover, and of extreme importance to us, the framework also predicts algorithmic hardness if the other direction of the inequality holds. That is, if \[\gam \ge \al\beta + \eps\] for some constant $\eps > 0$, the log-density framework predicts that no efficient algorithm can distinguish the two distributions.
For the sake of clarity, let's look at the special case $\al = 1/2, \beta = 1/2$ and $\gam = 1/4 + \eps$. Then, we expect it to be hard for efficient algorithms to distinguish the following distributions,
\begin{enumerate}
    \item The null distribution $\calG_{n,\frac{1}{\sqrt{n}}}$ 
    \item The alternative distribution where we first sample $G \sim \calG_{n, \frac{1}{\sqrt{n}}}$, then a set $H \subseteq V(G)$ is chosen by including each vertex with probability $\frac{1}{\sqrt{n}}$ (so $|H| \approx \sqrt{n}$), and finally we replace $H$ by a sample from $\calG_{|H|, \frac{1}{n^{1/4 + 
 \eps}}}$.
\end{enumerate}

As an aside, note that in this case, since the average degree of the densest $k$-subgraph for the null distribution is $\widetilde{O}(\sqrt{n})$ and that of the alternative distribution is $\widetilde{\Omega}(n^{3/4 - \epsilon})$, hardness of the distinguishing problem implies $n^{1/4 - \epsilon}$-factor approximation hardness for Densest $k$-Subgraph.

\subsection{Our approach}

Since Sum-of-Squares is a convex program, in order to prove a lower bound,
it suffices to construct a feasible point, i.e. a \emph{pseudoexpectation operator} or \emph{moment matrix}, which is a large nonlinear random matrix that depends on the input. At a high level, our proof leverages an existing strategy for proving lower bounds
against the Sum-of-Squares algorithm on random inputs:
use \emph{pseudocalibration} \cite{BHKKMP16} to construct a candidate moment matrix, 
then study the spectrum of the candidate matrix using \emph{graph matrices} \cite{AMP20}. This strategy has been successfully applied in several contexts \cite{BHKKMP16, potechin2020machinery, GJJPR20, JPRTX22},
although in each case, including ours, significant additional insights have been required.

Given a random input graph, the first step is to construct the candidate pseudoexpectation operator or moment matrix.
Pseudocalibration
suggests a candidate matrix, which we can use here without further thinking.
Recall that a semidefinite program optimizes over the cone of positive semi-definite (PSD) matrices;
the main challenge is showing that the candidate moment matrix is feasible (PSD) with high probability over the random input.

The main issue we face is that matrix factorization strategies in prior works do not obviously lead to dominant PSD terms in our setting. 
There are several steps in the existing framework:
\begin{enumerate}
    \item Express the candidate moment matrix $\matLam$ in the graph matrix (i.e. Fourier) basis;
    \item Identify a class of spectrally dominant graph matrices in $\matLam$ which are together approximately PSD;
    \item Perform an approximate PSD decomposition to create PSD terms plus additional error terms;
    \item Show that all non-dominant terms and error terms can be charged to the dominant PSD terms, i.e. they are ``negligible''.
\end{enumerate}

For the purposes of the current discussion, it is enough to know that
each graph matrix in step (1) measures how a fixed small subgraph, or \emph{shape}, contributes to the candidate moment matrix, and furthermore that the spectral norm of a
graph matrix can be read off of combinatorial properties of the small shape graph.
It was shown in \cite{JPRTX22, rajendran2022concentration} that the norm of a graph matrix is determined up to lower-order factors by the \emph{Sparse Minimum-weight Vertex Separator (SMVS)} of the shape (\cref{thm:norm-bounds-informal}).
For intuition, shapes with smaller, denser separators have larger norms.

In order to identify the class of norm-dominant shapes in step (2),
previous work decomposes shapes using their leftmost and rightmost \emph{MVS} (in contrast to \emph{SMVS}), yielding for each shape an approximately PSD term that spectrally dominates the original graph matrix.
Using the norm bounds, combinatorial arguments about vertex separators are then employed
to show that all deviation terms in step (4) are small.

Although prior work has avoided using the SMVS as the decomposition criterion and used the MVS instead,
the SMVS is a necessity in our setting, because Densest $k$-Subgraph is sensitive
to small, local structures in the input.
To explain, for a fixed set of vertices $U$, if many vertices in $U$ have a common exterior neighbor or are part of a denser-than-average subgraph, then this greatly increases the algorithm's belief that $U$ is part of the dense subgraph.
Using the SMVS can be thought of as pinpointing, for each shape,
the small dense subgraph which has the strongest effect on the graph matrix's norm.

A decomposition based on SMVS poses new conceptual challenges. 
Surprisingly, the SMVS decomposition, without extra care, may rather lead to some supposedly ``PSD'' terms being negative instead.  We address these technical challenges, alongside our solution using the \emph{Positive Minimum-weight Vertex Separator} (see \cref{sec:proof-outline} for a technical overview) after providing the definitions needed for working with graph matrices.

Once we have properly identified the dominant PSD terms, what remains is to prove that the
error terms in the decomposition are small using an \emph{intersection tradeoff lemma}.
This is also one of our novel contributions as it is significantly different from intersection lemmas in prior works.
This combinatorial lemma is the most crucial part of the proof, as it ensures that the error
terms in the approximate PSD decomposition have small enough norms.

It's worth highlighting that the log-density criterion $\gam > \al\beta$ occurs multiple times throughout our proof, which is fascinating to the authors. A partial explanation is that if we look at the contribution of each Fourier character in \cref{lem:GeneralCoefficientTimesNorm}, the quantity $\gam - \al\beta$ measures the decay as the degree of the Fourier character increases, i.e. it's the edge decay in a shape. Therefore, this has a dampening effect on the higher Fourier levels in the decomposition. Such a Fourier decay is ubiquitous in the analysis of the low-degree likelihood ratio \cite{hop17, hop18, kunisky2022notes} and has been important in prior average-case SoS lower bounds \cite{BHKKMP16, potechin2022subexponential, GJJPR20, JPRTX22}.

\subsection{Related work}\label{related_work}

\paragraph{Algorithms}

Algorithms for the Densest $k$-Subgraph problem have been widely studied, e.g. \cite{feige1997densest, srivastav1998finding, feige2001dense, feige2001approximation, asahiro2002complexity, suzuki2008dense, goldstein2009dense, BCCFV10, manurangsi2015approximating, ames2015guaranteed, barman2015approximating, braverman2017eth, bombina2020convex, khanna2020planted}, and we do not attempt to give an overview of them (see e.g. \cite{khanna2020planted} for a nice overview of some of them). For general graphs, the work \cite{kortsarz1993choosing} (which also introduced the problem) gave a polynomial time $\tilde{O}(n^{0.3885})$-factor approximation algorithm. This was later improved to a $O(n^{1/3 - \eps})$-factor approximation (for a constant $\eps \approx 1/60$) in \cite{feige2001dense} and to a $O(n^{0.3159})$-factor approximation in \cite{goldstein2009dense} respectively. The seminal work of \cite{BCCFV10}, which also proposed the log-density framework improved this to give an algorithm that achieves a $n^{1/4 + \eps}$-factor approximation in time $n^{O(1/\epsilon)}$, for all constants $\eps > 0$. This is conjectured to be the best achievable by efficient algorithms.

\paragraph{Lower bounds for Densest $k$-Subgraph}

Because of its conceptual significance and wide applicability, studying lower bounds against the Densest $k$-Subgraph problem is an important research endeavour. We give a non-exhaustive list of such prior works below. 
\begin{enumerate}
    
    \item Conditional hardness: It's well known that Densest $k$-Subgraph is NP-hard to solve exactly, but to the best of our knowledge, NP-hardness of even constant factor approximation is unknown. There are various other conditional hardness results assuming more than $\textsf{P} \neq \textsf{NP}$, e.g. \cite{feige2002relations, khot2006ruling, raghavendra2010graph, alon2011inapproximabilty, braverman2017eth, manurangsi2017almost}. We highlight the influential work of Manurangsi \cite{manurangsi2017almost}, who assuming the Exponential Time Hypothesis showed almost-polynomial factor hardness for this problem. 
    See the same paper for a more detailed list of other conditional hardness results. It's worth noting that none of these results achieve polynomial factor hardness.
    
    These approaches argue that Densest $k$-Subgraph is hard by reduction.
    One source of difficulty is that
    reductions are not as successful for average-case problems, since a reduction
    tends to distort the input distribution and produce somewhat pathological outputs.
    Proving hardness of Densest $k$-Subgraph may be possible using a reduction to a novel non-random instance, but, if it is true that random (or sufficiently pseudorandom) graphs are the
    \emph{only} hard instances of Densest $k$-Subgraph, then a stronger theory of average-case reductions
    may be a prerequisite.
    Some recent works make exciting progress on realizing average-case reductions \cite{brennan2018reducibility, brennan2020reducibility, boix2021average, hirahara2021nearly}.

    The remaining lower bounds, including ours, are unconditional results that do not rely on any conjectures.
        
    \item Sherali-Adams hardness: An integrality gap of $\Omega(n^{\al(1 - \al) - o(1)})$ was shown for the degree-$\tilde{\Omega}(\log n)$ Sherali-Adams hierarchy (which is a family of linear programming relaxations) in \cite{bhaskara2012polynomial, chlamtavc2018sherali}. Our result is stronger than these Sherali-Adams lower bounds in three important ways. First, we consider SoS rather then Sherali-Adams. The SoS hierarchy captures the Sherali-Adams hierarchy and is known to be much stronger in many cases (e.g., see \cite{khot2015unique, devanur2006integrality, charikar2009integrality, chan2016approximate, kothari2017approximating} in conjunction with \cite{GW94, arora2009expander}) therefore we imply their results. Second, we obtain an $n^{\delta}$ degree lower bound as opposed to an $\tilde{\Omega}(\log n)$ degree lower bound. Finally, while these Sherali-Adams lower bounds are for the particular setting where $\beta = \al$ (the setting that maximizes the integrality gap for a fixed $\al$), our lower bounds work for the entire range of parameters $\al, \beta, \gam$.

    \item SoS hardness: Worst-case SoS lower bounds have been exhibited in \cite{bhaskara2012polynomial, manurangsi2015approximating, chlamtavc2017approximation} obtained by reducing from Max $k$-CSP hardness results, within the SoS framework as pioneered by \cite{Tulsiani09}. However, these SoS lower bounds were not optimal even for worst-case instances, since they didn't match known algorithmic guarantees (to be more precise, they showed an $n^{1/14 - O(\eps)}$-factor lower bound for degree $n^{\eps}$ SoS, whereas $n^{1/4 - O(\eps)}$-factor hardness is conjectured). Our work on the other hand studies average-case instances (as opposed to worst-case) and matches the guarantees of known algorithms. Therefore, we significantly improve these prior hardness results and close the gap. Moreover, our results can be reduced {\`a} la \cite{Tulsiani09} to show SoS hardness for other problems such as Densest $k$-Subhypergraph \cite[Theorem 3.17]{rajendran2022combinatorial} and also potentially Minimum $p$-Union \cite{rajendran2022combinatorial}.
\end{enumerate}

\paragraph{Average-case Sum-of-Squares lower bounds}

Sum-of-Squares lower bounds for average-case problems have proliferated in the last decade, for example, Planted Clique \cite{hopkins2015sos, meka2015sum, BHKKMP16}, Sherrington-Kirkpatrick Hamiltonian \cite{MRX20, GJJPR20, Kunisky20}, Sparse and Tensor PCA \cite{hop17, potechin2020machinery, potechin2022subexponential} and Max $k$-CSPs \cite{KMOW17}.
Most of these works have been in the colloquial ``dense'' regime where the random inputs are sampled from $\calG_{n, 1/2}$ or the standard normal distribution $\calN(0, 1)$. Recently, average-case SoS lower bounds have been shown for the sparse setting, i.e. inputs sampled from $\calG_{n, p}$ where $p = o(1)$, for the problem of Maximum Independent Set \cite{JPRTX22, rajendran2022concentration}.
The common thread underlying recent SoS lower bounds, including ours, is spectral analysis of large random
matrices.
See the works \cite{potechin2020machinery, rajendran2022nonlinear, Jones:thesis} for additional background and intuition on the matrix analysis framework used in these lower bounds.

\paragraph{The low-degree likelihood ratio hypothesis}
We add that similar predictions as the log-density framework for the threshold of algorithmic distinguishability may possibly be obtained by analyzing the \emph{low-degree likelihood ratio} \cite{hop17, hop18, kunisky2022notes}.
The low-degree likelihood ratio is used in the context
of noisy statistical inference problems to predict, among other things, the existence of
statistical-computational gaps, i.e. when the signal (the planted dense subgraph)
is information-theoretically detectable (and hence recoverable by a brute-force search), but is not detectable by efficient algorithms.
In the same context, the low-degree likelihood ratio is used to predict the distinguishing power of low-degree polynomial algorithms.
In \cite{schramm2022computational}, they analyze the low-degree likelihood ratio for certain parameter regimes of Densest $k$-Subraph, but their results do not seem to recover the predictions of the log-density framework precisely.
Our \cref{prop:pE-one} can be interpreted as showing that the low-degree likelihood ratio is $1 + o(1)$ in the entire hard regime for the log-density framework.

\paragraph{Planted Dense Subgraph and Planted Clique conjectures}

In our work, we have focused on the regime $\al \in (0, 1/2], \beta, \gam \in (0, 1)$. Other instantiations of these parameters have also been subject to intense study in recent years and various conjectures predicting the limits of efficient algorithms have been proposed, broadly referred to as the Planted Dense Subgraph conjecture or in the case $\gam = 0$, the Planted Clique conjecture. Furthermore, assuming these conjectures, inapproximability results have been derived for various problems such as Sparse PCA, Stochastic Block model, Biclustering, etc.
See e.g.
\cite{hajek2015computational, chen2014statistical, brennan2018reducibility, brennan2019universality, brennan2019optimal, manurangsi2020strongish, potechin2020machinery, potechin2022subexponential} and references therein. Densest $k$-Subgraph lies at the heart of many of these reductions, therefore it's plausible that our hardness result can be exploited to derive better inapproximability results for various other problems, which we leave for future work.

\subsection{Organization of the paper}

The rest of the paper is organized as follows. We start with a brief overview in \cref{sec:preliminaries} of graph matrices, which are at the heart of our spectral analysis, using it to construct our candidate moment matrix following the pseudo-calibration framework in \cref{sec:pseudocalibration}. With the matrix in hand, we then delve into the extensive PSDness analysis that forms the bulk of work. We
motivate and discuss our conceptually novel PMVS decomposition in \cref{sec:pmvs}, and show the combinatorial analysis for the key ``charging'' arguments of the PSDness proof in \cref{sec:psdness}. We defer the formal details and other technical verifications to appendices.

\paragraph{Acknowledgments}

We thank Madhur Tulsiani for useful discussions. Most of the work for this project was completed while CJ and GR were PhD students at the University of Chicago.

\section{Preliminaries}
\label{sec:preliminaries}

\subsection{The Sum-of-Squares algorithm}\label{subsec: sos}

We now formally describe the Sum-of-Squares hierarchy. For a detailed treatment and survey of SoS, see e.g.
\cite{raghavendra2018high, fleming2019semialgebraic, schramm2017random, hop18, Jones:thesis}.

SoS is used to check feasibility of a system of polynomials.
Given a graph $G = (V, E)$, the simplest polynomial formulation for the existence of a subgraph with $k$ vertices
and $m$ edges encodes the 0/1 indicator of the subgraph:
\begin{align*}
\label{eq:dks-polynomial}
    & \text{Variables: } \xsos_v, \,\,\forall v \in V\\
    & \text{Constraints:}\\
    & \sum_{v \in V}\xsos_v = k &&& (\text{Vertex count})\\
    & \sum_{\{u,v\} \in E} \xsos_u \xsos_v  = m &&& (\text{Edge count})\\
    & \xsos_v^2 = \xsos_v & \forall v \in V & & (\text{Boolean})
\end{align*}

The sum-of-squares algorithm is parameterized by the \emph{degree} $\dsos \in \N$. We assume $\dsos$ is even.
For formal variables $\xsos_1, \dots, \xsos_n$, let $\R^{\leq \dsos}[\xsos_1, \dots, \xsos_n]$
denote the set of polynomials with degree at most $\dsos$.

\begin{definition}[Pseudoexpectation]
    Given a set of variables $\xsos_1, \dots, \xsos_n$, 
    a \emph{degree-$\dsos$ pseudoexpectation
    operator} is a linear functional $\pE : \R^{\leq \dsos}[\xsos_1, \dots, \xsos_n] \to \R$
    such that $\pE[1] = 1$.
\end{definition}

\begin{definition}[Satisfying an equality constraint]
\label{def:satisfying-a-constraint}
    A degree-$\dsos$ pseudoexpectation operator $\pE$ satisfies a polynomial constraint ``$f(\xsos) = 0$''
    if $\pE[f(\xsos)p(\xsos)] = 0$ for all polynomials $p(\xsos)$ such that ${\deg(p) + \deg(f) \leq \dsos}$.
\end{definition}
\begin{definition}[SoS-feasible]
    A degree-$\dsos$ pseudoexpectation operator $\pE$ is \emph{SoS-feasible}
    if for every polynomial $p \in \R^{\leq \dsos/2}[\xsos_1, \dots, \xsos_n]$,
    $\pE[p(\xsos)^2] \geq 0$.
\end{definition}

\begin{definition}[Sum-of-squares algorithm]
Given a system of polynomial constraints $\{f_i(\xsos) = 0\}$
in $n$ variables $\xsos_1, \dots, \xsos_n$,
the degree-$\dsos$ Sum-of-Squares algorithm checks for the existence
of an SoS-feasible degree-$\dsos$ pseudoexpectation operator $\pE$ that satisfies the constraints.
If $\pE$ exists, the algorithm outputs ``may be feasible'', otherwise it outputs ``infeasible''. This can be done algorithmically by solving a semidefinite program
of size $n^{O(\dsos)}$ that searches for a feasible moment matrix~(\cref{def:moment-matrix}).
\end{definition}

If no pseudoexpectation operator exists, then SoS successfully refutes the polynomial system (i.e., it proves that there is no dense subgraph in the input).
On the other hand, if a pseudoexpectation operator exists, SoS cannot rule out that 
the polynomial system is feasible (the pseudoexpectation operator fools SoS, but it may or 
may not correspond to a true distribution on feasible points).
A lower bound against SoS consists of a feasible pseudoexpectation operator in the case when the system is actually infeasible.

\subsection{Moment matrices}

Analysis of the SoS algorithm on an $n$-variable polynomial system is typically accomplished by formulating it in terms of large matrices indexed by subsets of $[n]$,
known as \emph{moment matrices}.

\begin{definition}[Matrix index]
    Let $\calI$ be the set of ordered subsets of $[n]$ of size at most $\dsos / 2$.
\end{definition}

\begin{remark}
Another reasonable definition of $\calI$ uses
subsets of $[n]$ and not ordered subsets.
For technical simplifications, we include an ordering.
\end{remark}

The degree-$\dsos$ sum-of-squares algorithm can be equivalently formulated
in terms of $\R^{\calI \times \calI}$ matrices, which are called \emph{moment matrices}.

\begin{definition}[Moment matrix]\label{def:moment-matrix}
  The moment matrix $\matLam=\matLam(\pE)$ associated to a degree-$\dsos$ pseudoexpectation $\pE$ is an $\calI$-by-$\calI$ matrix defined as
  \[
  \matLam[I, J] \defeq \pE\left[ \xsos^I \cdot \xsos^J \right].
  \]
\end{definition}

\begin{fact}
    $\pE$ is SoS-feasible if and only if $\matLam(\pE) \psdgeq 0$.
\end{fact}

\begin{definition}[SoS-symmetric]
    A matrix $\matLam \in \R^{\calI \times \calI}$ is \emph{SoS-symmetric} if $\matLam[I,J]$
    depends only on the disjoint union $I \sqcup J$ as an unordered multiset.
    Along with the additional constraint $\matLam[\emptyset, \emptyset] = 1$,
    this characterizes $\matLam \in \R^{\calI \times \calI}$ which are moment matrices of degree-$\dsos$ pseudoexpectation operators.
\end{definition}

In the presence of Boolean constraints ``$\xsos_i^2 = \xsos_i$'', a moment matrix satisfies
these constraints if and only if $\matLam[I, J]$ depends only on the union $I \cup J$ as an unordered set (ignore duplicates).

\subsection{\texorpdfstring{$p$}{p}-biased Fourier analysis and graph matrices}

We are interested in matrices which depend on a random graph $G \sim \calG_{n,p}$. To analyze these as functions of $G$, we encode $G$ via its
edge indicator vector in $\{0,1\}^{\binom{n}{2}}$ and perform $p$-biased Fourier analysis. 

\begin{definition}[Fourier character]
$\chi$ denotes the $p$-biased Fourier character,
\begin{align}
    \chi(0) = -\sqrt{\frac{p}{1-p}}, \qquad \chi(1) = \sqrt{\frac{1-p}{p}}\,.
\end{align}
For $H$ a subset or multi-subset of $\binom{[n]}{2}$, let $\chi_H(G) := \prod_{e \in H}\chi(G_e)$.
\end{definition}

\begin{definition}[Ribbon]
    A \emph{ribbon} is a tuple $R = (A_R, B_R, E(R))$ where $A_R,B_R \in \calI$ and
    $E(R) \subseteq \binom{[n]}{2}$.
    The corresponding matrix $\mM_R \in \R^{\calI \times \calI}$ is:
    \begin{align*}
        \mM_R[I,J] =
        \begin{cases}
            \chi_{E(R)}(G) & I = A_R, J = B_R\\
            0 & \text{otherwise}\,.
        \end{cases}
    \end{align*}
\end{definition}

The ribbon matrices $\mM_R$ are mean-zero, orthonormal under the expectation of the Frobenius inner product on matrices,
and form a basis for all $\R^{\calI \times \calI}$-valued functions of $G$.
They are the natural Fourier basis for random matrices that depend on $G$.

In the matrices that we study, the coefficient on a ribbon will not depend on the particular
labels of the ribbon's vertices, but only on the graphical structure of the ribbon.
This graphical structure is called the \emph{shape}.
\begin{definition}[Shape]
	A \emph{shape} $\al$ is an equivalence class of ribbons under relabeling of the vertices (equivalently, permutation by $S_n$). Each shape is associated with a representative graph $(U_\al, V_\al, E(\al))$.
	We let $V(\al) := U_\al \cup V_\al \cup V(E(\al))$.
\end{definition}

For an example of a shape, see \cref{fig: example_shape}.
We use the convention of Greek letters such as $\alpha, \gamma, \tau$ for shapes and Latin letters $R, L, T$ for ribbons.

\begin{definition}[Embedding]
    Given a shape $\al$ and an injective function $\varphi: V(\al) \to [n]$, we let
    $\varphi(\al)$ be the ribbon obtained by labeling $\al$ in the natural way (preserving the order on $U_\al$ and $V_\al$).
\end{definition}

A ribbon $R$ has shape $\al$ if and only if $R$ can be obtained by an embedding of $V(\al)$
into $[n]$. Note that different embeddings may produce the same ribbon.

\begin{definition}[Graph matrix]
    Given a shape $\al$, the graph matrix $\mM_\al$ is
    \[ \mM_\al = \sum_{\text{injective }\varphi: V(\al) \to [n]} \mM_{\varphi(\al)}\,.\]
\end{definition}

The entries of a graph matrix are degree-$\abs{E(\alpha)}$ monomials in the variables $G_e$,
therefore we think of graph matrices as low-degree polynomial random matrices in $G$.
We call them ``nonlinear'' to distinguish them from the degree-1 case, which
is well-studied (being essentially the adjacency matrix of $G$).

\begin{definition}[Trivial]\label{def:trivial}
    A ribbon or shape $\al$ is \emph{trivial} if $E(\al) = \emptyset$.
\end{definition}

\begin{definition}[Diagonal]\label{def:diagonal}
    A ribbon or shape $\alpha$ is \emph{diagonal} if $V(\alpha) = U_\al = V_\al$.
\end{definition}
A diagonal shape is only nonzero on the diagonal entries of the matrix in the block corresponding to $U_\al$.
Note that there are additional shapes which have the same support, namely
shapes which potentially have additional edges and vertices outside of $U_\al = V_\al$.
The diagonal shapes as we have defined them are the most important contributors
to the diagonal entries of the matrix.

\begin{definition}[Transpose]
    The transpose of a ribbon or shape swaps $A_R, B_R$ or $U_\al, V_\al$ respectively.
    This has the effect of transposing the matrix for the ribbon/shape.
\end{definition}

\subsection{Norm bounds}

\begin{definition}[Weight of a set]
For a graph $S$, let $w(S) = |V(S)| - \log_n(1/p) |E(S)|$.
\end{definition}

\begin{definition}[Vertex separator]
    A \emph{vertex separator} of two sets $A, B$ in a graph $G$ is a set $S \subseteq V(G)$
    such that all paths from $A$ to $B$ pass through $S$.
\end{definition}

\begin{definition}[Sparse minimum vertex separator (SMVS)]
    Given a ribbon or shape $\al$, a \emph{sparse minimum vertex separator (SMVS)} is a minimizer of $w(S)$ over $S \subseteq V(\al)$
    which separate $U_\al$ and $V_\al$.
\end{definition}

Observe that up to lower-order factors, $w(S) = \log_n \left(\E [\#\text{ of copies of graph $S$ in }\calG_{n,p}]\right)$.
The SMVS is thus the rarest separator of $\alpha$.

\begin{theorem}[{Norm bound, informal \cite{JPRTX22, rajendran2022concentration}}]
\label{thm:norm-bounds-informal}
With high probability, for all proper shapes $\al$:
\[\norm{\mM_\al} \leq \softO\left(n^{\frac{|V(\al)| - |w(S_{\min})|}{2}}\right)\]
where $S_{\min}$ is the SMVS of $\al$.
\end{theorem}

\subsection{Graph matrix calculus: factoring}

In light of the relevance of vertex separators to the spectrum
of a graph matrix,
a key ingredient underlying our machinery is that each shape admits an (approximate) factorization into three pieces based on its vertex separators. For the following discussion, fix a shape $\al$. The separators of a shape $\al$ have a natural partial order as follows.

\begin{definition}[Left and right]
    A vertex separator $S$ is \emph{left} (respectively \emph{right}) of a vertex separator $S'$ if $S$ separates $U_\al$ and $S'$ (resp. $S'$ and $V_\al$).
\end{definition}

We will define the \emph{leftmost SMVS} to be the SMVS which is left of
all other SMVS for $\alpha$, and similarly for the \emph{rightmost SMVS}.
For an example, see \cref{fig: example_shape}.
We will decompose each shape into three pieces: the ``left shape'' between $U_\al$ and the leftmost SMVS, the ``middle shape'' between the leftmost and rightmost SMVS, and the ``right shape'' between the rightmost SMVS and $V_\al$.
We now work towards making this formal.

Unfortunately, it is not always true that there is a \emph{unique} SMVS that is left of every SMVS.
Nonetheless, we can define the leftmost SMVS in a natural and canonical way
using the following proposition, whose proof is in \cref{app:graph-matrix}.

\begin{figure}
    \centering
    \includegraphics[height=4cm]{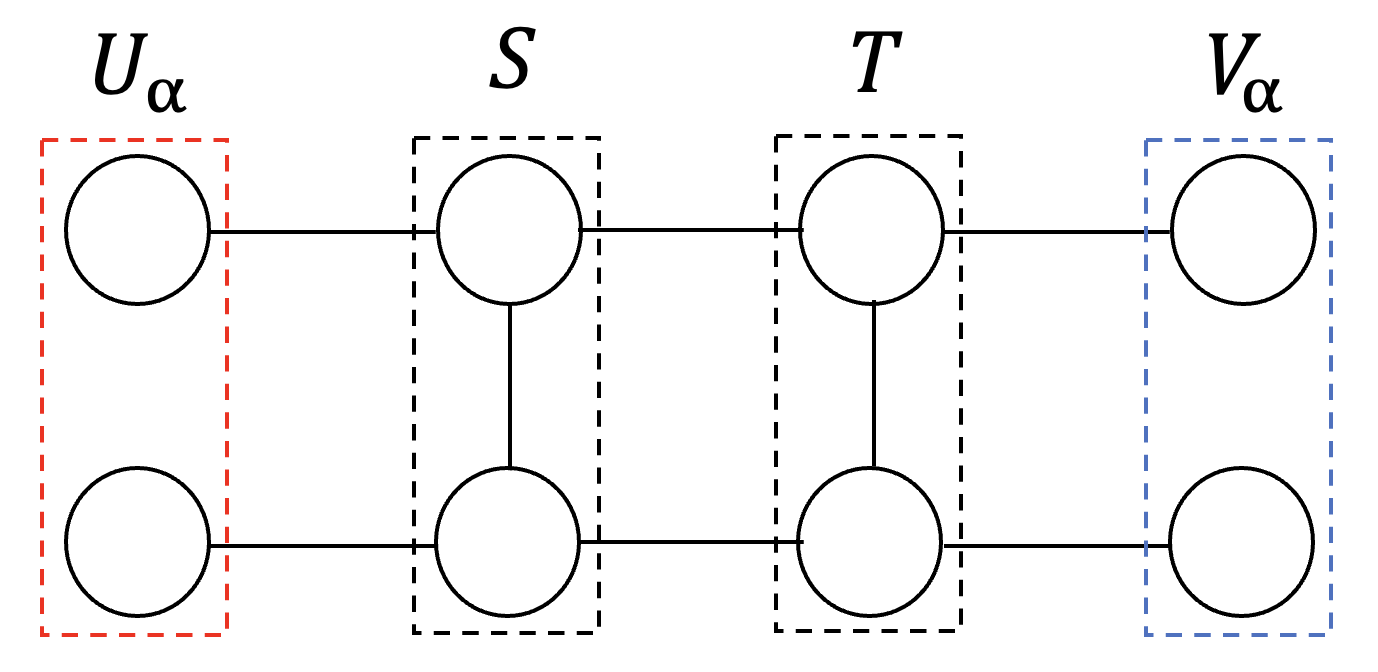}
    \caption{Example of a shape $\al$ with its leftmost SMVS $S$ and rightmost SMVS $T$. This shape has norm $\tilde{O}(n^{\frac{8-2}{2}} \sqrt{\frac{1-p}{p}} ) = \tilde{O}({\frac{n^3}{\sqrt{p}}})$}
    \label{fig: example_shape}
\end{figure}

\begin{restatable}{proposition}{leftmostSmvs}
\label{prop:smvs-uniqueness}
    Every shape has an SMVS which is left of every SMVS.
    Furthermore, there is a unique SMVS left of every SMVS with minimum vertex size.
\end{restatable}

\begin{definition}[Leftmost and rightmost SMVS]
    The \emph{leftmost SMVS} is the SMVS
    which is left of every SMVS and has minimum vertex size.
    The \emph{rightmost SMVS} is defined analogously.
\end{definition}

\begin{definition}[Left shape]\label{def:left-shape}
    A shape $\sigma$ is a \emph{left shape} if the unique SMVS is $V_\sigma$ (hence, it is both leftmost and rightmost).
\end{definition}

\begin{definition}[Middle shape]
    A shape $\tau$ is a \emph{middle shape} if $U_\tau$ is the leftmost SMVS, and $V_\tau$ is the rightmost SMVS.
\end{definition}

\begin{definition}[Right shape]
    A \emph{right shape} $\sigma$ is the transpose of a left shape.
\end{definition}
We also extend these definitions to ribbons.
By splitting a ribbon across its leftmost and rightmost SMVS,
we have the following canonical decomposition theorem for ribbons, to be presented formally in the next section.

\begin{proposition}[informal version of \cref{prop:decomposition-uniqueness}]
\label{prop:unique-factorization}
    Every ribbon $R$ can be expressed uniquely as the composition of a left, middle, and right ribbon.
\end{proposition}

\subsection{Graph matrix calculus: composition}

Multiplying graph matrices can be carried out ``diagramatically'' by \emph{composing} the ribbons or shapes.

\begin{definition}[Composing ribbons]
    Two ribbons $R, S$ are \emph{composable} if $B_{R} = A_{S}$.
    The composition $R \circ S$ is the ribbon $T = (A_R, B_S, E(R) \cup E(S))$. Although it never occurs in the current work, see the footnote\footnote{In this case, define $R \circ S$ as an improper ribbon (\cref{def:improper}) whose edge multiset is the disjoint union of $E(R)$ and $E(S)$.}
    for the case $E(R) \cap E(S) \neq \emptyset$.
\end{definition}

\begin{fact}
\label{fact:ribbon-composition}
If $R, S$ are composable ribbons, then $\mM_{R \circ S} = \mM_R\mM_S$. Otherwise, $\mM_R\mM_S = 0$.
\end{fact}

Therefore, when two matrices expressed as a linear combination of ribbons are multiplied, the effect
is to compose every pair of ribbons.

The composition of two ribbons $R, S$ can be easily visualized by drawing 
the two ribbons next to each other, then identifying the sets $B_R$ and $A_S$.
However, if the vertex sets of two composable ribbons $R, S$ overlap nontrivially (i.e.,
beyond the ``necessary'' overlap $B_R = A_S$), then the resulting ribbon is smaller
than this picture suggests.
We will call these types of ribbons \emph{intersection terms} and classify
them based on their \emph{intersection pattern}.
The intersection terms are error terms in our analysis, but 
carefully bounding them is the most important and difficult conceptual step of the proof.

\begin{definition}[Properly composable]\label{def:properly-composable}
Composable ribbons $R_1, \dots, R_k$ are \emph{properly composable} if there are no intersections
beyond the necessary ones $B_{R_i} = A_{R_{i+1}}$.
\end{definition}

With the above definitions and \cref{prop:smvs-uniqueness}, we can
deduce the main proposition about shape and ribbon factoring.
For a ribbon $R$ and a set of edges $F\subseteq \binom{[n]}{2}$,
we use the notation $R \setminus F$ to denote the ribbon $(A_R, B_R, E(R) \setminus F)$.

\begin{definition}[Floating component]
    The connected components of a ribbon $R$ which are not connected to $A_R \cup B_R$ are called \emph{floating components}.
\end{definition}

\begin{proposition}[Ribbon decomposition]
\label{prop:decomposition-uniqueness}
    Every ribbon $R$ can be expressed as 
    \[R = (L \setminus E(B_L)) \circ M \circ (R' \setminus E(A_{R'}))\]
    where $L, M, R'$ are properly composable left, middle, and right ribbons respectively, such that $E(B_L) = E(A_M)$ and $E(B_M) = E(A_{R'})$.
    Up to the orderings of $B_{L} = A_{M}$ and $B_{M} = A_{R'}$ 
    and the floating components, the decomposition is unique.
\end{proposition}

\begin{proof}
The existence of the decomposition follows by splitting $R$ across
the leftmost and rightmost SMVS. Edges inside the SMVS should be put into the middle ribbon. Any floating components can be put into the middle ribbon.

To argue uniqueness, suppose $R = (L \setminus E(B_L)) \circ M \circ (R' \setminus E(A_{R'}))$. Then $B_L = A_M$ is an SMVS of $R$ which is left of all other SMVS of $R$.
By the structural result \cref{prop:smvs-structure},
in order for $L$ to have a unique SMVS, it must be that $B_L$
is the leftmost SMVS of $R$. The same holds for $R'$ with respect to the rightmost SMVS.
\end{proof}

\begin{remark}
When we decompose a ribbon, we will always put the floating components
into the middle part.
\end{remark}

\subsection{Graph matrix calculus: intersections}

Based on the decomposition theorem for ribbons, it would be ideal if for any shape $\al = \sigma \circ\tau\circ\sigma'$, we also had an exact matrix equality
\[ 
\mM_\al  = \mM_\sigma \cdot \mM_\tau \cdot \mM_{\sigma'}\,.
\]
Unfortunately this fails as we have may ''surprise'' intersections beyond the necessary ones along the boundary. Let us describe in further detail these ``intersection terms''.

\begin{definition}[Composing shapes]
    Given shapes $\al$, $\beta$, we call them \emph{composable} if $V_\al = U_\beta$ as subsets of $\calI$. The \emph{composition} $\al \circ \beta$ is the shape
    whose multigraph is the result of gluing together $\al$ and $\beta$ along $V_\al$ and $U_\beta$ (following their respective orders)\footnote{Although this never occurs in the current work,
    if an edge occurs inside both $V_\al$ and $U_\beta$,
    the composition $\al \circ \beta$ is an improper shape (\cref{def:improper}).},
    whose left side is $U_\al$, and whose right side is $V_\beta$.
    See the footnote\footnote{If a vertex in $V_\alpha$ and $U_\beta$ has degree 0, it becomes an isolated vertex in $\al \circ \beta$ (\cref{def:isolated}).} for an additional technicality.
\end{definition}

\begin{definition}[Intersection pattern]
\label{def:intersection-pattern}
    For composable shapes $\alpha_1, \alpha_2, \dots, \alpha_k$, let $\alpha = \alpha_1 \circ \alpha_2 \circ \cdots \circ \alpha_k$.
    An intersection pattern $P$
    is a partition of $V(\alpha)$ such that for all $i$ and $v, w \in V(\alpha_i)$, $v$ and $w$ are not in the same block of the partition.
    We say that a vertex ``intersects'' if its block has size at least 2 and let $V_{intersected}(\al_i)$ denote the set of intersecting vertices in $\al_i$.
    
    Let $\calP_{\alpha_1, \alpha_2, \dots, \alpha_k}$ be the set of intersection
    patterns between $\alpha_1, \alpha_2, \dots, \alpha_k$.
\end{definition}

\begin{definition}[Intersection shape]
\label{def:intersection-shape}
    For composable shapes $\alpha_1, \alpha_2, \dots, \alpha_k$ and an intersection pattern
    $P~\in~\calP_{\alpha_1, \alpha_2, \dots, \alpha_k}$, let $\alpha_P = \alpha_1 \circ \alpha_2 \circ \cdots \circ \al_k$ then
    identify all vertices in blocks of $P$, i.e. contract them into a single vertex. Keep all edges.\footnote{Keep duplicated edges with multiplicity; $\alpha_P$ may be improper (\cref{def:improper}).}
\end{definition}

\begin{proposition}
\label{mtx:prop:multiply-graph-matrices}
    For composable shapes $\alpha_1, \alpha_2, \dots, \alpha_k$,
    \[\mM_{\alpha_1} \cdots \mM_{\alpha_k} =  \sum_{P \in \calP_{\al_1, \dots, \al_k}} 
    \mM_{\alpha_P}\,.\]
\end{proposition}

\begin{proof}
The claimed statement expands to
\[\prod_{i=1}^k\sum_{\text{injective }\varphi_i :V(\al_i) \to [n]}
\mM_{\varphi_i(\al_i)}
= \sum_{P \in \calP_{\al_1, \dots, \al_k}} 
    \sum_{\text{injective }\varphi: V(\al_P) \to [n]}\mM_{\varphi(\alpha_P)}\,.\]
Each of the ribbons on the left-hand side is an injective embedding of $\al_i$ into $[n]$;
however, the joint embedding need not be injective.
The intersection pattern $P$ cases on which vertices overlap.
\end{proof}

\subsection{Graph matrix calculus: improper shapes and linearization}

We will need to manipulate matrices expressed in the ribbon or shape basis, for example by casing on whether certain edges exist or multiplying two matrices together.
To simplify the intermediate manipulations, we will allow ribbons to be \emph{improper}.
We lose uniqueness of representation in the ribbon basis (there are multiple ways to express
a given matrix as a linear combination of improper ribbons), but the augmentation lets us easily track combinatorial features such
as the presence of specific edges or subgraphs.
At the end of the proof, we convert improper ribbons back into proper ones by \emph{linearizing} them.

Our ribbons may be improper in three ways: \emph{edge indicators}
, \emph{products of Fourier characters}, and \emph{isolated vertices}.
These are all contained in the following general definition.

\begin{definition}[Improper ribbon]
\label{def:improper}
    An \emph{improper ribbon} is a tuple $R$ given by\\ $R = (A_R, B_R, V(R), E(R), \yes(R))$
    where $A_R,B_R \in \calI$ as before, and additionally $A_R \cup B_R \subseteq V(R) \subseteq [n]$, also $E(R)$ is a multigraph on $V(R)$ without self-loops, and $\yes(R)
    \subseteq \binom{V(R)}{2}$.
    We extend the definition of $\mM_R \in \R^{\calI \times \calI}$ to:
    \begin{align*}
        \mM_R[I,J] =
        \begin{cases}
            \displaystyle\prod_{e \in \yes(R)} \one_{e \in E(G)} 
            \prod_{e \in \binom{V(R)}{2}} \chi_e(G)^{\text{multiplicity of $e$ in }E(R)} & I = A_R, J = B_R\\
            0 & \text{otherwise}\,.
        \end{cases}
    \end{align*}
\end{definition}

\begin{definition}[Improper shape]
	An improper shape $\al$ is defined, as before, as an equivalence class of improper ribbons under relabeling of the vertices. Each improper shape is associated with a representative tuple $(U_\al, V_\al, V(\al), E(\al), \yes(\al))$.
	Equivalently, a ribbon $R$ has shape $\al$ if $R$ can be obtained by labeling $V(\al)$ by an injective mapping $\phi: V(\al)\rightarrow [n]$.
\end{definition}

\begin{definition}[Isolated vertices]
\label{def:isolated}
Let $\Iso(\al)$ be the set of vertices in $V(\alpha) \setminus (U_\al \cup V_\al)$
which are not incident to any edges in $E(\al)$ or $\yes(\al)$.
\end{definition}

We linearize an improper ribbon by using identities
\begin{align*}
\chi_e^k &= c_0 + c_1 \cdot \chi_e\\
\one_{e \in E(G)} &= c'_0  + c'_1 \cdot \chi_e\\
\end{align*}
for some coefficients $c_0,c_1,c_0',c_1'$. The equalities hold for inputs from $\{0,1\}$.

\begin{definition}[Linearization]
\label{def:linearization}
    Given an improper ribbon $R$, we \emph{linearize} $R$ by using the equality
    \[ \mM_R = \sum_{\text{proper ribbons }S}c_{S} \mM_{S}\]
    where $c_S$ are the Fourier coefficients of $\mM_R$.
    The ribbons $S$ which appear with nonzero coefficient $c_S$ are called
    the \emph{linearizations} of $R$.
\end{definition}

When we linearize just the multiedges $\chi_e^k$ into either 1 or $\chi_e$, we use the following proposition to bound the new coefficient.
\begin{proposition}
\label{prop:linearization-coefficient}
    Given an improper ribbon $R$, if we linearize the multiedges,
    the coefficient on a resulting ribbon $S$ satisfies \[\abs{c_S} \leq \left(\sqrt{\frac{1-p}{p}} \right)^{\sum_{e\in mul(R)} \mult(e) - 1 - \one_{e \text{ vanishes} }} \]
    where $mul(R)$ is the set of multi-edges in $R$.
\end{proposition}

\begin{proof}
The linearization coefficients are $\chi_e^k = \E[\chi_e^k] + \E[\chi_e^{k+1}]\chi_e$.
By induction, $\abs{\E[\chi_e^k]} \leq \left(\sqrt{\frac{1-p}{p}}\right)^{k-2}$ for all $k \geq 2$, which yields the claim.
\end{proof}

We will estimate the norm of an improper shape by linearizing it and taking the maximum
norm among all of its linearizations.

\subsection{Pseudocalibration}
\label{sec:pseudocalibration}

Pseudocalibration is a heuristic used to construct candidate pseudoexpectation operators $\pE$ for SoS lower bounds, introduced in the context of SoS lower bounds for Planted Clique \cite{BHKKMP16}. 
See e.g. \cite{BHKKMP16, raghavendra2018high, GJJPR20} for a formal description.

The pseudocalibrated operator $\pE[\xsos^I]$ is defined using the Fourier coefficients of the 
corresponding function $\xsos^I(H)$ evaluated on the planted distribution.
First we need to compute these Fourier coefficients. A similar computation was performed by \cite{chlamtavc2018sherali} to exhibit integrality gaps for the Sherali-Adams hierarchy.

\begin{lemma}
Let $\xsos^I(H)$ be the 0/1 indicator function for $I$ being in the planted solution i.e. $I \subseteq H$. Then, for all $I \subseteq [n]$ and $\alpha \subseteq \binom{[n]}{2}$,
\[
\E_{(G, H) \sim \calD_{pl}}[\xsos^I(H) \cdot \chi_\alpha(\widetilde{G})] = \left(\frac{k}{n}\right)^{|V(\al) \cup I|} \left(\frac{q - p}{\sqrt{p(1-p)}}\right)^{|E(\al)|}
\]
\begin{proof}
    First observe that if any vertex of $V(\al) \cup I$ is outside $H$, then the expectation is $0$. This is because either $I$ is outside $H$, in which case $\xsos^I(H) = 0$, or an edge of $\al$ is outside $H$, in which case the expectation of this Fourier character is $0$. Now, each vertex of $V(\al) \cup I$ is in $H$ independently with probability $\frac{k}{n}$. Conditioned on this event happening, each edge independently evaluates to 
    \[
    \E_{e \sim \Bernoulli(q)}\chi(e) = q \cdot \chi(1) + (1-q)\cdot \chi(0) = \frac{q - p}{\sqrt{p(1-p)}}\,.
    \]
    Putting these together gives the result.
\end{proof}
\end{lemma}

Pseudocalibration suggests transferring the low-degree Fourier coefficients 
from the planted distribution, in order for the pseudoexpectation operator to emulate a true expectation operator
from a planted distribution.
As long as the size of the Fourier coefficients is larger than the SoS degree,
then SoS should not notice that we are using a truncation.

We will truncate up to shapes of size $D_V$ for a parameter $D_V = O(\dsos)$ which is formally specified in \cref{assumptions}.

\begin{definition}[$\calS$]
\label{def:calS}
    Let $\calS$ be the set of (proper) ribbons $R$ such that:
    \begin{enumerate}[(i)]
        \item (\textbf{Degree bound}) $\abs{A_R}, \abs{B_R} \leq \dsos/2$
        \item (\textbf{Size bound}) $|V(R)| \leq \dv$
    \end{enumerate}
    We will sometimes use $\al \in \calS$ as the set of shapes with
    the same properties, following the convention of using Latin letters for ribbons and Greek letters for shapes.
\end{definition}

\begin{definition}[$\matM$]
Define the pseudocalibrated candidate moment matrix 
\[\matM = \sum_{R \in \calS}\left(\frac{k}{n}\right)^{\abs{V(R)}} \left(\frac{q-p}{\sqrt{p(1-p)}}\right)^{|E(R)|}\mM_R\]
\end{definition}

For technical convenience, we adjust the parameters $\beta$ and $\gam$ slightly so that
\[n^{-\beta} = \frac{p}{1-p}, \qquad \qquad \qquad n^{-\gam} = \frac{q-p}{1-p}\]
This change does not formally affect the statement of \cref{thm:main}.

For the purposes of analyzing the spectrum of $\matM$ in later sections, it is more convenient to rescale the entries so that $\pE[\xsos^I]$ has order 1 for all $I \subseteq [n]$. This will be the matrix $\matLam$.

\begin{definition}[$\lam_\al$]
    Given a shape or ribbon $\alpha$, let
    \begin{align*}
    \lam_\al &= \left(\frac{k}{n}\right)^{\abs{V(\al)} - \frac{|U_\al| + |V_\al|}{2}} \left(\frac{q-p}{\sqrt{p(1-p)}}\right)^{|E(\al)|}\\
    & = n^{(\alpha - 1)\left(|V(\al)| - \frac{|U_\al| + |V_\al|}{2}\right) + (\frac{\beta}{2} - \gamma)|E(\al)|}\,.
    \end{align*}
\end{definition}

\begin{lemma}\label{lem:coefficients-factor}
    If $R,S$ are properly composable ribbons, then $\lam_{R \circ S} = \lam_R \lam_S$.
\end{lemma}

\begin{definition}[$\matLam$]
Define $\matLam =\sum_{R \in \calS}\lam_R \mM_R$.
\end{definition}

\begin{lemma}
$\matM \psdgeq 0$ if and only if $\matLam \psdgeq 0$.
\end{lemma}
\begin{proof}
We have $\matM = \matD\matLam \matD$ where $\matD$ is a diagonal matrix with positive entries $\matD[I,I] = \left(\frac{k}{n}\right)^{\frac{|I|}{2}}$. Hence $x^\T \matM x \geq 0$ for all $x \in \R^\calI$ if and only if $x^\T \matLam x \geq 0$ for all $x\in \R^\calI$.
\end{proof}

\begin{lemma}
    $\matM$ is SoS-symmetric and satisfies the constraints ``$\xsos_i^2 = \xsos_i$''.
\end{lemma}

\begin{restatable}{proposition}{pseudoEone}
\label{prop:pE-one}
With high probability, we have
$\pE[1] = 1\pm o(1)$.
\end{restatable}
\begin{proof}[Proof (formal version in \cref{sec:formal-pseudocal})]
\begin{align*}
    \abs{\pE[1]-1} &= \big\lvert\sum_{\substack{\al \in \calS:\\ U_\al = V_\al =\emptyset,\\ E(\alpha) \neq \emptyset} }\lam_\al \mM_\al\big\rvert
        \leq \sum_{\substack{\al \in \calS:\\U_\al = V_\al = \emptyset, \\ E(\alpha) \neq \emptyset}} \lam_\al \|{\mM_\al}\|
\end{align*}
Up to a subpolynomial factor that is offset by edge decay, the norm bounds are (to be computed later in \cref{lem:GeneralCoefficientTimesNorm}),
\[ \lam_\al \|M_\al\| \lesssim n^{-(\frac{1}{2}-\alpha)w(\al) - \frac{w(S)}{2} - (\gam - \al \beta)|E(\al)|}\,.\]
With high probability over the random graph $G \sim \calG_{n,p}$, all small subgraphs $H$ have $w(H) \geq 0$ (up to a small error term).
Therefore, a conditioning argument will get rid of any small shape such with $w(\al) < 0$ or $w(S) < 0$.
The remaining shapes have 
both $w(\al)\geq 0$ and $w(S) \geq 0$.
Since the parameters satisfy $\frac{1}{2} - \al \geq 0$ and $\gam - \al \beta > 0$, the exponent of $n$ is negative
and the entire sum is $o(1)$.
\end{proof}

\begin{remark}
    $\pE[1] = 1 \pm o(1)$ fails with inverse polynomial probability.
    This is equivalent to the observation that there is a low-degree distinguisher
    that succeeds with inverse polynomial probability.
    Specifically, if we consider a constant-size subgraph which is unlikely to appear in $\calG_{n,p}$ (e.g. $K_{10}$ when $p$ is sufficiently small), the probability that the dense subgraph $\calG_{k,q}$ contains
    a copy of the subgraph is larger by a $\poly(n)$ factor. That said, even when $\pE[1] \gg 1$, we can still show that $\matLam \succeq 0$ with high probability. For details, see the appendix.
\end{remark}
\begin{restatable}{proposition}{pseudoEXv}
\label{prop:approximate-constraints}
With high probability, \[ 
\left|\sum_{i=1}^n\pE[\xsos_i] - k\right| = o(k) ,
\]
and \[ 
\left|\sum_{\{i,j\} \in E(G)}\pE[\xsos_i\xsos_j] - \frac{k^2q}{2}\right| = o({k^2}q)
\]
\end{restatable}

\begin{proof}[Proof (formal version in \cref{sec:formal-pseudocal})]
For each $i$, $\pE[\xsos_i] = \mM[(i),\emptyset]$. For most vertices $i$, the dominant term in $\mM[(i),\emptyset]$ is $\frac{k}{n}\mM_R[(i),\emptyset]$ where $R$ is the ribbon such that $V(R) = \{i\}$, $A_R = (i)$, $B_R = \emptyset$, and $E(R) = \emptyset$. $\mM_R[(i),\emptyset] = 1$ so this gives a contribution of $\frac{k}{n}$. Summing this over all $i \in [n]$ gives $k$. In \cref{sec:formal-pseudocal}, we verify that the contribution from the other terms is $o(k)$ with high probability.

For each $i,j \in V(G)$ such that $i < j$, 
$\one_{\{i,j\} \in E(G)}\pE[{\xsos_i}{\xsos_j}] = \one_{\{i,j\} \in E(G)}\mM[(i),(j)]$. For most $i,j$, the dominant term in $\one_{\{i,j\} \in E(G)}\mM[(i),(j)]$ is $\frac{{k^2}(q-p)}{{n^2}\sqrt{p(1-p)}}\one_{\{i,j\} \in E(G)}\mM_R[(i),(j)]$ where $R$ is the ribbon such that $V(R) = \{i,j\}$, $A_R = (i)$, $B_R = (j)$, and $E(R) = \{\{i,j\}\}$.

$\one_{\{i,j\} \in E(G)}\mM_R[(i),(j)] = \sqrt{\frac{1-p}{p}}$ if $\{i,j\} \in E(G)$ and $0$ if $\{i,j\} \notin E(G)$. Summing over all $i < j$ gives a total contribution of $\frac{{k^2}(q-p)}{{n^2}p}|E(G)| \approx \frac{{k^2}q}{2}$. In \cref{sec:formal-pseudocal}, we verify that the contribution from the other terms is $o({k^2}q)$ with high probability.
\end{proof}

\section{Positive Minimum Vertex Separator Decomposition}
\label{sec:pmvs}

\subsection{Motivation for the positive minimum vertex separator}
\label{sec:proof-outline}

After pseudocalibration, to complete the proof of \cref{thm:main}, 
we need to show that the rescaled candidate moment matrix is PSD with high probability, \[
\matLam  = \sum_{\al \in \calS}  \lambda_\al\cdot  \mM_\al \succeq 0\,.
\]
 
For each graph matrix $\lam_\al \mM_\al$ in $\matLam$, we want to find an approximately-PSD term
which spectrally dominates it.
Previous work led to the following idea:
for each shape $\alpha$, we can split it across the leftmost and rightmost minimum vertex separators so that $\alpha$ is decomposed into three parts, 
\[ 
    \alpha= \sigma \circ \tau \circ \sigma'^\T\,.
\]
Then the target spectral upper bound is given by 
\[ \lam_\sigma^2 \mM_{\sigma \circ \sigma^\T} + \lam_{\sigma'}^2 \mM_{\sigma' \circ\sigma'^\T}\,.\]
This is approximately PSD since $\mM_{\sigma \circ \sigma^\T} \approx \mM_\sigma \mM_\sigma^\T \psdgeq 0$. To make this strategy work, we need to prove that the middle shape $\mM_\tau$ is spectrally dominated by the corresponding identity via combinatorial charging. In previous work, it
has been essentially possible to charge all middle shapes to the identity matrix,
but this breaks down in the setting of Densest $k$-Subgraph. 
In the baby case, this is evident in our calculation for $\sum_{(u,v) \in E(G)}\pE[\xsos_u\xsos_v]$ in \cref{prop:approximate-constraints}, where the dominant term is no longer the trivial shape but instead the shape with an edge in between $i$ and $j$.

A second, related issue is the presence of edges inside the separator.
Concretely, say that 
$(U_\tau, E(U_\tau))$ and $(V_\tau, E(V_\tau))$ are the leftmost/rightmost SMVS of a middle shape $\tau$, and we hope to charge $\tau$ to the diagonal matrix corresponding to the leftmost/rightmost SMVS. Concretely, letting $U_\tau$ also
denote the diagonal shape with edges $E(U_\tau)$, we want to charge
\[ \lam_\tau (\mM_\tau + \mM_\tau^\T) \psdleq \lam_{U_\tau}\mM_{U_\tau} + \lam_{V_\tau}\mM_{V_\tau}\,.\]

However, this strategy crucially requires that $\lambda_{U_\tau} \cdot \mM_{U_\tau}$ and $\lambda_{V_\tau}\cdot \mM_{V_\tau}$ are PSD by themselves in order to conclude that the result is PSD.
Since $\lambda_\alpha$ is non-negative, this boils down to the PSD-ness of the diagonal shape $(U_{\tau}, E(U_{\tau}))$ for the SMVS. This latter matrix is easy to verify as the non-zero diagonal entries are given by, for a ribbon $R$ of the corresponding shape $U_\tau$, \[
 \chi_{E(R)}(G) =\prod_{e\in E(R)} \chi_e(G)
\]
and recall that we are working on the $p$-biased Fourier basis,\begin{align*}
\chi_e(1) = \sqrt{\frac{1-p}{p}} ,  \;\;\:\:\: \chi_e(0) = -\sqrt{\frac{p}{1-p }}
\end{align*}
At this point, we observe that the instantiation of the SMVS edges $E(R)$ plays a crucial role as they determine whether our candidate ''PSD'' mass is truly positive. 
If all edges of $E(R)$ are present in $G$, then the diagonal entry is positive,
\[\prod_{e\in E(R)} \chi_e(G) = \sqrt{\frac{1-p}{p}}^{\abs{E(R)}} \geq 0\,. \]
On the other hand, if an edge is missing, then positivity is not guaranteed.
Ignoring this bad case for now, we have the following sufficient criterion for finding
a PSD dominant term.
If $T$ is a ribbon of shape $\tau$, and $R$ is the restricted ribbon to $U_\tau$, then if $E(R) \subseteq E(G)$,
we must charge $\lam_\tau \mM_T$ to $\lam_{U_\tau}\mM_R$.

When an edge is missing inside the SMVS, then we need to look harder. 
Despite the candidate PSD term not being truly positive, it is not yet time to panic. In this case,
(1) a missing edge scales down the matrix, in line with the intuition that subgraphs with edges present are the highest-norm terms, therefore (2) we look in the remainder of the shape for the new SMVS, to determine the new matrix norm.
This creates a recursive process, and when all edges inside the candidate SMVS are
actually present in the graph, we terminate, calling this the \emph{Positive Minimum-weight Vertex Separator (PMVS).}

\begin{figure}[!htbp]
    \centering
    \includegraphics[width=0.5\textwidth]{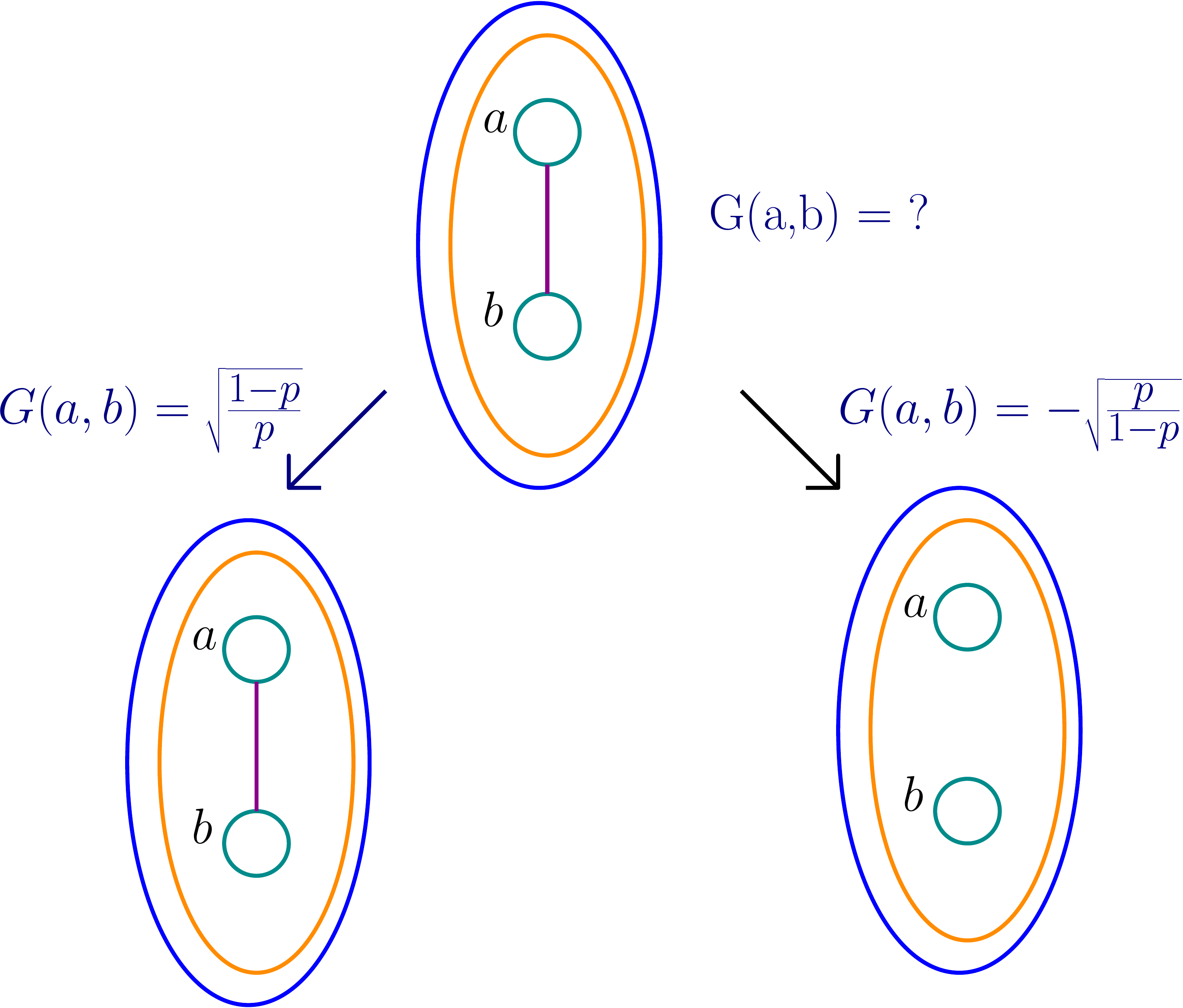}
    \caption{PMVS Search}
    \label{fig:eg}
\end{figure}

Let us give an example. The graph matrix at the top of  \cref{fig:eg} appears on the diagonal of our moment matrix. In this example shape, the only vertex separator is the entire shape, and so the SMVS contains the edge $G(a,b)$. We check whether or not the edge appears in the graph.
In the ``yes'' outcome on the left, we have a PSD matrix whose $(a,b)$-th diagonal entry is $\one_{(a,b) \in E(G)} \sqrt{\frac{1-p}{p}}$. In the ``no'' outcome on the right, the $(a,b)$-th diagonal entry is $-\one_{(a,b) \not\in E(G)} \sqrt{\frac{p}{1-p}}$, which is negative and therefore the matrix is not PSD.
This matrix comes with a small coefficient of approximately $\sqrt{p}$ and hence it can be charged to the corresponding identity matrix, whose $(a,b)$-th diagonal entry is just $1$.
In this example, the recursion terminates after just one step, but in larger shapes, we would need to find the new SMVS for the case on the right.

It would have been cleaner if one can define the PMVS in ``one shot'',
rather than through a recursion. As described above, the recursion outputs the minimizer of the weight function defined in \cref{app:true-pmvs}. However, we need to slightly modify the recursive process described above so that it always ``moves left'', in order for the crucial \cref{rmk:left-oblivious} to hold.

\subsection{PMVS subroutine}\label{sec:ribbonoperations}

We make the following extended definition of the \emph{Positive Minimum Vertex Separator (PMVS)} of a ribbon $R$.

\begin{definition}[Left and right indicators]
We say that a ribbon $R$ has \emph{left indicators} if $R$ has edge indicators for every edge $e \in E(A_R)$.
Similarly, we say that a ribbon $R$ has \emph{right indicators} if $R$ has edge indicators for every edge $e \in E(B_R)$.
\end{definition}

Our goal is to have composable triples of ribbons $R_1,R_2,R_3$ with the following properties:
\begin{definition}[Ribbons with PMVS identified]
\label{def:pmvs}
A composable triple of ribbons $R_1, R_2, R_3$ has PMVS identified if:
\begin{enumerate}[(i)]
\item $R_1$ is a left ribbon and $R_3$ is a right ribbon.
\item $R_1,R_2,R_3$ are properly composable.
\item $R_1$ has right indicators, $R_2$ has both left and right indicators, and $R_3$ has left indicators.
\item The edges and edge indicators agree on $B_{R_1} = A_{R_2}$ and $B_{R_2} = A_{R_3}$.
\item $R_1, R_2, R_3$ have no other edge indicators.
\end{enumerate}
When these properties hold, we say that the left PMVS is $A_{R_2}$ and the right PMVS is $B_{R_2}$.
\end{definition}

\begin{remark}
The left and right PMVS may not have the same size or weight.
In fact, they may not even be an SMVS of $R_2$.
We will bound the difference between the PMVS and the SMVS in \cref{sec:psdness}.
\end{remark}

At the beginning, we take each ribbon $R$ and decompose it into ribbons $R_1,R_2,R_3$ based on the leftmost and rightmost SMVS using \cref{prop:decomposition-uniqueness}. This gives us a composable triple of ribbons $R_1,R_2,R_3$ such that 
\begin{enumerate}[(i)]
\item $R_1$ is a left ribbon, $R_2$ is a middle ribbon, and $R_3$ is a right ribbon.
\item $R_1,R_2,R_3$ are properly composable.
\setcounter{enumi}{3}
\item The edges agree on $B_{R_1} = A_{R_2}$ and $B_{R_2} = A_{R_3}$.
\item $R_1$, $R_2$, and $R_3$ have no edge indicators.
\end{enumerate}

\begin{remark}
    The ribbon encoded by the triple $R_1, R_2, R_3$ is $(R_1 \setminus E(B_{R_1})) \circ R_2 \circ (R_3 \setminus E(A_{R_3}))$ rather than $R_1 \circ R_2 \circ R_3$
    because edges inside $B_{R_1} = A_{R_2}$ should not be duplicated.
\end{remark}

In order to satisfy the condition that $R_1$ has right indicators, $R_2$ has both left and right indicators, and $R_3$ has left indicators, we repeat the following sequence of operations as many times as needed.

\begin{enumerate}
\item \textbf{Adding left and right indicators operation:} To add indicators to $B_{R_1} = A_{R_2}$ and $B_{R_2} = A_{R_3}$, we replace each edge $e \in E(A_{R_2}) \cup E(B_{R_2})$ that does not yet have an indicator using\footnote{The high-level overview of the PMVS alluded to the slightly different formula $\chi_e = \one_{e \in E(G)}\chi_e + \one_{e \notin E(G)}\chi_e$. These are morally equivalent, but the formula here is simpler to analyze.} the equation $\chi_e = \frac{1}{1-p}\one_{e \in E(G)}\chi_e - \sqrt{\frac{p}{1-p}}$.
This leads to two possible new ribbons which have different edge structure, one with $e$ still present and the other with $e$ removed.

\item \textbf{PMVS operation:} After adding the edge indicators to $B_{R_1} = A_{R_2}$ and $B_{R_2} = A_{R_3}$, we check if $R_1$ is still a left ribbon and $R_3$ is still a right ribbon. If so, we stop and exit the loop. If not, we let $A'$ be the leftmost SMVS separating $A_{R_1}$ from $B_{R_1}$ and we let $B'$ be the rightmost SMVS separating $A_{R_3}$ from $B_{R_3}$. We then replace $R_1$, $R_2$, and $R_3$ with the ribbons $R'_1$, $R'_2$, and $R'_3$ where 
\begin{enumerate}
\item $R'_1$ is the part of $R_1$ between $A_{R_1}$ and $A'$.
\item $R'_2$ is the composition of the part of $R_1 \setminus E(B_{R_1})$ between $A'$ and $B_{R_1}$, $R_2$, and the part of $R_3 \setminus E(A_{R_3})$ between $A_{R_3}$ and $B'$.
\item $R'_3$ is the part of $R_3$ between $B'$ and $B_{R_3}$.
\end{enumerate}
\item \textbf{Removing middle edge indicators operation:} If $R_2$ has one or more edge indicators which are now outside of $A_{R_2}$ and $B_{R_2}$, we re-convert them back into Fourier characters using the equation $\frac{1}{1-p}\one_{e \in E(G)}\chi_{e} = \sqrt{\frac{p}{1-p}} + \chi_{e}$.
\end{enumerate}

We call this repeated sequence of operations the \textbf{Finding PMVS subroutine}, which takes a triple of composable ribbons $R_1,R_2,R_3$ which have all the needed properties except having left and right indicators (some but not all indicators may be present) and gives us a triple of composable ribbons with all of the needed properties.

\begin{remark}
    Note that each triple $R_1, R_2, R_3$ leads to many
    triples $R'_1, R'_2, R'_3$ depending on which summand is taken in each equation.
    The recursion proceeds on every term except for the one in which every $\chi_e$ is replaced by $\frac{1}{1-p}\one_{e\in E(G)}\chi_e$.
\end{remark}

\begin{remark}\label{rmk:graph-dependence}
At first glance, checking whether or not edges inside $A_{R_2}$ and $B_{R_2}$ are present leads to a complicated dependence on the input graph $G$.
In order to mathematically express the recursion in a $G$-independent way,
we formally use the edge indicator function to express the two cases.
\end{remark}

\subsection{Intersection term operation}
\label{sec:intersectionoperation}

Once we have these triples of ribbons $R_1,R_2,R_3$, we can apply an approximate factorization across the PMVS. When we do this, we will obtain error terms which can be described by triples of ribbons $R_1$, $R_2$, and $R_3$ which have at least one non-trivial intersection (they are not properly composable) but satisfy the other four properties in \cref{def:pmvs}. We handle this as follows.
\begin{enumerate}
\item \textbf{Intersection term decomposition operation:} Let $A'$ be the leftmost SMVS between $A_{R_1}$ and $B_{R_1} \cup V_{intersected}(R_1)$ and let $B'$ be the rightmost SMVS between $A_{R_3} \cup V_{intersected}(R_3)$ and $B_{R_3}$.
We now replace $R_1$, $R_2$, and $R_3$ with the ribbons $R'_1$, $R'_2$, and $R'_3$ where 
\begin{enumerate}
\item $R'_1$ is the part of $R_1$ between $A_{R_1}$ and $A'$.
\item To obtain $R'_2$, we improperly compose the part of $R_1 \setminus E(B_{R_1})$ between $A'$ and $B_{R_1}$, $R_2$, and the part of $R_3 \setminus E(A_{R_3})$ between $A_{R_3}$ and $B'$. We then linearize the multi-edges, replacing $\chi_e^k = c_0 + c_1 \chi_e$ using the appropriate coefficients $c_0, c_1$.

\smallskip

In the edge case that a multi-edge also has an edge indicator (for example, because an each inside $A_{R_2}$ intersects with an edge from $R_3$), we instead use the equation $\one_e \chi_e^k = \left(\sqrt{\frac{1-p}{p}}\right)^{k-1} \one_e \chi_e$.

\item $R'_3$ is the part of $R_3$ between $B'$ and $B_{R_3}$.
\end{enumerate}
\item We apply the \textbf{Removing middle edge indicators operation} to $R_2$.
\end{enumerate}
The ribbon $R'_2$ is
defined to ``grow'' $R_2$ so that it includes the intersections.
After these steps, we are in essentially the same situation as we started. More precisely, we have a triple of ribbons $R_1,R_2,R_3$ such that 
\begin{enumerate}[(i)]
\item $R_1$ is a left ribbon and $R_3$ is a right ribbon.
\item $R_1,R_2,R_3$ are properly composable.
\setcounter{enumi}{3}
\item The edges and edge indicators agree
on $B_{R_1} = A_{R_2}$ and $B_{R_2} = A_{R_1}$.
\item $R_1, R_2, R_3$ have no edge indicators outside of $B_{R_1} = A_{R_2}$ and $B_{R_2} = A_{R_3}$.
\end{enumerate}

At this point, we can repeat the operations, applying the \textbf{Finding PMVS subroutine} to identify a new PMVS, approximately factoring, then decomposing intersection terms, as many times as needed.

\subsection{Summary of the operations and overall decomposition}
\label{sec:psd-decomposition-informal}

We now summarize our procedure. 

\textbf{Finding PMVS subroutine:} repeat the following until
convergence,
\begin{enumerate}
\item Apply the \textbf{Adding left and right indicators operation} to add indicators to $B_{R_1} = A_{R_2}$ and $B_{R_2} = A_{R_3}$.
\item Apply the \textbf{PMVS operation} to ensure that $R_1$ is a left ribbon and $R_3$ is a right ribbon. If no change is made to $R_1$ or $R_3$, then we have identified the PMVS.
\item Apply the \textbf{Removing middle edge indicators operation} to $R_2$ to ensure that $R_2$ has no middle indicators.
\end{enumerate}
Overall decomposition procedure:
\begin{enumerate}
\item We start with triples of composable ribbons $R_1,R_2,R_3$ which have all the needed properties except having left and right indicators.
\item We apply the \textbf{Finding PMVS subroutine}.
\item Recursive factorization: We apply the following procedure repeatedly until there are no more error terms.
\begin{enumerate}[1.]
\item We approximate the sum over the composable triples of ribbons $R_1, R_2, R_3$ by enlarging the sum to include all left ribbons $R_1$ and right ribbons $R_3$
(not necessarily properly composable with $R_2$ or with each other).
This yields a matrix $\matL\matQ_i\matL^\T$ where $\matL$ sums over left ribbons and $\matQ_i$ sums over the ribbons $R_2$ on the $i$th iteration of the loop.
We then move to the triples of ribbons $R_1,R_2,R_3$ for the intersection error terms, if any.\footnote{There are also additional error terms for the truncation error, as the maximum size of the left ribbons $R_1, R_3$ will be slightly smaller for intersection terms. This must be handled separately.}
\footnote{We enlarge the sum to include only ribbons $R_1$ such that
$B_{R_1}= A_{R_2}$ and $R_3$ such that $A_{R_3} = B_{R_2}$. For this reason, the matrix $\matL$ is slightly more restricted than including all left ribbons.}
\item We apply the \textbf{Intersection term decomposition operation} to obtain a triple of ribbons $R_1,R_2,R_3$ which are properly composable.
\item We apply the \textbf{Removing middle edge indicators operation} to $R_2$ to ensure that $R_2$ has no middle indicators.
\item We apply the \textbf{Finding PMVS subroutine}.
\end{enumerate}
\end{enumerate}
\begin{remark}\label{rmk:left-oblivious}
As with previous SoS lower bounds using graph matrices, a key observation is that the PMVS operation and the intersection term decomposition operation are unaffected by replacing $R'_1$ with a different left ribbon $R''_1$ or replacing $R'_3$ with a different right ribbon $R''_3$ as long as $B_{R''_1} = B_{R_1'} = A_{R'_2}$ and $A_{R''_3} = A_{R_3'} = B_{R'_2}$.
This ensures that all left ribbons $R_1'$ and right ribbons $R_3'$ appear in the matrices $\matL$ and $\matL^\T$.
\end{remark}

Carrying out this process, the overall decomposition of the moment matrix is then
\begin{align*}
    \matLam &= \matL\left(\sum_{i = 0}^{D_V} \matQ_i\right)\matL^\T
     \pm \text{truncation error}\,.
\end{align*}

Therefore, the main requirement for $\matLam \psdgeq 0$ is to show that $\sum_{i = 0}^{D_V}\matQ_i \psdgeq 0$.
We will show that the norm-dominant terms are the diagonal shapes (\cref{def:diagonal}). By virtue of the PMVS factorization,
these shapes are PSD, as we can easily check.

\begin{lemma}
    If $R_1, R_2, R_3$ are ribbons with PMVS identified, such that $R_2$ is diagonal, then $\lam_{R_2}\mM_{R_2} \psdgeq 0$.
\end{lemma}
\begin{proof}
    $\lam_{R_2} \geq 0$ and $R_2$ is diagonal with one nonzero entry, so we need that the entry is nonnegative.
    Since $R_2$ has edge indicators, the entry is
    \[\prod_{e \in E(R_2)} \one_{e \in E(G)} \chi_e(G)\,.\]
    Any time the entry is nonzero, its value is $\chi(1)^{|E(R_2)|} = \left(\sqrt{\frac{1-p}{p}}\right)^{|E(R_2)|}\geq 0$.
\end{proof}
In the next section, we will prove that the norm of any individual term in the $\matQ_i$ is small
relative to these PSD terms, which is the key remaining component of the PSDness proof.
Summing over all the terms and bounding the truncation error is done in \cref{sec:formal-details}.

So far, we have described how the ribbons are manipulated.
Each ribbon also comes with a coefficient that we need to track.
Initially, the coefficient of every ribbon $R$ is $\lam_R$.
Since the coefficients satisfy $\lam_{R \circ S} = \lam_R \lam_S$ (\cref{lem:coefficients-factor}),
we may factor $\lam_R$ whenever we factor the ribbon.
Doing so, the left and right ribbons $R_1$ and $R_3$ always come with the factors $\lam_{R_1}$ and $\lam_{R_3}$.

The coefficient on $R_2$ is initially $\lam_{R_2}$, but
it accrues extra factors in some steps of the process.

\begin{definition}[$c_R$]
\label{def:coefficient}
Given ribbons $R_1, R_2, R_3$ which produce ribbons $R_1', R_2', R_3'$ let $c_{R_2'}$ be such that the final coefficient is $c_{R_2'}\lam_{R'_1}\lam_{R'_2}\lam_{R'_3}$. Note that $c_{R'_2}$ does not depend on $R'_1$ or $R'_3$, but it does depend on
the cases in the decomposition process.
\end{definition}

For example, during the \textbf{Adding left
and right indicators operation}, $c_{R'_2}$ accrues a factor of $\frac{1}{1-p}$ for each edge indicator or a factor $-\sqrt{\frac{p}{1-p}}$
if an edge is removed.
It changes in the same way during the \textbf{Removing middle edge indicators operation}. It will also be multiplied 
by the excess edge and vertex factors during the \textbf{Intersection term decomposition operation}, as well
as the linearization coefficient either $c_0$ or $c_1$.

\section{Combinatorial Norm Charging Arguments}
\label{sec:psdness}

Now that we have identified the dominant PSD terms, we analyze the norms of the non-dominant terms that appear during the decomposition process
in \cref{sec:psd-decomposition-informal} and show that they are small.

Each graph matrix making up $\matLam$ has norm $\poly(n)$ times additional log factors.
It is most important to perform the proof at the coarse level of $\poly(n)$, ignoring the log factors
and other relatively small\footnote{In order to study the regime where the random graph has subpolynomial average degree, or improve the SoS degree above $n^\delta$, we would need to carefully track log factors and $\dsos$ respectively.
} combinatorial factors such as $\poly(\dsos)$.
In this section, we will work at the coarse level by defining away all of the subpolynomial factors in order to focus on the key
combinatorial arguments.
Only the first subsection will involve $n$, and the remaining
subsections will be pure combinatorics on shapes.
The lower-order factors will be formally incorporated in \cref{sec:formal-details}.

\subsection{Setup}
In this section we will use the parameterization $\al, \beta, \gam$ instead of $k, p, q$.

\begin{remark}
Take note that there is a notation clash between the size of the dense subgraph $\al \in (0,1/2)$ and a generic shape $\al$, and also between the edge density of the planted subgraph $\gam \in (0,1)$ and a left shape $\gam$ which participates in an intersection.
It should be clear from context whether the symbol refers to a shape or a real number.
\end{remark}

Let $\tau$ be the shape of a ribbon $R_2$ that appears in a matrix $\matQ_i$ in \cref{sec:psd-decomposition-informal}.
The contribution of this term to $\matQ_i$ is $c_\tau \lam_\tau \mM_\tau$.
We wish to show that when $\tau$ is not a diagonal shape, this
expression has small norm.

Recall that the pseudocalibrated coefficients are
\[\lambda_{\tau} =  n^{(\alpha - 1)\left(|V(\tau)| - \frac{|U_\tau| + |V_\tau|}{2}\right)  + (\frac{\beta}{2} - \gamma)|E(\tau)|}\,.\]

Recall the weight function $w(S) = |S| - \beta \cdot |E(S)|$.
\begin{definition}[Approximate norm bound]
Given a shape $\alpha$ (possibly improper), let:
\begin{align*}
    \normapx{\al} &= n^{\frac{|V(\al)| - w(S_{min}) + |\Iso(\al)|}{2}}\,.
\end{align*}
\end{definition}
It would be more proper to write $\norm{\mM_\al}^\approx$ although we use this version for more compact notation.
\begin{definition}[Approximate coefficient change, informal]
\label{def:approx-coefficient}
Given a shape $\tau$, let $c_\tau^\approx = \abs{c_\tau}$ when ignoring subpolynomial factors.
\end{definition}

For the shape $\tau$,
we view $U_\tau$ and $V_\tau$ also as diagonal shapes which include
only edges with both endpoints inside $U_\tau$ or $V_\tau$.
We want to bound $c_{\tau}\lambda_{\tau}\mM_{\tau}$ using the diagonal shapes $\lambda_{{U_{\tau}}}\mM_{U_{\tau}}$ and $\lambda_{{V_{\tau}}}\mM_{V_{\tau}}$.
In order to do this, we need to have that
\[
c^\approx_{\tau}\lambda_{\tau}\normapx{\tau} \leq \sqrt{\lambda_{{U_{\tau}}}\lambda_{{V_{\tau}}}\normapx{U_\tau}\normapx{V_\tau}}\,.
\]
It turns out that this inequality will hold with a $\poly(n)$ factor of slack,
which furthermore increases for larger shapes $\tau$. We will use this extra slack
to control the subpolynomial factors in the formal analysis.
To keep track of this extra slack, we define the following slack parameter.
\begin{definition}[Slack]
Given a shape $\tau$ with a coefficient $c_{\tau}$, we define $\slack(\tau)$ so that $c^\approx_{\tau}\lambda_{\tau}\normapx{\tau} = n^{-\slack(\tau)}\sqrt{\lambda_{{U_{\tau}}}\lambda_{{V_{\tau}}} \normapx{U_\tau}\normapx{V_\tau}}$.
\end{definition}

By construction, $\slack(\tau)$ is a combinatorial quantity that does not depend on $n$.
It is crucial to prove that $\slack(\tau) \geq 0$, and in the remaining subsections, we will prove the following positive lower bound on $\slack(\tau)$, by proving combinatorially that the slack increases during each operation of the recursion.
\begin{definition}[$V_{tot}(\tau)$ and $E_{tot}(\tau)$]
    During the recursion, some vertices and edges are lost during intersections, or when adding or removing indicators.
    Given a shape $\tau$ of a ribbon during the recursion,
    let $V_{tot}(\tau)$ be the vertex set without performing the intersections. Let $E_{tot}(\tau)$ be the enlargement of $E(\tau)$ to include all of the removed edges.
\end{definition}

\begin{remark}
    Identify $(U_\tau, V_\tau, V_{tot}(\tau), E_{tot}(\tau))$ with the shape $\gam_j \circ \cdots \circ \gam_1 \circ \tau_{0} \circ \gam'^\T_1 \circ \cdots \circ\gam'^\T_j$ where $\tau_0$ is the ``initial'' middle shape, $j$ is the level of the recursion, and $\gam_i, \gam'^\T_i$ will be described in the following sections.
\end{remark}

\begin{restatable}{theorem}{slackTheorem}\textnormal{(Slack lower bound).}
\label{thm:slack}
    At all times in the decomposition procedure described in \cref{sec:psd-decomposition-informal}, letting $\tau$ be the shape of $R_2$,
    \[\slack(\tau) \geq \epsilon\left(|E_{tot}(\tau)| - \frac{|E(U_{\tau})| + |E(V_{\tau})|}{2} + |V_{tot}(\tau)| - \frac{|U_{\tau}| + |V_{\tau}|}{2}\right)\]
    where $\eps = \min\left\{1-\al, \frac{\gam - \al \beta}{8}\right\}$.
\end{restatable}

We develop a combinatorial formula for the slack in the next few lemmas.

\begin{lemma}\label{lem:GeneralCoefficientTimesNorm}
For any shape $\tau$, if $S$ is an SMVS of $\tau$ then
\[
\lambda_{\tau}\normapx{\tau} = n^{(1 -\al)\left(\frac{|U_\tau| + |V_\tau|}{2}\right)-(\frac{1}{2} - \alpha)w(\tau) - \frac{w(S)}{2} + \frac{|\Iso(\tau)|}{2}  - (\gamma - {\alpha}{\beta})|E(\tau)|}
\]
\end{lemma}
\begin{proof}
\begin{align*}
\lambda_{\tau}\normapx{\tau} &= n^{(\alpha-1)\left(|V(\tau)| - \frac{|U_\tau| + |V_\tau|}{2}\right) + (\frac{\beta}{2} - \gamma)|E(\tau)|}n^{\frac{|V(\tau)| - w(S) + |\Iso(\tau)|}{2}} \\
&= n^{(1 - \al)\left(\frac{|U_\tau| + |V_\tau|}{2}\right) - (\frac{1}{2} - \alpha)|V(\tau)| + (\frac{1}{2} - {\alpha}){\beta}|E(\tau)| - (\gamma - {\alpha}{\beta})|E(\tau)| - \frac{w(S)}{2} + \frac{|\Iso(\tau)|}{2}} \\
&= n^{(1 -\al)\left(\frac{|U_\tau| + |V_\tau|}{2}\right) -(\frac{1}{2} - \alpha)w(\tau) - (\gamma - {\alpha}{\beta})|E(\tau)| - \frac{w(S)}{2} + \frac{|\Iso(\tau)|}{2}}
\end{align*}
\end{proof}

\begin{lemma}\label{lem:norm-ratio}
For any shape $\tau$, if $S$ is an SMVS of $\tau$ then
\begin{align*}
&\frac{\lam_\tau \normapx{\tau}}{\sqrt{\lambda_{{U_{\tau}}}\lambda_{{V_{\tau}}} \normapx{U_\tau}\normapx{V_\tau}}} \\
=& n^{- (\frac{1}{2} -\al)\left(w(\tau) - \frac{w(U_\tau) + w(V_\tau)}{2}\right) - (\gam - \al \beta)\left(|E(\tau)| - \frac{|E(U_\tau)| + |E(V_\tau)|}{2}\right) + \frac{1}{2}\left(\frac{w(U_\tau) + w(V_\tau)}{2} - w(S)\right) + \frac{|\Iso(\tau)|}{2}}\,.
\end{align*}
\end{lemma}
\begin{proof}
For the diagonal shapes $U_\tau$ and $V_\tau$, we have 
\begin{align*}
\lam_{U_\tau}\normapx{U_\tau} &= n^{(\frac{\beta}{2} - \gam)|E(U_\tau)|}n^{\frac{\beta}{2}|E(U_\tau)|} = n^{(\beta- \gam)|E(U_\tau)|}\\
\lam_{V_\tau}\normapx{V_\tau} &= n^{(\frac{\beta}{2} - \gam)|E(V_\tau)|}n^{\frac{\beta}{2}|E(V_\tau)|} = n^{(\beta- \gam)|E(V_\tau)|}\,.
\end{align*}
Therefore,
\[ \sqrt{\lambda_{{U_{\tau}}}\lambda_{{V_{\tau}}} \normapx{U_\tau}\normapx{V_\tau}} = n^{(\beta- \gam)\left(\frac{|E(U_\tau)| + |E(V_\tau)|}{2}\right)}\,.\]
Multiplying with \cref{lem:GeneralCoefficientTimesNorm},
\begin{align*}
&\frac{\lam_\tau \normapx{\tau}}{\sqrt{\lambda_{{U_{\tau}}}\lambda_{{V_{\tau}}} \normapx{U_\tau}\normapx{V_\tau}}} \\
=& n^{(1 -\al)\left(\frac{|U_\tau| + |V_\tau|}{2}\right)-(\frac{1}{2} - \alpha)w(\tau) - \frac{w(S)}{2} + \frac{|\Iso(\tau)|}{2}  - (\gamma - {\alpha}{\beta})|E(\tau)| - (\beta - \gam)\left(\frac{|E(U_\tau)| + |E(V_\tau)|}{2}\right)} \\
=& n^{(1 -\al)\left(\frac{|U_\tau| + |V_\tau|}{2}\right)-(\frac{1}{2} - \alpha)w(\tau) - \frac{w(S)}{2} + \frac{|\Iso(\tau)|}{2}  - (\gamma - {\alpha}{\beta})|E(\tau)| - (\al\beta - \gam)\left(\frac{|E(U_\tau)| + |E(V_\tau)|}{2}\right)  + (\al\beta - \beta) \left(\frac{|E(U_\tau)| + |E(V_\tau)|}{2}\right)}\\
=& n^{(1 -\al)\left(\frac{w(U_\tau) + w(V_\tau)}{2}\right)-(\frac{1}{2} - \alpha)w(\tau) - \frac{w(S)}{2} + \frac{|\Iso(\tau)|}{2}  - (\gamma - {\alpha}{\beta})\left(|E(\tau)| - \frac{|E(U_\tau)| + |E(V_\tau)|}{2}\right)}\\
=& n^{- (\frac{1}{2} -\al)\left(w(\tau) - \frac{w(U_\tau) + w(V_\tau)}{2}\right) - (\gam - \al \beta)\left(|E(\tau)| - \frac{|E(U_\tau)| + |E(V_\tau)|}{2}\right) + \frac{1}{2}\left(\frac{w(U_\tau) + w(V_\tau)}{2} - w(S)\right) + \frac{|\Iso(\tau)|}{2}}\,.
\end{align*}
\end{proof}

As a corollary, we have the following combinatorial formula for the slack.
\begin{lemma}\label{cor:findingandusingslack}
\begin{align*}
\slack(\tau) =\\
&(\frac{1}{2} -\al)\left(w(\tau) - \frac{w(U_\tau) + w(V_\tau)}{2}\right) + (\gam - \al \beta)\left(|E(\tau)| - \frac{|E(U_\tau)| + |E(V_\tau)|}{2}\right) \\
-& \frac{1}{2}\left(\frac{w(U_\tau) + w(V_\tau)}{2} - w(S)\right) - \frac{|\Iso(\tau)|}{2} - \log_n(c_\tau^\approx)\\
\end{align*}
\end{lemma}

\subsection{Slack for middle shapes}
\label{sec:charging-middle-shapes}

We start by computing the slack for middle shapes.
This is the slack at the start of the process.
\begin{theorem}\label{thm:middleshapeslack}
Let $\tau$ be a proper middle shape. Then:
\[
\slack(\tau) \geq (\gamma - {\alpha}{\beta})\left(|E(\tau)| - \frac{|E(U_{\tau})| + |E(V_{\tau})|}{2}\right)
\]
\end{theorem}
\begin{proof}
Since $\tau$ is a proper middle shape, $|\Iso(\tau)| = 0$ and $c_\tau = 1$.
By \cref{cor:findingandusingslack} we have
\begin{align*}
\slack(\tau) & =(\frac{1}{2} -\al)\left(w(\tau) - \frac{w(U_\tau) + w(V_\tau)}{2}\right) + (\gam - \al \beta)\left(|E(\tau)| - \frac{|E(U_\tau)| + |E(V_\tau)|}{2}\right) \\
&- \frac{1}{2}\left(\frac{w(U_\tau) + w(V_\tau)}{2} - w(S)\right)\,.
\end{align*}
Furthermore, since $w(\tau) \geq w(S) = w(U_\tau) = w(V_\tau)$, the last term is 0,
and the first term is non-negative. Thus,
\[ \slack(\tau) \geq (\gam - \al \beta)\left(|E(\tau)| - \frac{|E(U_\tau)| + |E(V_\tau)|}{2}\right)\,.
\]
\end{proof}

\subsection{Slack for the PMVS subroutine}
\label{sec:pmvs-slack}
For this subsection, suppose that one iteration of the \textbf{Finding PMVS subroutine} replaces ribbon $R_2$ of shape $\tau$ by ribbon $R_2'$ of shape $\tau'$.

We will need to consider the ``removed edges'' $E(R_2) \setminus E(R'_2)$, which intuitively are the edges of the ribbon that were queried and not present, and formally are the edges that disappear during either the \textbf{Adding left and right indicators operation} or the \textbf{Removing middle edge indicators operation}.
Observe that all removed edges are in $U_\tau \cup V_\tau$.

Recall that in the \textbf{PMVS operation},
we take a triple of ribbons $R_1, R_2, R_3$ after indicators have been added,
and split $R_1$ across the leftmost SMVS $A'$ between $A_{R_1}$ and $B_{R_1}$, and likewise split $R_3$ across the rightmost SMVS $B'$ between $A_{R_3}$ and $B_{R_3}$.
\begin{definition}[$\gamma$ and $\gam'$]\label{def:gamma}
Let $\gamma$ be the shape of the part of $R_1$ between $A'$ and $B_{R_1}$. Let $\gamma'^\T$ be the shape of the part of $R_3$ between $A_{R_3}$ and $B'$.
\end{definition}
\begin{remark}
    Note that $\gam, \tau,$ and $\gam'^\T$ should include all removed edges, whereas $\tau'$ has the edges removed.
\end{remark}
\begin{remark}
While this $\gamma$ is technically different than the $\gamma$ for an intersection term in \cref{sec:intersection-terms}, it plays a similar role.
\end{remark}

\begin{lemma}\label{lem:gam-left-shape}
$\gam, \gam'$ are left shapes.
\end{lemma}
\begin{proof}
    Suppose that $S$ is a separator of $\gam$.
    We claim that $S$ is also a separator of $R_1$.
    Let $P$ be any path from $A_{R_1}$ to $B_{R_1}$.
    Since $A' = U_{\gam}$ is a separator for $R_1$,
    $P$ must pass through $A'$.
    Starting from the final vertex of the path in $A'$ gives a path from
    $A' = U_{\gam}$ to $B_{R_1} = V_{\gam}$
    which is entirely contained in $\gam$.
    Finally, since $S$ is a separator of $\gam$,
    it must contain a vertex of the path $P$, and thus
    $S$ cuts the path $P$.

    Since $R_1$ is a left ribbon, we conclude that $w(S) \geq w(V_{\gam})$ and the
    unique SMVS of $\gam$ is $V_\gam$.
\end{proof}

\begin{theorem}\label{thm:PMVSslack}
Let $R_2 \to R_2'$ be a ribbon that undergoes one iteration of the \textbf{Finding PMVS subroutine}.
Let $\tau$ and $\tau'$ be their respective shapes. Then
\begin{align*}
    \slack(\tau')& - \slack(\tau)\\
    \geq\;& \alpha \left(\frac{w(U_{\tau'}) + w(V_{\tau'}) - w(U_{\tau}) - w(V_{\tau})}{2}\right)\\
    & +(\gam - \al\beta)\left(|E(\tau')| - \frac{|E(U_{\tau'})| + |E(V_{\tau'})|}{2} - |E({\tau})| + \frac{|E(U_{{\tau}})| + |E(V_{{\tau}})|}{2}\right) \\
    &+(\gam-\al\beta)x
\end{align*}
where $x$ is the total number of removed edges.
\end{theorem}
\begin{proof}
By \cref{cor:findingandusingslack},
\begin{align*}
    &\slack(\tau') - \slack(\tau)\\
    =\;&(\frac{1}{2} -\al)\left(w(\tau') - \frac{w(U_{\tau'}) + w(V_{\tau'})}{2} - w(\tau) + \frac{w(U_\tau) + w(V_\tau)}{2}\right) \\
    &+ (\gam - \al \beta)\left(|E(\tau')| - \frac{|E(U_{\tau'})| + |E(V_{\tau'})|}{2} - |E(\tau)| + \frac{|E(U_\tau)| + |E(V_\tau)|}{2}\right) \\
    &- \frac{1}{2}\left(\frac{w(U_{\tau'}) + w(V_{\tau'})}{2} - w(S') - \frac{w(U_\tau) + w(V_\tau)}{2} + w(S)\right) \\
    &- \frac{|\Iso(\tau')| - |\Iso(\tau)|}{2} - \log_n(c_{\tau'}^\approx) + \log_n(c_\tau^\approx)\\
    =\;& (\frac{1}{2}-\al)\left(w(\tau') - w(\tau)\right)
    + \alpha \left(\frac{w(U_{\tau'}) + w(V_{\tau'}) - w(U_{\tau}) - w(V_{\tau})}{2}\right)\\
    & + \frac{w(U_\tau) + w(V_\tau) - w(U_{\tau'}) - w(V_{\tau'}) + w(S') - w(S) + |\Iso(\tau)| - |\Iso(\tau')|}{2}\\
    & +(\gam - \al\beta)\left(|E(\tau')| - \frac{|E(U_{\tau'})| + |E(V_{\tau'})|}{2} - |E(\tau)| + \frac{|E(U_{\tau})| + |E(V_{\tau})|}{2}\right) - \log_n\left(\tfrac{c^\approx_{\tau'}}{c^\approx_\tau}\right)
\end{align*}
The large term multiplied by $\frac{1}{2}$ is the most dangerous term
as it does not come with any small coefficient.
We analyze the terms as follows.

\begin{claim}
\label{claim:coefficient-edge-removal}
$\tfrac{c^\approx_{\tau'}}{c^\approx_{\tau}} = n^{-{\gamma}x}$
\end{claim}
\begin{proof}[Proof of \cref{claim:coefficient-edge-removal}]
    As noted after \cref{def:coefficient}, for each edge which is removed, we get a factor of magnitude  $n^{-\frac{\beta}{2}}$. Furthermore, we get a factor of $n^{\frac{\beta}{2} - \gamma}$ shifted from $\lambda_{\tau}$ to $c_{\tau'}$. Multiplying these factors together gives a factor of magnitude $n^{-\gamma}$ per removed edge.
\end{proof}

\begin{restatable}{claim}{tauprimelowerbound}
\label{claim:tauprimelowerbound}
$w(\tau') \geq w(\tau) + {\beta}x$
\end{restatable}
\begin{proof}[Proof of \cref{claim:tauprimelowerbound}]
Observe that
\begin{align*}
w(\tau') &= w(\gamma) + w(\tau) + w(\gamma'^\T) - w(U_{\tau}) - w(V_{\tau}) + \beta x\,.
\end{align*}
$U_{\tau} = V_\gam$ is an SMVS for $\gam$ so $w(\gam) \geq w(U_{\tau})$.
Similarly, $w(\gam') \geq w(V_\tau)$. 
Putting these equations together, we have that 
$w(\tau') \geq w(\tau) + {\beta}x$, as needed.
\end{proof}
The next lemma is a tradeoff lemma for the PMVS, which will be proven in the next sub-subsection.
\begin{restatable}[PMVS tradeoff lemma]{lemma}{PMVStradeofflemma}
\label{lem:PMVStradeofflemma}
\begin{align*}
w(U_{\tau'}) + w(V_{\tau'}) - w(S') + |\Iso(\tau')|  &\leq w(U_{\tau}) + w(V_{\tau}) - w(S) + |\Iso(\tau)| + \beta{x}
\end{align*}
\end{restatable}
Multiplying \cref{claim:tauprimelowerbound} by $\frac{1}{2} - \alpha$ and multiplying \cref{lem:PMVStradeofflemma} by $\frac{1}{2}$, we have that 
\begin{enumerate}
    \item $(\frac{1}{2} - \alpha)(w(\tau') - w(\tau)) \geq \frac{\beta}{2} x - \alpha\beta x$
    \item 
    \begin{align*}
    \frac{w(U_\tau) + w(V_\tau) - w(U_{\tau'}) - w(V_{\tau'}) + w(S') - w(S) + |\Iso(\tau)| - |\Iso(\tau')|}{2}&\geq -\frac{\beta}{2}x
    \end{align*}
\end{enumerate}
Using these equations in the formula above, we have that
\begingroup
\allowdisplaybreaks
\begin{align*}
    \slack(\tau')& - \slack(\tau)\\
    \geq\;& \frac{\beta}{2} x-\al\beta x +\al \left(\frac{w(U_{\tau'}) + w(V_{\tau'}) - w(U_{\tau}) - w(V_{\tau})}{2}\right)  - \frac{\beta}{2} x\\
    & +(\gam - \al\beta)\left(|E(\tau')| - \frac{|E(U_{\tau'})| + |E(V_{\tau'})|}{2} - |E(\tau)| + \frac{|E(U_{\tau})| + |E(V_{\tau})|}{2}\right)\\
    & + \gamma x\\
   =\;& (\gam-\al\beta)x + \alpha \left(\frac{w(U_{\tau'}) + w(V_{\tau'}) - w(U_{\tau}) - w(V_{\tau})}{2}\right)\\
    & +(\gam - \al\beta)\left(|E(\tau')| - \frac{|E(U_{\tau'})| + |E(V_{\tau'})|}{2} - |E(\tau)| + \frac{|E(U_{\tau})| + |E(V_{\tau})|}{2}\right)
\end{align*}
\endgroup
as needed.
\end{proof}

\begin{corollary} \label{cor:pmvs-slack}
\begin{align*} 
\slack(&\tau') - \slack(\tau) \\
&\geq  (\gam-\al\beta) \cdot \left(|E_{tot}(\tau')| - \frac{|E(U_{\tau'})|+|E(V_{\tau'})|}{2} - |E_{tot}(\tau)| + \frac{|E(U_\tau)|+|E(V_\tau)|}{2}\right) 
\end{align*}
\end{corollary}
\begin{proof}
Using our slack calculation above, observe that we can remove the vertex factor as
\[  \alpha \left(\frac{w(U_{\tau'}) + w(V_{\tau'}) - w(U_{\tau}) - w(V_{\tau})}{2}\right) \geq 0\]
since $U_{\tau'}$ (including the removed edges) is a larger separator than $U_\tau$, as $U_\tau$ is the SMVS of $R_1$ at this point in the iteration (and likewise for $V_{\tau'}$ and $V_\tau$).
Removing edges only increases $w(U_{\tau'})$.
\end{proof}

\subsubsection{Proof of the PMVS tradeoff lemma}

\PMVStradeofflemma*

\begin{proof}
To prove this, we construct and analyze the following sets of vertices.
\begin{enumerate}
\item We take $X_0$ to be the set of vertices in $V_{\gamma} \setminus S'$ which can be reached in $\gamma$ from $U_{\gamma}$ without passing through a vertex in $S'$. We then take $X = X_0 \cup (S' \cap V(\gamma))$.
\item We take $Y_{l}$ to be the set of non-isolated vertices in $V_{\gamma} \setminus S'$ which are not reachable in $\gamma$ from $U_{\gamma}$ without passing through $S'$. Similarly, we take $Y_{r}$ to be the set of non-isolated vertices in $U_{{\gamma'}^\T} \setminus S'$ which are not reachable in ${\gamma'}^\T$ from $V_{{\gamma'}^\T}$ without passing through a vertex in $S'$. We then take $Y = Y_l \cup Y_r \cup (S' \cap V(\tau)) \cup (V_{\gamma} \cap U_{{\gamma'}^\T})$.
\item We take $Z_0$ to be the set of vertices in $U_{{\gamma'}^\T} \setminus S'$ which can be reached in ${\gamma'}^\T$ from $V_{{\gamma'}^\T}$ without passing through a vertex in $S'$. We take $Z_{extra}$ to be the set of non-isolated vertices in $(U_{\tau} \cap V_{\tau}) \setminus S'$ which are not reachable from $V_{{\gamma'}^\T}$ in ${\gamma'}^\T$. We then take $Z = Z_0 \cup Z_{extra} \cup (S' \cap V({\gamma'}^\T))$.
\end{enumerate}
Let $x_\cap$ be the number of edges removed from $U_\tau \cap V_\tau$.
We now observe that it is sufficient to show the following statements.
\begin{enumerate}
\item $w(X) \geq w(U_{\gamma})$, $w(Y) \geq w(S) + \beta x_\cap$, and $w(Z) \geq w(V_{{\gamma'}^\T})$.
\item 
\begin{align*}
w(X) + w(Y) + w(Z) &\leq w(S') + w(V_{\gamma}) + w(U_{{\gamma'}^\T}) - (|\Iso(\tau')| - |\Iso(\tau)|)\\
&= w(S') \\
& \quad+ w(U_{\tau}) + {\beta}(\# \text{ of edges removed from } U_{\tau}) \\
&\quad+ w(V_{\tau}) + {\beta}(\# \text{ of edges removed from } V_{\tau}) \\
&\quad-  (|\Iso(\tau')| - |\Iso(\tau)|)
\end{align*}
\end{enumerate}
Using the three initial inequalities on the left-hand side of the second statement,
\begin{align*}
w(U_\gam) + w(S) + \beta x_\cap + w(V_{\gam'^\T}) \leq &\;
 w(S') + w(U_{\tau}) + w(V_{\tau}) \\&\quad + \beta x + \beta x_\cap
-  \abs{\Iso(\tau')} + |\Iso(\tau)|\,.
\end{align*}
Rearranging this, we have
\begin{align*}
w(U_{\tau'}) + w(V_{\tau'}) - w(S') + |\Iso(\tau')|  &\leq w(U_{\tau}) + w(V_{\tau}) - w(S) + |\Iso(\tau)| + \beta{x}
\end{align*}
as needed.

Moving to the statements, to show that $w(X) \geq w(U_{\gamma})$ and $w(Z) \geq w(V_{{\gamma'}^\T})$, we observe that because of how we chose $X$ and $Z$, $X$ is a vertex separator of $\gamma$ and $Z$ is a vertex separator of ${\gamma'}^\T$. 

To show that $w(Y) \geq w(S) + \beta x_\cap$, we first observe that $Y$ is a vertex separator of $\tau$ (with or without the missing edges). To see this, assume that $Y$ is not a vertex separator of $\tau$. If so, there is a path $P_m$ from a vertex $u \in U_{\tau}$ to a vertex $v \in V_{\tau}$ which does not intersect $Y$ and thus does not intersect $S'$. Since $u \notin Y$, there is a path $P_l$ from $U_{\gamma}$ to $u$ in $\gamma$ which does not intersect $S'$. Similarly, since $v \notin Y$, there is a path $P_r$ from $v$ to $V_{{\gamma'}^\T}$ in ${\gamma'}^\T$ which does not intersect $S'$. Composing $P_l$, $P_m$, and $P_r$ gives a path from $U_{\tau'}$ to $V_{\tau'}$ which does not intersect $S'$ which is a contradiction as $S'$ is an SMVS for $\tau'$.

We now observe that when the missing edges are removed, all vertex separators of $\tau$ have their weight increased by at least $\beta x_\cap$. Thus, after the missing edges are deleted, all vertex separators of $\tau$ have weight at least $w(S) + \beta x_\cap$ so $w(Y) \geq w(S) + \beta x_\cap$.

To show that
\[
w(X) + w(Y) + w(Z) \leq w(S') + w(V_{\gamma}) + w(U_{{\gamma'}^\T}) - (|\Iso(\tau')| - |\Iso(\tau)|)
\]
we show the following two statements:
\begin{enumerate}
\item For each vertex $v$, $\one_{v \in X} + \one_{v \in Y} + \one_{v \in Z} \leq \one_{v \in S'} + \one_{v \in V_{\gamma}} + \one_{v \in U_{{\gamma'}^\T}} - \one_{v \in \Iso(\tau') \setminus \Iso(\tau)}$.
\item For each edge $e$, $\one_{e \in E(X)} + \one_{e \in E(Y)} + \one_{e \in E(Z)} \geq \one_{e \in E(S')} + \one_{e \in E(V_{\gamma})} + \one_{e \in E(U_{{\gamma'}^\T})}$.
\end{enumerate}
For the first statement, we make the following observations:
\begin{enumerate}
\item If $v \in S'$ then $v$ is not isolated. If $v \notin V_{\gamma} \cup U_{{\gamma'}^\T}$ then $v$ is in exactly one of $X$, $Y$, and $Z$ depending on whether $v$ is in $V(\gamma)$, $V(\tau)$, or $V({\gamma'}^\T)$. If $v \in V_{\gamma} \cup U_{{\gamma'}^\T}$ then $v \in Y$, $v \in X$ if and only if $v \in V_{\gamma}$ and $v \in Z$ if and only if $v \in U_{{\gamma'}^\T}$.
\item If $v \in \Iso(\tau') \setminus \Iso(\tau)$ then $v \notin X$ and $v \notin Z$. Moreover, $v$ must be in $V_{\gamma}$ or $U_{{\gamma'}^\T}$ and $v \in Y$ if and only if $v$ is in both $V_{\gamma}$ and $U_{{\gamma'}^\T}$.
\item If $v \in (V_{\gamma} \setminus S') \setminus U_{{\gamma'}^\T}$ and $v$ is not isolated then $\one_{v \in S'} + \one_{v \in V_{\gamma}} + \one_{v \in U_{{\gamma'}^\T}} - \one_{v \text { is isolated}} = 1$ and $\one_{v \in X} + \one_{v \in Y} + \one_{v \in Z} = 1$ as $v \notin Z$ and either $v \in X$ or $v \in Y$ but not both.
\item If $v \in (U_{{\gamma'}^\T} \setminus S') \setminus V_{\gamma}$ and $v$ is not isolated then $\one_{v \in S'} + \one_{v \in V_{\gamma}} + \one_{v \in U_{{\gamma'}^\T}} - \one_{v \text { is isolated}} = 1$ and $\one_{v \in X} + \one_{v \in Y} + \one_{v \in Z} = 1$ as $v \notin X$ and either $v \in Y$ or $v \in Z$ but not both.
\item If $v \in (V_{\gamma} \cap U_{{\gamma'}^\T}) \setminus S' $ and $v$ is not isolated then $\one_{v \in S'} + \one_{v \in V_{\gamma}} + \one_{v \in U_{{\gamma'}^\T}} - \one_{v \text { is isolated}} = 2$ and $\one_{v \in X} + \one_{v \in Y} + \one_{v \in Z} = 2$ as $v \in Y$ and either $v \in X$ or $v \in Z$ but not both.
\end{enumerate}
For the second statement, we make the following observations:
\begin{enumerate}
\item If $e \in E(S')$ and $e \notin E(V_{\gamma}) \cup E(U_{{\gamma'}^\T})$ then $e$ is in exactly one of $E(X)$, $E(Y)$, and $E(Z)$ depending on whether $e$ is in $E(\gamma)$, $E(\tau)$, or $E({\gamma'}^\T)$. If $e \in E(V_{\gamma}) \cup E(U_{{\gamma'}^\T})$ then $e \in E(Y)$, $e \in E(X)$ if and only if $e \in E(V_{\gamma})$, and $e \in E(Z)$ if and only if $e \in E(U_{{\gamma'}^\T})$. In all of these cases, $\one_{e \in E(X)} + \one_{e \in E(Y)} + \one_{e \in E(Z)} \geq \one_{e \in E(S')} + \one_{e \in E(V_{\gamma})} + \one_{e \in E(U_{{\gamma'}^\T})}$.
\item If $e \in (E(V_{\gamma}) \setminus E(S')) \setminus E(U_{{\gamma'}^\T})$ then $\one_{e \in E(S')} + \one_{e \in E(V_{\gamma})} + \one_{e \in E(U_{{\gamma'}^\T})} = 1$ and $\one_{e \in E(X)} + \one_{e \in E(Y)} + \one_{e \in E(Z)} = 1$ as $e \notin E(Z)$ and either $e \in E(X)$ or $e \in E(Y)$ but not both.
\item If $e \in (E(U_{{\gamma'}^\T}) \setminus E(S')) \setminus E(V_{\gamma})$ then $\one_{e \in E(S')} + \one_{e \in E(V_{\gamma})} + \one_{e \in E(U_{{\gamma'}^\T})} = 1$ and $\one_{e \in E(X)} + \one_{e \in E(Y)} + \one_{e \in E(Z)} = 1$ as $e \notin E(X)$ and either $e \in E(Y)$ or $e \in E(Z)$ but not both.
\item If $e \in (E(V_{\gamma}) \cap E(U_{{\gamma'}^\T})) \setminus E(S')$ then $\one_{e \in E(S')} + \one_{e \in E(V_{\gamma})} + \one_{e \in E(U_{{\gamma'}^\T})} = 2$ and $\one_{e \in E(X)} + \one_{e \in E(Y)} + \one_{e \in E(Z)} = 2$ as $e \in E(Z)$ and either $e \in E(X)$ or $e \in E(Z)$ but not both.
\end{enumerate}
\end{proof}

\subsection{Slack for intersection terms}
\label{sec:intersection-terms}

We now analyze the slack for the \textbf{Intersection term decomposition operation}.
The analysis is similar to that for the PMVS subroutine, albeit with additional considerations.

Suppose that the operation replaces a ribbon $R_2$ of shape $\tau$
by a ribbon $R'_2$ of shape $\tau_P$.
Recall that in the \textbf{Intersection term decomposition operation},
we take a triple of ribbons $R_1, R_2, R_3$ with intersections,
and split $R_1$ across the leftmost SMVS $A'$ between $A_{R_1}$ and $B_{R_1} \cup V_{intersected}(R_1)$, and likewise split $R_3$ across the rightmost SMVS $B'$ between $A_{R_3} \cup V_{intersected}(R_3)$ and $B_{R_3}$.
\begin{definition}[$\gam$ and $\gam'$]
\label{def:gam}
    Let $\gam$ be the shape of the part of $R_1$ between $A'$ and $B_{R_1}$. Let $\gam'^\T$ be the shape of the part of $R_3$ between $A_{R_3}$ and $B'$.
\end{definition}

The notation $\tau_P$ is used because $\tau_P$ is an intersection shape (\cref{def:intersection-shape}) for some intersection pattern $P \in \calP_{\gam, \tau, \gam'^\T}$ (after linearization).

The edges that are linearized away into a constant term during the \textbf{Intersection term decomposition operation}
are referred to as ``vanishing edges''.

\begin{theorem}\label{thm:intersectiontermslack}
Let $R_2 \to R_2'$ be a ribbon that undergoes the \textbf{Intersection term decomposition operation}.
Let $\tau$ and $\tau_P$ be their respective shapes.
Then for $\gam, \gam'$ as defined in \cref{def:gam},
\begin{align*}
\slack(\tau_P)& - \slack(\tau) \geq\\
&(1 - \al)\left(\frac{w(U_{\gam}) + w(V_{\gam'^\T}) - w(U_{\tau}) - w(V_{{\tau}})}{2}\right) \\
&+(\gam - {\al}\beta)\left(|E(\tau_P)| - \frac{|E(U_{\tau_P})| + |E(V_{\tau_P})|}{2} - |E({\tau})| + \frac{|E(U_{{\tau}})| + |E(V_{{\tau}})|}{2}\right)\\
&+(\gam - \al\beta)\cdot\text{edgereduction}
\end{align*}
where $\text{edgereduction}$ is the total number of vanishing edges.
\end{theorem}
\begin{proof}
By \cref{cor:findingandusingslack},
\begin{align*}
    &\slack(\tau_P) - \slack(\tau)\\
    &=(\frac{1}{2} -\al)\left(w(\tau_P) - \frac{w(U_{\tau_P}) + w(V_{\tau_P})}{2} - w(\tau) + \frac{w(U_\tau) + w(V_\tau)}{2}\right) \\
    &+ (\gam - \al \beta)\left(|E(\tau_P)| - \frac{|E(U_{\tau_P})| + |E(V_{\tau_P})|}{2} - |E(\tau)| + \frac{|E(U_\tau)| + |E(V_\tau)|}{2}\right) \\
    &- \frac{1}{2}\left(\frac{w(U_{\tau_P}) + w(V_{\tau_P})}{2} - w(S') - \frac{w(U_\tau) + w(V_\tau)}{2} + w(S)\right) \\
    &- \frac{|\Iso(\tau_P)| - |\Iso(\tau)|}{2} - \log_n(c_{\tau_P}^\approx) + \log_n(c_\tau^\approx)\\
    &= (\frac{1}{2}-\al)\left(w(\tau_P) - w(\tau)\right) - \frac{1 - \alpha}{2}\left(w(U_{\tau_P}) + w(V_{\tau_P}) - w(U_{\gam}) - w(V_{{\gam'}^{\T}})\right)\\
    &+ \alpha \left(\frac{w(U_{\gam}) + w(V_{{\gam'}^{\T}}) - w(U_{\tau}) - w(V_{\tau})}{2}\right)\\
    &+ \frac{w(U_\tau) + w(V_\tau) - w(U_{\gam}) - w(V_{{\gam'}^{\T}}) + w(S') - w(S) + |\Iso(\tau)| - |\Iso(\tau_P)|}{2}\\
    & +(\gam - \al\beta)\left(|E(\tau_P)| - \frac{|E(U_{\tau_P})| + |E(V_{\tau_P})|}{2} - |E(\tau)| + \frac{|E(U_{\tau})| + |E(V_{\tau})|}{2}\right)- \log_n\left(\tfrac{c^\approx_{\tau_P}}{c^\approx_\tau}\right)
\end{align*}

We now analyze the different terms that appear.

\begin{claim}\label{claim:coefficient-intersection-term}
\begin{align*}
\frac{c^\approx_{\tau_P}}{c^\approx_{\tau}} & = n^{(\alpha - 1)\cdot (\# \text{ of intersections})} \left(\prod_{e \in E_{tot}(\tau_P)}{\left(n^{\frac{\beta}{2} - \gamma}\right)^{\mult(e) - 1 + \one_{e \text{ vanishes}}}\left(n^{\frac{\beta}{2}}\right)^{\mult(e) - 1 - \one_{e \text{ vanishes}}}}\right)
\end{align*}
\end{claim}
Note that we consider the iteration $e\in E_{tot}(\tau_P)$
to yield each multiedge only once.
\begin{proof}[Proof of \cref{claim:coefficient-intersection-term}]
Observe that for each intersection, we have one fewer vertex factor of $n^{(\alpha - 1)}$ in $\lambda_{\tau_P}$ so we need to add this factor to $c_{\tau_P}$. For each multiedge in $E_{tot}(\tau_P)$, we have $\mult(e) - 1 + \one_{e \text{ vanishes}}$ fewer factors of $n^{\frac{\beta}{2}-\gam}$ in $\lambda_{\tau_P}$ so these factors need to be added to $c_{\tau_P}$.  Finally, for each multiedge in $e \in E_{tot}(\tau_P)$, when we express it as a linear combination of $1$ and $\chi_{\{e\}}$, we gain $\mult(e) - 1 - \one_{e \text{ vanishes}}$ factors of $n^{\frac{\beta}{2}}$ (\cref{prop:linearization-coefficient}) which also need to be added to $c_{\tau_P}$. When a multiedge with an indicator is linearized, it never vanishes, and the coefficient is the same
as linearizing a multiedge without an indicator.
\end{proof}

For notational convenience, we define the following expressions:
\begin{align*}
\text{linearization} &= \sum_{e \in E_{tot}(\tau_P)} \mult(e) - 1 - \one_{e\text{ vanishes}}\,.
\end{align*}
\[
\text{edgereduction} = \sum_{e \in E_{tot}(\tau_P)} \mult(e) - 1 + \one_{e\text{ vanishes}}
\]
With these expressions, we can express $\log_n\left(\tfrac{c^\approx_{\tau_P}}{c^\approx_\tau}\right)$ as follows.
\begin{corollary}
\[
\log_n\left(\tfrac{c^\approx_{\tau_P}}{c^\approx_\tau}\right) = (\alpha - 1)\cdot (\# \text{ of intersections}) + (\frac{\beta}{2} - \gamma)(\text{edgereduction}) + \frac{\beta}{2}(\text{linearization})
\]
\end{corollary}

\begin{restatable}{claim}{tauPlowerbound}
\label{claim:tauPlowerbound}
\begin{align*}
w(\tau_P) &\geq w(U_{\gamma}) + w(V_{{\gam'}^{\T}}) - w(U_{\tau}) - w(V_{\tau}) + w(\tau) \\
&- (\# \text{ of intersections}) + {\beta}(\text{edgereduction})
\end{align*}
\end{restatable}
\begin{proof}[Proof of \cref{claim:tauPlowerbound}]
We first observe that 
\begin{align*}
w(\tau_P) &= w(\gamma) + w(\tau) + w({\gamma'}^\T) - w(U_{\tau}) - w(V_{\tau})\\
&- (\# \text{ of intersections}) + {\beta}(\text{edgereduction})\,.
\end{align*}
To see this, note that if there were no intersections then we would have that $w(\tau_P) = w(\gamma) + w(\tau) + w({\gamma'}^\T) - w(U_{\tau}) - w(V_{\tau})$. Each intersection reduces the number of vertices and thus decreases the weight of $\tau_P$ by $1$. The change in the number of edges increases the weight of $\tau_P$ by ${\beta}(\text{edgereduction})$.

We now observe that $w(\gamma) \geq w(U_{\gam})$ as otherwise $\gamma$ would be a separator in $\gam$ between $U_{\gamma}$ and $V_{\gamma} \cup V_{intersected}(\gamma)$ with smaller weight than $U_{\gamma}$. Following similar logic, $w(\gamma') \geq w(V_{{\gam'}^\T})$. Thus, we have that 
\begin{align*}
w(\tau_P) &\geq w(U_{\gamma}) + w(V_{{\gamma'}^\T}) - w(U_{\tau}) - w(V_{\tau}) + w(\tau) \\
&- (\# \text{ of intersections}) + {\beta}(\text{edgereduction})
\end{align*}
as needed.
\end{proof}

The next lemma is an intersection tradeoff lemma to be proven in the next sub-section.
\begin{restatable}[Intersection tradeoff lemma]{lemma}{intersectiontradeofflemma}
\label{lem:intersectiontradeofflemma}
Given is a shape $\tau$, left shapes $\gam, \gam'$, an intersection pattern $P \in \calP_{\gamma^-, \tau, (\gamma'^\T)^-}$ such that the following structural property holds:
\begin{quote}
    $U_\gamma$ is the leftmost SMVS of $U_\gam$ and $V_\gam \cup V_{intersected}(\gam)$, and $V_{\gam'^\T}$ is the rightmost SMVS of $U_{\gam'^\T} \cup V_{intersected}(\gam'^\T)$ and $V_{\gam'^\T}$.
\end{quote}
Let $S$ be an SMVS of $\tau$, let $\tau_P$ be the shape resulting from $P$ followed by linearization, and let $S'$ be an SMVS of $\tau_P$. Then,
\begin{align*}
w(U_{\gam}) + w(V_{{\gam'}^{\T}}) - w(S') +  |\Iso(\tau_P)| \leq \;& w(U_{\tau}) + w(V_{\tau}) - w(S) + |\Iso(\tau)| \\
&+(\# \text{ of intersections}) - \beta\left(\text{linearization}\right) \\
&-\beta(\# \text{ of vanishing edges in } U_{\gam}) \\
&-\beta(\# \text{ of vanishing edges in } V_{{\gam'}^\T})
\end{align*}
\end{restatable}
Multiplying \cref{claim:tauPlowerbound} by $\frac{1}{2} - \alpha $ and multiplying \cref{lem:intersectiontradeofflemma} by $\frac{1}{2}$, we have that
\begin{enumerate}
    \item \begin{align*}
       &(\frac{1}{2}-\al)\left(w(\tau_P) - w(\tau)\right) \geq \\
       &(\frac{1}{2} - \al)\left(w(U_{\gamma}) + w(V_{{\gamma'}^\T}) - w(U_{\tau}) - w(V_{\tau}) - (\# \text{ of intersections})\right) \\
       &+(\frac{1}{2} - \al){\beta}(\text{edgereduction})
    \end{align*}
    \item 
    \begin{align*}
    &\frac{w(U_{\tau}) + w(V_{\tau}) - w(U_{\tau_P}) - w(V_{\tau_P}) + w(S') - w(S) + |\Iso(\tau)| - |\Iso(\tau_P)|}{2} \\
    &\geq \frac{-1}{2}(\# \text{ of intersections}) +\frac{\beta}{2}\left(\text{linearization}\right) \\
    &+\frac{\beta}{2}((\# \text{ of vanishing edges in } U_{\gam}) + (\# \text{ of vanishing edges in } V_{{\gam'}^\T}))
    \end{align*}\end{enumerate}
Using these equations in the formula above, we have that 
\begingroup
\allowdisplaybreaks
\begin{align*}
    \slack(\tau_P)& - \slack(\tau)\\
    \geq\;& (\frac{1}{2}-\al)\left(w(U_{\gam}) +w(V_{{\gam'}^\T}) - w(U_{\tau}) - w(V_\tau) - (\# \text{ of intersections})\right)\\
    & +(\frac{1}{2} - \al){\beta}\left(\text{edgereduction}\right) - \frac{1 - \alpha}{2}\left(w(U_{\tau_P}) + w(V_{\tau_P}) - w(U_{\gam}) - w(V_{{\gam'}^{\T}})\right) \\
    &+ \alpha \left(\frac{w(U_{\gam}) + w(V_{{\gam'}^{\T}}) - w(U_{\tau}) - w(V_{\tau})}{2}\right)\\
    & + \frac{\beta}{2} \left(\text{linearization}\right) - \frac{1}{2}(\# \text{ of intersections})\\
    &+\frac{\beta}{2}((\# \text{ of vanishing edges in } U_{\gam}) + (\# \text{ of vanishing edges in } V_{{\gam'}^\T}))\\
    & +(\gam - \al\beta)\left(|E(\tau_P)| - \frac{|E(U_{\tau_P})| + |E(V_{\tau_P})|}{2} - |E(\tau)| + \frac{|E(U_{\tau})| + |E(V_{\tau})|}{2}\right)\\
    &- (\alpha - 1)(\# \text{ of intersections}) - \frac{\beta}{2}\left(\text{linearization}\right) \\
    &- (\frac{\beta}{2} - \gam)(\text{edgereduction})\\
=\; & \frac{1-\al}{2}\left(w(U_{\gam}) +w(V_{{\gam'}^\T}) - w(U_{\tau}) - w(V_\tau)\right)\\
    &- \frac{1 - \alpha}{2}\left(w(U_{\tau_P}) + w(V_{\tau_P}) - w(U_{\gam}) - w(V_{{\gam'}^{\T}})\right) + (\gam - {\alpha}\beta)\left(\text{edgereduction}\right)\\
    &+\frac{\beta}{2}((\# \text{ of vanishing edges in } U_{\gam}) + (\# \text{ of vanishing edges in } V_{{\gam'}^\T}))\\
    &+(\gam - \al\beta)\left(|E(\tau_P)| - \frac{|E(U_{\tau_P})| + |E(V_{\tau_P})|}{2} - |E(\tau)| + \frac{|E(U_{\tau})| + |E(V_{\tau})|}{2}\right)
\end{align*}
\endgroup
We now use that
    \begin{align*}
    &w(U_{\tau_P}) + w(V_{\tau_P}) - w(U_{\gam}) - w(V_{{\gam'}^\T})\\
    =\;& \beta(\# \text{ of vanishing edges in } U_{\gam}) + \beta(\# \text{ of vanishing edges in } V_{{\gam'}^\T})\\
    &-\beta(\# \text{ of edges added to } U_{\gam}) - \beta(\# \text{ of edges added to } V_{{\gam'}^\T})\\
    \leq\;& \beta(\# \text{ of vanishing edges in } U_{\gam}) + \beta(\# \text{ of vanishing edges in } V_{{\gam'}^\T})
    \end{align*}
Plugging this in, we have that
\begin{align*}
&\slack(\tau_P) - \slack(\tau) \\
\geq\;&\frac{1 - \al}{2}\left(w(U_{\gam}) + w(V_{{\gam'}^\T}) - w(U_{\tau}) - w(V_{{\tau}})\right) \\
&+(\gam - \al\beta)\left(|E(\tau_P)| - \frac{|E(U_{\tau_P})| + |E(V_{\tau_P})|}{2} - |E(\tau)| + \frac{|E(U_{\tau})| + |E(V_{\tau})|}{2}\right)\\
&- \frac{\beta-\alpha \beta}{2} \left((\text{\# of vanishing edges in }U_\gam) + (\text{\# of vanishing edges in }V_{\gam'^\T})\right)\\
&+ (\gam - \al\beta)\left(\text{edgereduction}\right)\\
&+\frac{\beta}{2}\left((\text{\# of vanishing edges in }U_\gam) + (\text{\# of vanishing edges in }V_{\gam'^\T})\right)\\
=\;&\frac{1 - \al}{2}\left(w(U_{\gam}) + w(V_{{\gam'}^\T}) - w(U_{\tau}) - w(V_{{\tau}})\right) \\
&+(\gam - \al\beta)\left(|E(\tau_P)| - \frac{|E(U_{\tau_P})| + |E(V_{\tau_P})|}{2} - |E(\tau)| + \frac{|E(U_{\tau})| + |E(V_{\tau})|}{2}\right)\\
&+ (\gam - \al\beta)\cdot\text{edgereduction}\\
&+\frac{\al\beta}{2}\left((\text{\# of vanishing edges in }U_\gam) + (\text{\# of vanishing edges in }V_{\gam'^\T})\right)
\end{align*}
as needed.
\end{proof}

\begin{corollary}\label{cor:intersection-slack}
\begin{align*} 
\slack(&\tau_P) - \slack(\tau) \\
&\geq\;(\gam-\al\beta)\cdot\left(|E_{tot}(\tau_P)| - \frac{|E(U_{\tau_P})| +|E(V_{\tau_P})| }{2} - |E_{tot}(\tau)| + \frac{|E(U_{\tau})| +|E(V_{\tau})| }{2} \right)\\
&+\frac{1-\alpha}{2}\left(\text{\# of vertices in }(U_{\tau_P} \cup V_{\tau_P}) \setminus (U_{\tau_P} \cap V_{\tau_P})\text{ not incident to }E_{tot}(\tau_P)\right)\,.
\end{align*}
\end{corollary}
\begin{proof}
Following our slack calculation above, observe that
\[  (1-\alpha) \left(\frac{w(U_{\gam}) + w(V_{\gam'^\T}) - w(U_{\tau}) - w(V_{\tau})}{2}\right) \geq 0\]
since $U_{\gam}$ is a larger separator than $U_\tau = V_\gam$, as $\gam$ is a left shape (and likewise for $V_{\gam'^\T}$ and $V_\tau = U_{\gam'^\T}$).
Furthermore, this holds if we remove the degree-0 vertices from
$U_\gam$ or $V_{\gam'^\T}$, since $U_\gam$ remains a separator
without these vertices.
Therefore we may replace the right-hand side by
\[ \frac{1-\alpha}{2}\left(\text{\# of degree-0 vertices in }(U_\gam \cup V_{\gam'^\T}) \setminus (U_\gam \cap V_{\gam'^\T})\right)\,.\]
The edge factors follow immediately from the slack formula.
Thus the claim holds.
\end{proof}

\subsubsection{Proof of the intersection tradeoff lemma}

\intersectiontradeofflemma*
\begin{proof}
Let $S'_{pre}$ be the preimage of $S'$ before the intersections, as a subset of $V(\gamma \circ \tau \circ \gamma'^\T)$. We construct sets $X \subseteq V(\gamma)$, $Y \subseteq V(\tau)$, and $Z \subseteq V(\gamma'^\T)$ as follows.
\begin{enumerate}
\item We take $X_0$ to be the set of vertices in $V(\gamma) \setminus S'_{pre}$ which can be reached from $U_{\gamma}$ by a path of non-vanishing edges in $\gamma \setminus S'_{pre}$, are either intersected or are in $V_{\gamma}$, and are not isolated in $\tau_P$. We then take $X = X_0 \cup (S'_{pre} \cap V(\gamma))$.
\item We take $Y_{l}$ to be the set of non-isolated vertices in $V(\tau)$ which are also in $V(\gamma)$ (either because they are in $V_{\gamma} = U_{\tau}$ or because of an intersection) and which are not reachable from $U_{\gamma}$ by a path of non-vanishing edges in $\gamma \setminus S'_{pre}$. Similarly, we take $Y_{r}$ to be the set of non-isolated vertices in $V(\tau)$ which are also in $V({\gamma'}^\T)$ (either because they are in $V_{\tau} = U_{\gamma'^\T}$ or because of an intersection) and which are not reachable from $V_{{\gamma'}^\T}$ by a path of non-vanishing edges in ${\gamma'}^\T \setminus S'_{pre}$. We then take $Y = Y_l \cup Y_r \cup V_{common} \cup (S'_{pre} \cap V_{\tau})$ where $V_{common}$ is the set of vertices which appear in $\gamma$, $\tau$, and ${\gamma'}^\T$.
\item We take $Z_0$ to be the set of vertices in $V({\gamma'}^\T) \setminus S'_{pre}$ which can be reached from $V_{{\gamma'}^\T}$ by a path of non-vanishing edges in ${\gamma'}^\T \setminus S'_{pre}$, are either intersected or are in $U_{{\gamma'}^\T}$, and are not isolated in $\tau_P$. We take $Z_{extra}$ to be the set of vertices in $V({\gamma'}^\T) \setminus S'_{pre}$ which are also in $V(\gamma)$ and which are not reachable from $U_{\gamma}$ by a path of non-vanishing edges in $\gamma \setminus S'_{pre}$. We then take $Z = Z_0 \cup Z_{extra} \cup (S'_{pre} \cap V({\gamma'}^\T))$.
\end{enumerate}
\begin{claim}\label{claim:xyz-separators}\ 
\begin{enumerate}
\item $X$ separates $U_{\gamma}$ from $V_{\gamma} \cup \{\text{intersected vertices}\}$.
\item $Y$ separates $U_{\tau}$ from $V_{\tau}$
\item $Z$ separates $U_{\gamma'^\T} \cup \{\text{intersected vertices}\}$ from $V_{{\gamma'}^\T}$
\end{enumerate}
\end{claim}
\begin{proof}[Proof of \cref{claim:xyz-separators}]
Statements 1 and 3 follow from the definitions of $X$ and $Z$. For statement $2$, assume there is a path $P$ from $U_{\tau}$ to $V_{\tau}$ which does not intersect $Y$. Let $u$ be the last vertex in $P$ which is in $\gamma$ (either because $u$ is in $V_{\gamma} = U_{\tau}$ or because of an intersection) and let $v$ be the next vertex in $P$ which is in $\gamma'^\T$ (either because $v$ is in $V_{\tau} = U_{{\gamma'}^\T}$ or because of an intersection). Since $u \notin Y$, there is a path of non-vanishing edges in $\gamma \setminus S'_{pre}$ from $U_{\gamma}$ to $u$. Since $v \notin Y$, there is a path of non-vanishing edges in ${\gamma'}^\T \setminus S'_{pre}$ from $v$ to $V_{{\gamma'}^\T}$.

Now consider the part of $P$ between $u$ and $v$. For the vertices in this part, only $u$ is in $V(\gamma)$ and only $v$ is in $V({\gamma'}^\T)$, so no edges in this part vanish. Since no vertex of $P$ is in $Y$, none of these vertices are in $S'_{pre}$. Thus, we can compose these path segments to obtain a path of non-vanishing edges from $U_{\gamma}$ to $V_{\gamma}$ which does not intersect $S'$. This contradicts the fact that $S'$ is the SMVS of $\tau_P$.
\end{proof}
We conclude that $w(X) \geq w(U_{\gamma})$, $w(Y) \geq w(S)$, and $w(Z) \geq w(V_{{\gamma'}^\T})$ .
To complete the proof of the intersection tradeoff lemma, we now need to show that 
\begin{align*}
w(X) + w(Y) + w(Z) &\leq (\# \text{ of intersections}) - (|\Iso(\tau_P)| - |\Iso(\tau)|) \\
&- \beta\left(\sum_{e \in E_{tot}(\tau_P)}{(\mult(e) - 1 - \one_{e \text{ vanishes}})}\right) + w(U_{\tau}) + w(V_{\tau}) + w(S')\\
&-\beta((\# \text{ of vanishing edges in } U_{\gam}) + (\# \text{ of vanishing edges in } V_{{\gam'}^\T}))
\end{align*}
To show this, we consider the number of times vertices in $V(\tau_P)$ and edges in $E_{tot}(\tau_P)$ appear on both sides.
\begin{enumerate}
\item Vertices $u$ in $S'$ appear $\one_{u \in V(\gamma)} + \one_{u \in V(\tau)} + \one_{u \in V({\gamma'}^\T)}$ times on the left hand side and $1 + \one_{u \in U_{\tau}} + \one_{v \in V_{\tau}} + (\# \text{ of intersections for } u)$ times on the right hand side. It is not hard to check that these two expressions are equal.
\item Vertices $u$ which are not in $S'$ and are not isolated appear $(\# \text{ of intersections for } u) + \one_{u \in U_{\tau}} + \one_{v \in V_{\tau}}$ times on the right hand side and at most $(\# \text{ of intersections for } u) + \one_{u \in U_{\tau}} + \one_{v \in V_{\tau}} = \one_{u \in V(\gamma)} + \one_{u \in V(\tau)} + \one_{u \in V({\gamma'}^\T)} - 1$ times on the left hand side. To see this, observe that $u$ cannot be in both $X$ and $Y_{l}$, cannot be in both $Y_{r}$ and $Z_0$, and cannot be in both $X$ and $Z$.
\item Vertices $u$ which are in $\Iso(\tau_P) \setminus \Iso(\tau)$ (and thus not in $S'$) cannot appear in $X$ or $Z$ and appear in $Y$ if and only if they appear in $\gamma$, $\tau$, and ${\gamma'}^\T$. Thus, $u$ appears on the left hand side $\one_{u \in V(\gamma)} + \one_{u \in V(\tau)} + \one_{u \in V({\gamma'}^\T)} - 2$ times and appears on the right hand side $(\# \text{ of intersections for } u) + \one_{u \in U_{\tau}} + \one_{v \in V_{\tau}} - 1 = \one_{u \in V(\gamma)} + \one_{u \in V(\tau)} + \one_{u \in V({\gamma'}^\T)} - 2$ times.
\item Edges $e$ between two vertices in $S'$ appear $\one_{e \in E(\gamma)} + \one_{e \in E(\tau)} + \one_{e \in E({\gamma'}^\T)}$ times on the left hand side and 
\begin{align*}
&\one_{e \text{ does not vanish}} + \one_{e \in E(U_{\tau})} + \one_{e \in E(V_{\tau})} + (\mult(e) - 1 - \one_{e \text{ vanishes}}) \\
&+ \one_{e \text{ vanishes from } U_{\gam}} + \one_{e \text{ vanishes from } V_{{\gam'}^\T}} \\
&= \one_{e \in E(\gamma)} + \one_{e \in E(\tau)} + \one_{e \in E({\gamma'}^\T)} - 2\cdot\one_{e \text{ vanishes}} + \one_{e \text{ vanishes from } U_{\gam}} + \one_{e \text{ vanishes from } V_{{\gam'}^\T}} \\
&\leq \one_{e \in E(\gamma)} + \one_{e \in E(\tau)} + \one_{e \in E({\gamma'}^\T)}
\end{align*}
times on the right hand side.
\item Non-vanishing edges $e$ which are not between two vertices in $S'$ appear $\one_{e \in E(U_{\tau})} + \one_{e \in E(V_{\tau})} + \mult(e) - 1$ times on the right hand side and at least 
$\one_{e \in E(\gamma)} + \one_{e \in E(\tau)} + \one_{e \in E({\gamma'}^\T)} - 1 = \one_{e \in E(U_{\tau})} + \one_{e \in E(V_{\tau})} + \mult(e) - 1$ times on the left hand side. To see this, observe that if $e$ appears in both $E(\gamma)$ and $E(\tau)$ then either $e \in E(X)$ or $e \in E(Y)$ (depending on whether the endpoints of $e$ are reachable from $U_{\gamma}$). Similarly, if $e$ appears in both $E(\tau)$ and $E({\gamma'}^\T)$ then either $e \in E(Y)$ or $e \in E(Z)$. Finally, if $e$ appears in both $E(\gamma)$ and $E({\gamma'}^\T)$ then either $e \in E(X)$ or $e \in E(Z)$
\item Vanishing edges $e$ which are not between two vertices in $S'$ appear
\begin{align*}
&\one_{e \in E(U_{\tau})} + \one_{e \in E(V_{\tau})} + \mult(e) - 2 + \one_{e \text{ vanishes from } U_{\gam}} + \one_{e \text{ vanishes from } V_{{\gam'}^\T}}\\
&= \one_{e \in E(\gamma)} + \one_{e \in E(\tau)} + \one_{e \in E({\gamma'}^\T)} - 2 + \one_{e \text{ vanishes from } U_{\gam}} + \one_{e \text{ vanishes from } V_{{\gam'}^\T}}
\end{align*}
times on the right hand side and at least 
\[
\one_{e \in E(\gamma)} + \one_{e \in E(\tau)} + \one_{e \in E({\gamma'}^\T)} - 2 + \one_{e \text{ vanishes from } U_{\gam}} + \one_{e \text{ vanishes from } V_{{\gam'}^\T}}
\]
times on the left hand side. To see this, we make the following observations:
\begin{enumerate}
    \item $\one_{e \in E(Y)} \geq \one_{e \in E(\gamma)} + \one_{e \in E(\tau)} + \one_{e \in E({\gamma'}^\T)} - 2$ as if $e$ appears in $E(\gamma)$, $E(\tau)$, and $E({\gamma'}^\T)$ then $e \in E(Y)$.
    \item $\one_{e \in E(X)} \geq \one_{e \text{ vanishes from } U_{\gam}}$ as if $e$ vanishes from $U_{\gam}$ then $e \in E(X)$.
    \item $\one_{e \in E(Z)} \geq \one_{e \text{ vanishes from } V_{{\gam'}^\T}}$ as if $e$ vanishes from $V_{{\gam'}^\T}$ then $e \in E(Z)$.
\end{enumerate}
\end{enumerate}
\end{proof}

\subsection{Slack for Removing Middle Edge Indicators}

We analyze the slack for the \textbf{Removing middle edge indicators operation} after the \textbf{Intersection term decomposition operation}.
To do this,
we may imagine that an edge removed during the \textbf{Removing middle edge indicators operation} was removed immediately before modifying the sets $U_\tau$ and $V_\tau$;
applying the \textbf{Intersection term decomposition operation} will result in the same shape whether the edge
is removed before or after the operation.\footnote{Formally, this requires the following check. When an edge indicator moves out of $U_\tau \cup V_\tau$ and into the middle, this means that the SMVS no longer includes the edge. Therefore, whether or not the edge is present is not affecting the calculation of the SMVS, and therefore the presence of the edge is not affecting the recursion.}
\footnote{Edges removed by the \textbf{Removing middle edge indicators operation} in the \textbf{Finding PMVS subroutine}
can be analyzed using either \cref{thm:adding-indicators-slack} or the analysis
in \cref{sec:pmvs-slack}.}

\begin{theorem}\label{thm:adding-indicators-slack}
Let $R_2 \to R_2'$ be a ribbon that undergoes the \textbf{Adding left and right indicators operation}.
Let $\tau$ and $\tau'$ be their respective shapes. Then
\[
\slack(\tau') - \slack(\tau) \geq \frac{\gam}{2}\left(x + x_{\cap}\right)
\]
where $x$ is the total number of removed edges, and $x_\cap$ is the number of edges removed from $U_\tau \cap V_\tau$.
\end{theorem}

\begin{proof}
By \cref{cor:findingandusingslack},
\begin{align*}
    &\slack(\tau') - \slack(\tau)\\
    &= (\frac{1}{2}-\al)\left(w(\tau') - w(\tau)\right)
    - (1-\alpha) \left(\frac{w(U_{\tau'}) + w(V_{\tau'}) - w(U_{\tau}) - w(V_{\tau})}{2}\right)\\
    & +(\gam - \al\beta)\left(|E(\tau')| - \frac{|E(U_{\tau'})| + |E(V_{\tau'})|}{2} - |E(\tau)| + \frac{|E(U_{\tau})| + |E(V_{\tau})|}{2}\right)\\
    &+\frac{w(S') - w(S)}{2} - \frac{|\Iso(\tau')| - |\Iso(\tau)|}{2} - \log_n\left(\tfrac{c^\approx_{\tau'}}{c^\approx_\tau}\right)\,.
\end{align*}
We now make the following observations:
\begin{enumerate}
    \item For all vertex sets $X \subseteq V(\tau)$ equal to $X' \subseteq V(\tau')$,
we have
\[
    w(X') - w(X) = \beta \left(\text{\# of edges removed from }X\right)\,.
\]
In particular, $w(\tau') - w(\tau) = {\beta}x$, $w(U_{\tau'}) - w(U_{\tau}) = {\beta}\left(\text{\# of edges removed from }U_{\tau}\right)$, and $w(V_{\tau'}) - w(V_{\tau}) = {\beta}\left(\text{\# of edges removed from }V_{\tau}\right)$.
    \item $\log_n(c_{\tau'}^\approx) = \log_n(c_{\tau}^\approx) - {\gam}x$. Similarly to \cref{claim:coefficient-edge-removal},
    when the edge indicator is replaced by a constant term, we get a factor of magnitude  $n^{-\frac{\beta}{2}}$ from the update equation. Furthermore, we get a factor of $n^{\frac{\beta}{2} - \gamma}$ shifted from $\lambda_{\tau}$ to $c_{\tau'}$. Multiplying these factors together gives a factor of magnitude $n^{-\gamma}$ per removed edge.
    \item 
    \begin{align*}
        &|E(\tau')| - \frac{|E(U_{\tau'})| + |E(V_{\tau'})|}{2} - |E(\tau)| + \frac{|E(U_{\tau})| + |E(V_{\tau})|}{2} \\
        &= -\frac{1}{2}((\# \text { of edges removed from } U_{\tau} \setminus V_{\tau}) + (\# \text { of edges removed from } V_{\tau} \setminus U_{\tau})) \\
        &= -\frac{1}{2}(x - x_{\cap}) 
        \end{align*}
    \item $w(S') \geq w(S) + {\beta}x_{\cap}$. To see this, let $S'_{pre}$ be $S'$ before the edges are removed. Since $S'_{pre}$ is also a separator for $\tau$, $w(S'_{pre}) \geq w(S)$. Since all separators for $\tau$ contain $U_{\tau} \cap V_{\tau}$, removing the edges increases the weight of all separators for $\tau$ by at least ${\beta}x_{\cap}$ which implies that $w(S') \geq w(S'_{pre}) + {\beta}x_{\cap}$. Putting these pieces together, $w(S') \geq w(S) + {\beta}x_{\cap}$, as needed.
    \item $\Iso(\tau') = \Iso(\tau)$.
\end{enumerate}
Putting these pieces together, we have that 
\begingroup
\allowdisplaybreaks
\begin{align*}
    \slack(\tau') - \slack(\tau) &\geq (\frac{1}{2}-\al){\beta}x - \frac{(1-\alpha)\beta}{2} (x + x_{\cap}) -\frac{\gam - \al\beta}{2}\left(x - x_{\cap}\right) + \frac{\beta}{2}x_{\cap} + {\gamma}x\\
    &= \frac{\gamma}{2}(x + x_{\cap})\,.
\end{align*}
\endgroup
as needed.
\end{proof}

\subsection{Final slack lower bound}

\slackTheorem*
\begin{proof}
Examining
\cref{cor:pmvs-slack}, \cref{cor:intersection-slack}, \cref{thm:adding-indicators-slack},
\begin{align*}
    \slack(&\tau) \geq \frac{\gam - \al \beta}{2}\left(|E_{tot}(\tau)| - \frac{|E(U_\tau)| + |E(V_\tau)|}{2}\right)\\
+&\frac{1-\alpha}{2}\left(\text{\# of vertices in }(U_{\tau} \cup V_{\tau}) \setminus (U_{\tau} \cap V_{\tau})\text{ not incident to }E_{tot}(\tau)\right)\,.
\end{align*}
Note that the latter term is initially zero and after the \textbf{Finding PMVS subroutine}, since otherwise a degree-0 vertex in $(U_{\tau'} \cup V_{\tau'}) \setminus (U_{\tau'} \cap V_{\tau'})$ could be removed
to reduce the size of the separator.

Every vertex in $V_{tot}(\tau) \setminus (U_\tau \cup V_\tau)$ is not isolated and so is incident to an edge in $E_{tot}(\tau) \setminus (E(U_\tau) \cup E(V_\tau))$. On the other hand, the vertices of $U_\tau \cup V_\tau$
which are not incident to an edge of $E_{tot}(\tau)$ are accounted
for by the second term. In summary,
\begin{align*}
    &\frac{\gam - \al \beta}{2}\left(|E_{tot}(\tau)| - \frac{|E(U_\tau)| + |E(V_\tau)|}{2}\right)\\
&+\frac{1-\alpha}{2}\left(\text{\# of vertices in }(U_{\tau} \cup V_{\tau}) \setminus (U_{\tau} \cap V_{\tau})\text{ not incident to }E_{tot}(\tau)\right)\\
\geq \;& \frac{\gam - \al \beta}{4}\left(|E_{tot}(\tau)| - \frac{|E(U_\tau)| + |E(V_\tau)|}{2}\right)\\
&+\min\left\{1-\alpha, \frac{\gam - \al\beta}{8}\right\}\left(|V_{tot}(\tau)| - \frac{|U_\tau| + |V_\tau|}{2}\right)
\end{align*}
as needed.
\end{proof}

\section{Conclusion}

In this work, we showed Sum-of-Squares lower bounds for Densest $k$-Subgraph.
Our results lend strength to the conjecture that Densest $k$-Subgraph is truly a hard problem in the predicted ``hard'' parameter regime.
Our results are in line with the log-density framework 
for Densest-$k$-Subgraph, complementing the extraordinary work of \cite{BCCFV10} from over a decade ago.

Our work provides a formal lower bound against a concrete class of algorithms for Densest $k$-Subgraph. For the optimistic algorithm designer that wishes to solve Densest $k$-Subgraph, what kind of algorithms could circumvent our lower bound?
First, one could try to modify the constraints or objective of the semi-definite program.
For example, ``mismatching'' the size of the hidden subgraph may be helpful for the related Planted Clique problem~\cite{angelini2021mismatching}.
Our proof does not formally rule out non-standard SDP-based algorithms, although we believe it is likely that our proof
could be modified into a lower bound against other SDPs.
Second, algebraic approaches based on finite fields, Gaussian elimination, or lattice-based methods 
are not captured by Sum-of-Squares reasoning \cite{zadik2022lattice}. However, these techniques typically require a rigid ``noise-free'' structure in the problem which isn't present in Densest $k$-Subgraph, so such an algorithm would be unexpected.

There are some technical limitations to our work, which are also present in almost all existing SoS lower bounds.
Technical improvements such as improving the SoS degree from $n^{\eps}$ to $\tilde{\Omega}(k)$, or tightening the slack $\gam - \al\beta$ seem out of reach for our current techniques.
We could also consider the closely related planted model where the size of the planted subgraph is not approximately but exactly $k$.
Our analysis doesn't go through immediately in this setting for technical reasons, which is also the case in most existing SoS lower bounds.
With additional
work, this might be overcome, as Pang~\cite{Pang21} did for Planted Clique.
That said, we believe
that the behavior of SoS is qualitatively the same.

\bibliographystyle{alpha}
{\footnotesize\bibliography{macros, madhur}}

\clearpage

\appendix
\section{Additional Content on Graph Matrices}

\subsection{Proof of \texorpdfstring{\cref{prop:smvs-uniqueness}}{Proposition \ref{prop:smvs-uniqueness}}}
\label{app:graph-matrix}

\leftmostSmvs*

To prove this, we prove a stronger structural characterization
of the leftmost SMVSs of a shape $\al$.
The leftmost SMVS of $\al$ with minimum vertex size is then $S$ in 
the statement of \cref{prop:smvs-structure}.
The intuition for this structural characterization is that the sets $S_i$ are ``subgraphs of weight 0''
that may be freely added to or removed from the ``necessary core'' SMVS $S$.

\begin{proposition}\label{prop:smvs-structure}
    The collection of SMVS of $\al$ which are left of every SMVS has the following structure:
    there are disjoint vertex sets $S$ and $S_1, \dots, S_k$ such that 
    the collection is exactly $S$ unioned with $\bigcup_{i \in I}S_i$ for all subsets $I \subseteq [k]$.
\end{proposition}

\begin{proof}
    First, we prove the existence of a leftmost SMVS.
    Let $S_1, S_2$ be two minimum vertex separators.
    Then we can construct a minimum vertex separator to the left of both of
    them as follows (see \cref{fig:mvs} for a picture). Since this process cannot continue indefinitely,
    it must terminate in a vertex separator which is left of all other vertex separators.
    
    Let $L_1 \subseteq S_1$ be vertices of $S_1$ reachable from $U_\al$
    without passing through $S_2$, and likewise for $L_2 \subseteq S_2$.
    Then we take $S_L \defeq L_1 \cup L_2 \cup (S_1 \cap S_2)$.
    
    To show that $S_L$ is a vertex separator, take a path $P$ from $U_\al$.
    Without loss of generality, $P$ passes through $S_1$ before $S_2$ (or at
    the same time).
    Then $L_1$ (or $S_1 \cap S_2$) blocks $P$.
    
    To show that $S_L$ is minimum, observe that if we perform the
    analogous construction of $S_R$ then $w(S_L) + w(S_R) \leq w(S_1) + W(S_2)$. To see this note that $|S_L| + |S_R| = |S_1| + |S_2|$ and each edge $e$ appears at least as many times in $S_L$ and $S_R$ as it does in $S_1$ and $S_2$. In particular, if $e \in E(S_1)$ and $e \in E(S_2)$ then $e$ is also in both $E(S_L)$ and $E(S_R)$. If $e \in E(S_1) \setminus E(S_2)$ then $e$ must be in either $E(S_L)$ or $E(S_R)$ (note that $e$ cannot go between $L_1$ and $R_1$ as otherwise $S_2$ would not be a separator, see Figure 2). Following similar logic, if $e \in E(S_2) \setminus E(S_1)$ then $e$ must be in either $E(S_L)$ or $E(S_R)$.
    
    Since $S_1$ and $S_2$ are minimum weight vertex separators and $w(S_L) + w(S_R) \leq w(S_1) + w(S_2)$, we must have that $w(S_L) = w(S_R) = w(S_1) = w(S_2)$ so both $S_L$ and $S_R$ must be minimum weight vertex separators as well.
    This finishes the proof of existence of a leftmost SMVS.

\begin{figure}[ht]
    \centering
    \includegraphics[width=0.5\textwidth]{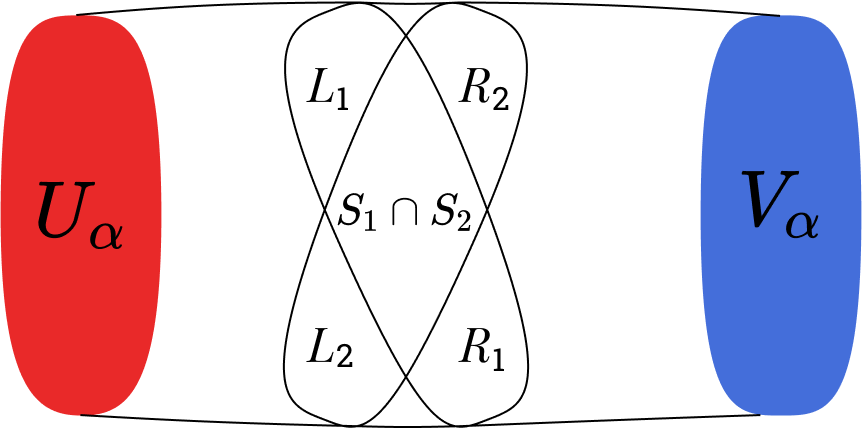}
    \caption{Leftmost and rightmost SMVS}
    \label{fig:mvs}
\end{figure}

\medskip

Next, we prove the structural characterization.
Suppose $S_1, S_2$ are both SMVS which are left of every SMVS.
Since $S_1$ is left of $S_2$, we have $L_2 =\emptyset$, and likewise $L_2 = \emptyset$.
The previous construction now shows that $S_L = S_1 \cap S_2$ is also an SMVS.

Furthermore, we claim that $S_1 \cap S_2$ is left of every SMVS.
Suppose $T$ is another SMVS,
and $P$ is a path from $U_\al$ to $T$.
Since $S_1$ is left of $T$, $P$ passes through $(S_1 \cap S_2) \cup R_1$, and likewise for $(S_1 \cap S_2) \cup R_2$.
Since $S_1$ is left of $S_2$ and vice versa, path $P$ must pass through both sets at the same time,
therefore it passes through $S_1 \cap S_2$.

Therefore, the collection of SMVS which are left of every SMVS
is closed under intersection.
We may now produce the set $S$ as the intersection of the
family. The sets $S_i$ are the refinement of the family under intersection.
It must hold that
$w(S \cup S_i) = w(S)$ for all $i$, since they are both SMVS.
Therefore, there cannot be any edges between $S_i$ and $S_j$,
otherwise $S \cup S_i \cup S_j$ would have smaller weight than $S$.
Therefore, unioning $S$ with any number of the $S_i$
always produces an SMVS with the same weight.
\end{proof}

\subsection{Additional definitions}

\begin{definition}[Automorphism group]
\label{def:aut}
For a shape $\al$, define $\Aut(\al)$ to be the group of graph isomorphisms of $(V(\al), E(\al))$. Equivalently, $\Aut(\al)$ is
the stabilizer subgroup of $S_n$ of any ribbon $R$ with shape $\al$.
\end{definition}

\begin{proposition}\label{prop:automorphism-sum}
$\displaystyle\abs{\Aut(\al)}\sum_{\text{ribbons $R$ of shape }\alpha} \mM_R = \sum_{\text{injective }\varphi : V(\al) \to [n]} \mM_{\varphi(\alpha)}$
\end{proposition}

We would like to enforce that all SMVS of a given shape are isomorphic graphs.
This can be achieved by adding an infinitesimally small quantity to $w(S)$ that breaks equality depending on the isomorphism class of $S$. Equivalently, we redefine the SMVS to minimize ($w(S)$, $S$) lexicographically using an arbitrary and fixed total order of all graphs $S$.
Either way, we will use the following proposition. 

\begin{proposition}
\label{prop:smvs-isomorphic}
    If $\tau$ is a middle shape, we may assume that
    $U_{\tau}$ and $V_\tau$ are isomorphic as graphs.
\end{proposition}

\section{Densest subgraph weight function}
\label{app:true-pmvs}

As pointed out in \cref{sec:pmvs}, whether or not edges are present inside a shape's separator affects the norm bound.
We may find the norm bound
using the following weight function for $S$ being a graph on $[n]$,

\[w_{densest}(S) = |S| - \log_n(1/p) \abs{E(S) \cap E(G)} + \log_n(1/p)\abs{E(S) \cap \overline{E(G)}} \,.\]

Letting $S_{min}$ be the minimizer of $w_{densest}$ over separators of $A_R$ and $B_R$ in a shape $\al$,
then we would like:
\[ \norm{\mM_\al} \leq \widetilde{O}\left(n^{\frac{|V(\al)| - w_{densest}(S_{min})}{2}}\right) = \widetilde{O}\left(n^{\frac{|V(\al)| - \abs{S_{min}}}{2}} p^{\frac{\abs{E(S_{min}) \cap \overline{E(G)}} - \abs{E(S_{min}) \cap E(G)}}{2}}\right)\,.\]

There is a problem with the above ``formula''. Observe that the definition of $w_{densest}(S)$
depends on the instantiation of $G$ restricted to $S$. Hence two ribbons of the same shape may not have the same
MVS, and it isn't possible to define the MVS with respect to $w_{densest}$ on a shape level instead of a ribbon level.
This is also true of the PMVS as described in \cref{sec:pmvs}, although there we handle it by using edge and missing edge indicators (\cref{rmk:graph-dependence}).
We might like to use edge indicators in a similar way to define $w_{densest}$. However, it is crucial in our analysis that the edge indicators are restricted to be in only the middle part, where they can be controlled, and are not allowed in the left and right parts.

The weight function $w_{densest}$ has the property that if it identifies an MVS
with missing edge evaluations, then removing those Fourier characters still has the same MVS (and hence this is a positive ribbon that could potentially be used to norm-dominate the original ribbon). This is stated as follows.
\begin{lemma}
    If $A_R$ is a minimizer of $w_{densest}$ in ribbon $R$, and $A_R$ contains a missing edge $\chi_e$,
    then $A_R$ is a minimizer of $w_{densest}$ in $R' = (A_R, B_R, E(R) \setminus \{e\})$. The same holds for $B_R$. 
\end{lemma}
\begin{proof}
    The connectivity of $A_R, B_R$ and $A_{R'}, B_{R'}$ is exactly the same,
    therefore the collection of left/right separators is the same. By
    removing the edge, $w_{densest}(A_R)$ decreases, and therefore it continues
    to be a minimizer.
\end{proof}

\section{Requirements for Combinatorial Adjustment Terms}

In order for the PSDness proof to go through, we need that for some $\epsilon > 0$, the following bounds hold:
\begin{bound}\label{assumptions}\
\begin{enumerate}
\item $\epsilon \leq \al \leq \frac{1}{2}$
\item $\gamma - {\alpha}{\beta} \geq \epsilon$
\item $\beta - \gamma \geq \epsilon$
\item $\log_n(D_V) \leq \frac{\epsilon}{20}$
\item $D_V \geq \frac{100}{\eps} \dsos$
\item For all proper middle shapes $\tau$, $\slack(\tau) \geq \eps\left(|E(\tau)| - \frac{|E(U_{\tau})| + |E(V_{\tau})|}{2} + |V(\tau)| - \frac{|U_{\tau}| + |V_{\tau}|}{2}\right)$. More generally, for any middle shape $\tau_P$ resulting from a sequence of interaction patterns,
\[
\slack(\tau_P) \geq \eps\left(|E_{tot}(\tau_P)| - \frac{|E(U_{\tau_P})| + |E(V_{\tau_P})|}{2} + |V_{tot}(\tau_P)| - \frac{|U_{\tau_P}| + |V_{\tau_P}|}{2}\right)
\]
\item $B_{adjust}(\alpha)$ and $c(\alpha)$ are at most $n^{\frac{\epsilon}{100}(|E(\alpha)| - \frac{|E(U_\alpha)| +|E(V_{\alpha})|}{2} + |V(\alpha)| - \frac{|U_{\alpha}| + |V_{\alpha}|}{2})}$
\item For an interaction pattern $P$ on $\gam, \tau, \gam'^\T$, $c(P)$ and $N(P)$ are at most \[n^{\frac{\epsilon}{100}\left(|E(\gam)| - \frac{|E(U_\gam)| +|E(V_{\gam})|}{2} + |V(\gam )| - \frac{|U_{\gam}| + |V_{\gam}|}{2} + |E(\gam'^\T)| - \frac{|E(U_{\gam'^\T})| +|E(V_{\gam'^\T})|}{2} + |V(\gam'^\T )| - \frac{|U_{\gam'^\T}| + |V_{\gam'^\T}|}{2}\right)}\] 
\end{enumerate}
\end{bound}
The first five bounds are satisfied by the choice of parameters in the main \cref{thm:main}.
The lower bound on the slack was proven in \cref{thm:slack}.
We will verify the bounds on the $c$-functions in \cref{sec:c-functions}.

\begin{remark}
    While we generally think of $\epsilon > 0$ as a fixed constant which is independent of $n$, here it is okay for $\epsilon$ to depend on $n$.
\end{remark}

Probabilistically, we will assume the following events occur in the random graph $G \sim \calG_{n,p}$.
\begin{enumerate}
    \item The norm bounds hold (\cref{cor:normbound})
    \item A small subgraph density bound holds (\cref{lem:density-gnp})
\end{enumerate}
Both of these occur with probability at least $1-n^{-\eta}$ for a tweakable parameter $\eta$. For \cref{thm:main} we may use $\eta = 1$.

\section{Norm bounds}
\label{sec:norm-bounds}

A central part of our analysis is norm bounds for graph matrices. 
Fortunately, these norm bounds were proven in our prior work \cite{JPRTX22}.
For completeness and exposition we reprove them here in a simpler form which can be applied more smoothly in our setting.
In particular, we simplify the lower-order dependence at the cost of having extra $\log(n)$ factors in the base of the exponent.

\begin{theorem}\label{thm:generalnormbound}
For all $\eta > 0$, the following statement holds with probability at least $1 - n^{-\eta}$. 

For all shapes $\alpha$ such that $|V(\alpha)| \leq 3D_V$, $|U_{\alpha}| \leq \dsos$, and $|V_{\alpha}| \leq \dsos$, letting $S$ be an SMVS of $\alpha$ and taking $q = \max{\left\{10D_V^2,\lceil{\dsos\ln(n)}\rceil,\lceil{{\eta}\ln(n)}\rceil\right\}}$,
\[
||\mM_{\alpha}|| \leq 5(6D_Vq)^{|V(\alpha)| - \frac{|U_{\alpha}| + |V_{\alpha}|}{2}}n^{\frac{|V(\alpha)| - w(S) + \Iso(\alpha)}{2}}
\]
\end{theorem}
\begin{proof}
It is easy to reduce the case when $\Iso(\alpha) \neq \emptyset$ to the case when $\Iso(\alpha) = \emptyset$ so it is sufficient to prove the theorem when $\Iso(\alpha) = \emptyset$. To prove this theorem when $\Iso(\alpha) = \emptyset$, we use the following lemma.
\begin{lemma}\label{lem:tracepowercalculation}
For all shapes $\alpha$ and all $q \in \mathbb{N}$, 
\begin{align*}
    \E\left[\tr\left(\left(\mM_{\alpha}\mM_{\alpha}^\T\right)^q\right)\right] \leq (2q|V(\alpha)|)^{2q(|V(\alpha)| - \frac{|U_{\alpha}| + |V_{\alpha}|}{2})}n^{q|V(\alpha)|}\left(\max_{\text{separator } S}{n^{-\frac{|S|}{2}}\left(\sqrt{\frac{1-p}{p}}\right)^{|E(S)|}}\right)^{2q-2}
\end{align*}
\end{lemma}
To use this lemma, we observe that for all $\epsilon > 0$ and all $q \in \mathbb{N}$, by Markov's inequality,
\[
\Pr\left(||\mM_{\alpha}|| > \sqrt[2q]{\frac{\E\left[\tr\left(\left(\mM_{\alpha}\mM_{\alpha}^\T\right)^q\right)\right]}{\epsilon}}\right) \leq \Pr\left(\tr\left(\left(\mM_{\alpha}\mM_{\alpha}^\T\right)^q\right) > \frac{\E\left[\tr\left(\left(\mM_{\alpha}\mM_{\alpha}^\T\right)^q\right)\right]}{\epsilon}\right) < \epsilon
\]
We want to ensure that with probability $1 - n^{-\eta}$, all of our bounds hold. To do this, we choose $\epsilon$ to be at most $n^{-\eta}$ times the number of shapes we are considering.
\begin{proposition}
For all $D_V \in \mathbb{N}$, there are at most $2^{10{D_V}^2}$ shapes $\alpha$ such that $|V(\alpha)| \leq 3D_V$
\end{proposition}
\begin{proof}
    We can specify each shape $\alpha$ with at most $3D_V$ vertices as follows.
    \begin{enumerate}
        \item For each $j \in [3D_V]$, we specify whether vertex $j$ is in $U_{\alpha} \setminus V_{\alpha}$, $V_{\alpha} \setminus U_{\alpha}$, $U_{\alpha} \cap V_{\alpha}$, $V(\alpha) \setminus (U_{\alpha} \cup V_{\alpha})$, or does not exist at all. There are $5^{3D_V}$ choices for this.
        \item For each $i,j \in [3D_V]$ such that $i < j$, we specify whether there is an edge between vertex $i$ and vertex $j$ (assuming they both exist). There are at most $2^{\binom{3D_V}{2}}$ choices for this.
    \end{enumerate}
This gives a total of ${5^{3D_V}}2^{\binom{3D_V}{2}}$ choices. For $D_V = 1$, ${5^{3D_V}}2^{\binom{3D_V}{2}} = 1000 \leq 1024 = 2^{10{D_V}^2}$ and the ratio $\frac{2^{10{D_V}^2}}{{5^{3D_V}}2^{\binom{3D_V}{2}}}$ grows for larger $D_V$.
\end{proof}
Taking $\epsilon = \frac{1}{2^{10{D_V}^2}n^{\eta}}$ and taking $q = \max{\{10D_V^2,\lceil{\dsos \ln(n)}\rceil,\lceil{{\eta}\ln(n)}\rceil\}}$, we make the following observations:
\begin{enumerate}
    \item $\sqrt[2q]{\frac{1}{\epsilon}} < e$
    \item Letting $S_{min}$ be an SMVS of $\alpha$,
    \[
    n^{q|V(\alpha)|}\left(\max_{\text{separator } S}{n^{-\frac{|S|}{2}}\left(\sqrt{\frac{1-p}{p}}\right)^{|E(S)|}}\right)^{2q-2} \leq n^{w(S_{min})}\left(n^{\frac{|V(\alpha)| - w(S_{min})}{2}}\right)^{2q}
    \]
    and since $w(S_{min}) \leq \dsos$ and $q \geq \dsos \ln(n)$, 
    $\sqrt[2q]{n^{w(S_{min})}} \leq \sqrt{e}$
\end{enumerate}
Putting these pieces together, we have that for all shapes $\alpha$ such that $|V(\alpha)| \leq 3D_V$, with probability at least $1 - \epsilon$,
\[
||\mM_{\alpha}|| \leq 5(6D_Vq)^{|V(\alpha)| - \frac{|U_{\alpha}| + |V_{\alpha}|}{2}}n^{\frac{|V(\alpha)| - w(S_{min})}{2}}
\]
Theorem \ref{thm:generalnormbound} now follows by taking a union bound over all such shapes $\alpha$.
\end{proof}
In our setting, we take $D_V$ so that $D_V \geq \dsos \ln(n)$ and ${\eta}\ln(n) \leq 10D_V^2$ so that we can bound the truncation error appropriately. Thus, in our setting we can take $q = 10D_V^2$.
\begin{corollary}\label{cor:normbound}
For all $D_{V}, \dsos \in \mathbb{N}$ and $\eta > 0$ such that $D_V \geq \dsos \ln(n)$ and ${\eta}\ln(n) \leq 10D_V^2$, the following statement holds with probability at least $1 - n^{-\eta}$.

For all shapes $\alpha$ (allowing isolated vertices but not multiedges or edge indicators) such that $|V(\alpha)| \leq 3D_V$, $|U_{\alpha}| \leq \dsos$, and $|V_{\alpha}| \leq \dsos$, letting $S_{min}$ be an SMVS of $\alpha$,
\[
||\mM_{\alpha}|| \leq 5(60{D_V^3})^{|V(\alpha)| - \frac{|U_{\alpha}| + |V_{\alpha}|}{2}}n^{\frac{|V(\alpha)| - w(S_{min}) + \abs{\Iso(\alpha)}}{2}}
\]
\end{corollary}

Based on this corollary, we make the following definition.
\begin{definition}[$B_{adjust}$]
Given a shape $\alpha$, we define 
\begin{align*}
B_{adjust}(\alpha) &= 5(60{D_V^3})^{|V(\alpha)| - \frac{|U_{\alpha}| + |V_{\alpha}|}{2}}\,.
\end{align*}
\end{definition}
Therefore for $\eta \leq D_V$, with probability at least $1 - n^{-\eta}$, for all of the shapes $\alpha$ which we consider,
recalling the notation $\normapx{\alpha}$ from \cref{sec:psdness},
\[\norm{\mM_{\alpha}} \leq B_{adjust}(\alpha)\normapx{\alpha}\]

We now prove \cref{lem:tracepowercalculation}, which says that 
\begin{align*}
\E\left[\tr\left(\left(\mM_{\alpha}\mM_{\alpha}^\T\right)^q\right)\right] \leq (2q|V(\alpha)|)^{2q(|V(\alpha)| - \frac{|U_{\alpha}| + |V_{\alpha}|}{2})}n^{q|V(\alpha)|}\left(\max_{\text{separator } S}{n^{-\frac{|S|}{2}}\left(\sqrt{\frac{1-p}{p}}\right)^{|E(S)|}}\right)^{2q-2}
\end{align*}
\begin{proof}[Proof of \cref{lem:tracepowercalculation}]
    \begin{definition}
        Define $H(\alpha,2q)$ to be the graph formed as follows.
        \begin{enumerate}
            \item Take the shapes $\alpha_1,\ldots,\alpha_{2q}$ where for all $j \in [q]$, $\alpha_{2j-1}$ is a copy of $\alpha$ and $\alpha_{2j}$ is a copy of $\alpha^\T$.
            \item For all $i \in [2q-1]$, glue $\alpha_{i}$ and $\alpha_{i+1}$ together by setting $V_{\alpha_i} = U_{\alpha_{i+1}}$.Then glue $\alpha_{2q}$ and $\alpha_{1}$ together by setting $V_{\alpha_{2q}} = U_{\alpha_{1}}$.
        \end{enumerate}
    \end{definition}
    When we expand out $\tr\left(\left(\mM_{\alpha}\mM_{\alpha}^\T\right)^q\right)$, we obtain
    \[       
        \tr\left(\left(\mM_{\alpha}\mM_{\alpha}^\T\right)^q\right) = \sum_{\substack{\pi:V(H(\alpha,2q)) \to [n]:\\\pi \text{ is injective on each } \alpha_i}}{\prod_{i=1}^{2q}{\prod_{e \in E(\alpha_i)}{\chi_{\{\pi(e)\}}}}}
    \]
    where if $e = \{u,v\}$ then $\pi(e) = \{\pi(u),\pi(v)\}$.

    We split the maps $\pi$ into equivalence classes based on the following data.

    For each vertex $v \in V(\alpha_j) \setminus U_{\alpha_j}$, we specify whether there exists an $i < j$ and a vertex $u \in V(\alpha_i)$ such that $\pi(v) = \pi(u)$. If so, we specify such a pair $(i,u)$. There are at most $2q|V(\alpha)|$ choices for this. We have that $\sum_{j=1}^{2q}{V(\alpha_j) \setminus U_{\alpha_j}} = 2q(|V(\alpha)| - \frac{|U_{\alpha}| + |V_{\alpha}|}{2})$ so the total number of equivalence classes is at most $(2q|V(\alpha)|)^{2q(|V(\alpha)| - \frac{|U_{\alpha}| + |V_{\alpha}|}{2})}$.

    We now analyze the contribution to $\E\left[\tr\left(\left(\mM_{\alpha}\mM_{\alpha}^\T\right)^q\right)\right]$ from each equivalence class.
    
    \begin{definition}
    For each $j \in [2,2q-1]$, let $S_j$ be the set of vertices $v \in \alpha_j$ such that there exists a $i < j$ and $u \in V(\alpha_i)$ such that $\pi(u) = \pi(v)$ and there exists a $k > j$ and a vertex $w \in V(\alpha_k)$ such that $\pi(v) = \pi(w)$. Note that we may take $u = v$ if $v \in U_{\alpha_{j}} = V_{\alpha_{j-1}}$ and we may take $w = v$ if $v \in V_{\alpha_{j}} = U_{\alpha_{j+1}}$.

    \end{definition}
    We now observe that for the terms with nonzero expected value, each $S_j$ must be a vertex separator of $\alpha_j$.
    \begin{proposition}
        If $S_j$ is not a vertex separator of $\alpha_j$ for some $j \in [2,2q-1]$ then \\ $\E\left[\prod_{i=1}^{2q}{\prod_{e \in E(\alpha_i)}{\chi_{\{\pi(e)\}}}}\right] = 0$
        
    \end{proposition}
    \begin{proof}
        Assume $S_j$ is not a vertex separator of $\alpha_j$ for some $j \in [2,2q-1]$ and let $P$ be a path from $U_{\alpha_j}$ to $V_{\alpha_j}$ which does not intersect $S_j$.

        Let $v$ be the last vertex in $P$ such that there is an $i < j$ and a $u \in V(\alpha_i)$ such that $\pi(u) = \pi(v)$. This vertex must exist as the first vertex in $P$ is in $U_{\alpha_j} = V_{\alpha_{j-1}}$ and thus has this property. 
        
        Note that there cannot be a $k > j$ and a vertex $w \in V(\alpha_k)$ such that $\pi(v) = \pi(w)$ as otherwise we would have that $v \in S_j$. This implies that $v \notin V_{\alpha_j}$ so $v$ cannot be the last vertex in $P$.

        Let $v'$ be the vertex after $v$ in $P$ and consider the edge $e = \{v,v'\}$. By the way we chose $v$, there is no $i' < j$ and $u' \in V(\alpha_i)$ such that $\pi(u') = \pi(v')$. As we noted above, there is no $k > j$ and $w \in V(\alpha_k)$ such that $\pi(v) = \pi(w)$. This implies that $\chi_{\{\pi(e)\}}$ only appears once in the product $\prod_{i=1}^{2q}{\prod_{e \in \E(\alpha_i)}{\chi_{\{\pi(e)\}}}}$ so $\E\left[\prod_{i=1}^{2q}{\prod_{e \in E(\alpha_i)}{\chi_{\{\pi(e)\}}}}\right] = 0$, as needed.
    \end{proof}
    
    We now bound the contribution to $\E\left[\tr\left(\left(\mM_{\alpha}\mM_{\alpha}^\T\right)^q\right)\right]$ from an equivalence class as follows. Let $S_1 = S_{2q} = \emptyset$. For each $j$,
    \begin{enumerate}
        \item For each vertex $v \in V(\alpha_j) \setminus S_j$, we assign a factor of $\sqrt{n}$ to $v$. 
        
        Note that for each distinct vertex in $\pi(V(H(\alpha,2q)))$, this assigns a factor of $\sqrt{n}$ to this vertex for the first and last time it appears which gives a total factor of $n$.
        \item For each edge $e \in E(S_j)$, we assign a factor of $\sqrt{\frac{1-p}{p}}$ to $e$. 
        
        Note that for each distinct edge $\{x,y\}$ in the multiset $\pi(E(H(\alpha,2q)))$, letting $m_{\{x,y\}}$ be the number of times the edge $\{x,y\}$ appears in $\pi(E(H(\alpha,2q)))$, 
        \[
        m_{\{x,y\}} \leq 2 + |\{j \in [2q]: \{x,y\} \in \pi(S_j)\}|
        \]
        Thus, this assigns a factor of at least 
        \[
        \left(\sqrt{\frac{1-p}{p}}\right)^{m_{\{x,y\}}-2} \geq \E\left[\chi_{\{x,y\}}^{m_{\{x,y\}}}\right]
        \]
        to the edge $\{x,y\}$.
    \end{enumerate}
    This gives an upper bound of 
    \[ 
    n^{q|V(\alpha)|}\prod_{j=2}^{2q-1}{\left(n^{-\frac{|S_j|}{2}}\left(\sqrt{\frac{1-p}{p}}\right)^{|E(S_j)|}\right)} \leq n^{q|V(\alpha)|}\left(\max_{\text{separator } S}{n^{-\frac{|S|}{2}}\left(\sqrt{\frac{1-p}{p}}\right)^{|E(S)|}}\right)^{2q-2}
    \]
    Since there are at most $(2q|V(\alpha)|)^{2q(|V(\alpha)| - \frac{|U_{\alpha}| + |V_{\alpha}|}{2})}$ equivalence classes, 
    \[
        \E\left[\tr\left(\left(\mM_{\alpha}\mM_{\alpha}^\T\right)^q\right)\right] \leq (2q|V(\alpha)|)^{2q(|V(\alpha)| - \frac{|U_{\alpha}| + |V_{\alpha}|}{2})}n^{q|V(\alpha)|}\left(\max_{\text{separator } S}{n^{-\frac{|S|}{2}}\left(\sqrt{\frac{1-p}{p}}\right)^{|E(S)|}}\right)^{2q-2}
    \]
    as needed.
\end{proof}

\subsection{Conditioning}
Subgraphs of $G$ which are excessively dense are highly unlikely to occur and will need to be ``conditioned away''.\
This part of the analysis is only needed for controlling the truncation error and calculating $\pE[1]$ and is not needed for analysis of the approximate PSD factorization.

For a small graph $S$ on $[n]$, $w(S)$ is approximately equal
to the logarithm of the expected number of copies of $S$ in $G \sim \calG_{n,p}$.
Therefore, we expect that all subgraphs satisfy $w(S) \geq 0$.
We show that this holds approximately in the following proposition.

\begin{proposition}\label{lem:density-gnp}
For all $\eta > 0$ and $D \in \mathbb{N}$ such that $4\log_n(D) < \beta$, the probability that there is a subgraph $H$ of $G \sim \calG_{n, n^{-\beta}}$ such that $|V(H)| \leq D$ and $|E(H)| > \frac{|V(H)| + \eta}{\beta - 2\log_n(D)}$
is at most $n^{-\eta}$.
Equivalently, with probability at least $1 - n^{-\eta}$ all subgraphs $H$ of $G$ such that $|V(H)| \leq D$ satisfy
\begin{align*}
w(H) &\geq -\eta - 2 \abs{E(H)} \log_n(D)
\end{align*}
\end{proposition}

\begin{proof}
We use the first moment method. For all $a$ and $b$ such that $a \leq D$, the expected number of subgraphs with exactly $a$ vertices and at least $b$ edges is at most
\[
   \binom{n}{a}\binom{\frac{a(a-1)}{2}}{b}p^{b} \leq \frac{n^{a}\left(\frac{a^2}{2}\right)^b}{a!b!}n^{-b\beta} \leq \frac{n^{-\eta}}{a!}n^{a + \eta +  (2\log_{n}(D) - \beta)b}
\]
If $b \geq \frac{a + \eta}{\beta - 2\log_{n}(D)}$ then $n^{a + \eta +  (2\log_{n}(D) - \beta)b} \leq 1$ so the expected number of graphs with exactly $a$ vertices and at least $b$ edges is at most $\frac{n^{-\eta}}{a!}$. Using Markov's inequality, this implies that the probability that there is a subgraph with exactly $a$ vertices and more than $\frac{a + \eta}{\beta - 2\log_n(D)}$ edges is at most $\frac{n^{-\eta}}{a!}$. Taking a union bound over all $a \in [2,D]$ (the cases when $a \leq 1$ are trivial), the probability of having a subgraph $H$ with at most $D$ vertices and at least $\frac{|V(H)| + \eta}{\beta - 2
\log_{n}(D)}$ edges
is at most $n^{-\eta}$.
\end{proof}

\begin{theorem}\label{thm:conditionednormbounds}
Conditioned on $G$ having no subgraphs $H$ such that $|V(H)| \leq D$ and $|E(H)| > \frac{|V(H)| + \eta}{\beta - 2\log_n(D)}$ and the norm bounds holding,
\begin{enumerate}
    \item For all shapes $\alpha$ such that $|V(\alpha)| \leq D$,
\[
\lambda_{\alpha}||\mM_{\alpha}|| \leq 2B_{adjust}(\alpha)n^{(1 - \alpha)\frac{|U_{\alpha}| + |V_{\alpha}|}{2} + \eta - (\gamma - {\alpha}{\beta} - 3\log_n(D))|E(\alpha)|}
\]
\item For all left shapes $\sigma$ such that $|V(\sigma)| \leq D$ and $|U_{\sigma}| \leq \dsos$, 
\[
\lambda_{\sigma}||\mM_{\sigma}|| \leq 2B_{adjust}(\sigma)n^{(1 - \alpha)\dsos + \eta + (\frac{\beta}{2} - {\alpha}\beta)|E(V_{\sigma})| - (\gamma - {\alpha}{\beta} - 3\log_n(D))|E(\sigma)|}
\]
Note that this may be stronger than the first bound when $|V_{\sigma}|>\dsos$.
\end{enumerate}
\end{theorem}
\begin{proof}
We start with the first statement. Given a shape $\alpha$, let $S^{\emptyset}_{\alpha}$ be a set of vertices of $\alpha$ (not necessarily a separator) which minimizes $w(S^{\emptyset}_{\alpha}) = |S^{\emptyset}_{\alpha}| - {\beta}E(S^{\emptyset}_{\alpha})$. In other words, $S^{\emptyset}_{\alpha}$ determines the norm bound of the scalar $\one^\T \mM_\al \one$, which is equivalently the shape obtained from $\alpha$ by setting the matrix indices $U_{\alpha}$ and $V_{\alpha}$ to be $\emptyset$.

If $w(S^{\emptyset}_{\alpha}) \geq 0$, letting $S$ be the SMVS of $\alpha$, we have that $w(\alpha) \geq w(S) \geq 0$ so we have the following bounds
\begin{enumerate}
    \item 
    \begin{align*}
        \lambda_{\alpha} &= n^{(\alpha - 1)(|V(\alpha)| - \frac{|U_{\alpha}| + |V_{\alpha}|}{2}) + (\frac{\beta}{2} - \gamma)|E(\alpha)|} \\
        &= n^{(1 - \alpha)\frac{|U_{\alpha}| + |V_{\alpha}|}{2} - \frac{|V(\alpha)|}{2} - (\frac{1}{2} - \alpha)w(\alpha) - (\gamma - {\alpha}{\beta})|E(\alpha)|}\\
        &\leq n^{(1 - \alpha)\frac{|U_{\alpha}| + |V_{\alpha}|}{2} - \frac{|V(\alpha)|}{2} - (\gamma - {\alpha}{\beta})|E(\alpha)|}
     \end{align*}
    \item $||\mM_{\alpha}|| \leq B_{adjust}(\alpha)n^{\frac{|V(\alpha)| - w(S)}{2}} \leq B_{adjust}(\alpha)n^{\frac{|V(\alpha)|}{2}}$
\end{enumerate}
Putting these bounds together, 
\[
\lambda_{\alpha}||\mM_{\alpha}|| \leq B_{adjust}(\alpha)n^{(1 - \alpha)\frac{|U_{\alpha}| + |V_{\alpha}|}{2} - (\gamma - {\alpha}{\beta})|E(\alpha)|}
\]

We now analyze the case when $w(S^{\emptyset}_{\alpha}) < 0$. If $
|E(S^{\emptyset}_{\alpha})| > \frac{|S^{\emptyset}_{\alpha}| + \eta}{\beta - 2\log_{n}(D)}$, we apply the conditioning argument (Lemma 6.30) from \cite{JPRTX22} which allows us to reduce $\mM_{\alpha}$ to a linear combination of shapes with fewer edges before applying our norm bounds.
\begin{lemma}\label{cor:conditioning-factorization}
Given a set of edges $E$, if we know that not all of the edges of $E$ are in $E(G)$ then 
\[
\chi_E = -\sum_{E' \subseteq E: E' \neq E}{\left(-\sqrt{\frac{p}{1-p}}\right)^{|E| - |E'|}\chi_{E'}}
\]
\end{lemma}
If $
|E(S^{\emptyset}_{\alpha})| > \frac{|S^{\emptyset}_{\alpha}| + \eta}{\beta - 2\log_{n}(D)}$, we apply \cref{cor:conditioning-factorization} to $S^{\emptyset}_{\alpha}$. In particular, we apply the following step repeatedly:
\begin{enumerate}
\item Specify an edge which is removed from $E(S^{\emptyset}_{\alpha})$ and specify whether there is at least one more edge which is removed from this application of \cref{cor:conditioning-factorization} or we are done with this application of \cref{cor:conditioning-factorization}. Note that there are at most $2\binom{D}{2} \leq D^2$ possibilities for this.
\end{enumerate}
If there is still a subset $S \subseteq S^{\emptyset}_{\alpha}$ such that $
|E(S)| > \frac{|S| + \eta}{\beta - 2\log_{n}(D)}$, then we apply \cref{cor:conditioning-factorization} to $S$ and repeat this process. Otherwise, we stop.

For each $j \in \mathbb{N} \cup \{0\}$, this gives a total of at most $D^{2j}$ terms where $j$ edges were removed. Each such term comes with a factor of $\left(\sqrt{\frac{p}{1-p}}\right)^{j}$ for the removed edges.

We have the following bounds:
\begin{enumerate}
    \item $\lambda_{\alpha} = n^{(\alpha - 1)(|V(\alpha)| - \frac{|U_{\alpha}| + |V_{\alpha}|}{2}) + (\frac{\beta}{2} - \gamma)|E(\alpha)|} = n^{(1 - \alpha)\frac{|U_{\alpha}| + |V_{\alpha}|}{2} - \frac{|V(\alpha)|}{2} - (\frac{1}{2} - \alpha)w(\alpha) - (\gamma - {\alpha}{\beta})|E(\alpha)|}$
    \item For each of the resulting terms, letting $j$ be the number of edges which were removed and letting $I$ be the set of isolated vertices after the edges are removed, 
    \[
    |I| \leq (\beta - 2\log_n(D))j + w(S^{\emptyset}_{\alpha}) + 2\log_n(D)|E(S^{\emptyset}_{\alpha})| + \eta\\
    \]
    To see this, observe that $I \subseteq S^{\emptyset}_{\alpha}$ as only edges in $S^{\emptyset}_{\alpha}$ can be removed. Now consider the set $S$ obtained by deleting the removed edges and all vertices in $I$ from $S^{\emptyset}_{\alpha}$. We have that 
    \[
    |E(S^{\emptyset}_{\alpha})| - j = |E(S)| \leq \frac{|S| + \eta}{\beta - 2\log_n(D)} = \frac{|S^{\emptyset}_{\alpha}| + \eta - |I|}{\beta - 2\log_n(D)}
    \]
    which implies that 
    \begin{align*}
    |I| &\leq |S^{\emptyset}_{\alpha}| + \eta -(\beta - 2\log_n(D))(|E(S^{\emptyset}_{\alpha})| - j) \\
    &= w(S^{\emptyset}_{\alpha}) + 2\log_n(D)|E(S^{\emptyset}_{\alpha})| + \eta + (\beta - 2\log_n(D))j
    \end{align*}
   
    \item For each of the resulting terms $\beta$, letting $S'$ be the SMVS after the edges are removed, $w(S') \geq -2\log_n(D)|E(S^{\emptyset}_{\alpha})| - \eta$. We can show this as follows. Let $S'' = S' \cap S^{\emptyset}_{\alpha}$ and observe that $w(S'') \leq w(S')$. To see this, observe that since none of the edges in $E(S') \setminus E(S'')$ are removed, $w(S') - w(S'')$ is unaffected by removing the edges. Now note that before the edges are removed, $w(S') - w(S'') \geq w(S^{\emptyset}_{\alpha} \cup (S' \setminus S'')) - w(S^{\emptyset}_{\alpha}) \geq 0$. 
    
    Since $S'' \subseteq S^{\emptyset}_{\alpha}$, we have that after the edges are removed, $
    |E(S'')| \leq \frac{|S''| + \eta}{\beta - 2\log_n(D)}$ which implies that $w(S'') \geq -2\log_n(D)|E(S'')| - \eta \geq -2\log_n(D)|E(S^{\emptyset}_{\alpha})| - \eta$.
\end{enumerate}
Putting these bounds together, using the bounds that $|E(S^{\emptyset}_{\alpha})| \leq |E(\alpha)|$ and $w(S^{\emptyset}_{\alpha}) \leq \min{\{0,w(\alpha)\}}$, and observing that removing vertices and/or edges from a shape $\alpha$ can only reduce $B_{adjust}(\alpha)$, we have that 
\begin{align*}
\frac{\lambda_{\alpha}\norm{\mM_{\alpha}}}{B_{adjust}(\alpha)} &\leq \lambda_{\alpha}\sum_{j = 0}^{|E(\alpha)|}{\left(D^2\sqrt{\frac{p}{1-p}}\right)^{j}}n^{\frac{|V(\alpha)| + |I| - w(S')}{2}}\\
&\leq n^{(1 - \alpha)\frac{|U_{\alpha}| + |V_{\alpha}|}{2} - \frac{|V(\alpha)|}{2} - (\frac{1}{2} - \alpha)w(\alpha) - (\gamma - {\alpha}{\beta})|E(\alpha)|}n^{\frac{|V(\alpha)| + 2\log_n(D)|E(S^{\emptyset}_{\alpha})| + \eta}{2}}\\
&\sum_{j = 0}^{|E(\alpha)|}{\left(D^2\sqrt{\frac{p}{1-p}}\right)^{j}n^{\frac{(\beta - 2\log_n(D))j + w(S^{\emptyset}_{\alpha}) + 2\log_n(D)|E(S^{\emptyset}_{\alpha})| + \eta}{2}}} \\
&\leq n^{(1 - \alpha)\frac{|U_{\alpha}| + |V_{\alpha}|}{2} - (\gamma - {\alpha}{\beta})|E(\alpha)|}\sum_{j = 0}^{|E(\alpha)|}{\left(D^2\sqrt{\frac{p}{1-p}}\right)^{j}}n^{\frac{({\beta} - 2\log_n(D))j}{2} + 2\log_n(D)|E(\alpha)| + \eta} \\
&\leq 2n^{(1 - \alpha)\frac{|U_{\alpha}| + |V_{\alpha}|}{2} + \eta - (\gamma - {\alpha}{\beta} - 3\log_n(D))|E(\alpha)|}
\end{align*}
To prove the second statement, we split into cases based on whether $w(V_{\sigma})$ is non-negative. For the case where $w(V_{\sigma}) < 0$, we follow the same procedure as before. We let $S^{\emptyset}_{\sigma}$ be a set of vertices of $\sigma$ which minimizes $w(S^{\emptyset}_{\sigma})$. As long as there is a subset $S \subseteq S^{\emptyset}_{\sigma}$ such that $
|E(S)| > \frac{|S| + \eta}{\beta - 2\log_{n}(D)}$, we apply \cref{cor:conditioning-factorization} to $S$. For each $j \in \mathbb{N} \cup \{0\}$, this gives a total of at most $D^{2j}$ terms where $j$ edges were removed where each such term comes with a factor of $\left(\sqrt{\frac{p}{1-p}}\right)^{j}$ for the removed edges. For each such term, letting $I$ be the set of isolated vertices and letting $S'$ be the new SMVS, we have the same bounds as before:
\begin{enumerate}
    \item     
    $|I| \leq (\beta - 2\log_n(D))j + w(S^{\emptyset}_{\sigma}) + 2\log_n(D)|E(S^{\emptyset}_{\sigma})| + \eta
    $
    \item $w(S') \geq -2\log_n(D)|E(S^{\emptyset}_{\alpha})| - \eta$
\end{enumerate}
To bound $\lambda_{\sigma}$, we observe that
\begin{align*}
\lambda_{\sigma} &= n^{(1 - \alpha)\frac{|U_{\sigma}| + |V_{\sigma}|}{2} - \frac{|V(\sigma)|}{2} - (\frac{1}{2} - \alpha)w(\sigma) - (\gamma - {\alpha}{\beta})|E(\sigma)|}\\
&= n^{(1 - \alpha)\frac{|U_{\sigma}|}{2} - \frac{|V(\sigma)|}{2} + \frac{1 - \alpha}{2}w(V_{\sigma}) + (\frac{\beta}{2} - {\alpha}\beta)|E(V_{\sigma})| - (\frac{1}{2} - \alpha)w(\sigma) - (\gamma - {\alpha}{\beta})|E(\sigma)|}
\end{align*}
Using the bound that $w(V_{\sigma}) \leq |U_{\sigma}| + w(S^{\emptyset}_{\sigma}) \leq \dsos + w(S^{\emptyset}_{\sigma})$, we have that
\[
\lambda_{\sigma} \leq n^{(1 - \alpha)\dsos - \frac{|V(\sigma)|}{2} + \frac{1 - \alpha}{2}w(S^{\emptyset}_{\sigma}) + (\frac{\beta}{2} - {\alpha}\beta)|E(V_{\sigma})| - (\frac{1}{2} - \alpha)w(\sigma) - (\gamma - {\alpha}{\beta})|E(\sigma)|}
\]

Putting these bounds together, using the bounds that $|E(S^{\emptyset}_{\sigma})| \leq |E(\sigma)|$ and $w(S^{\emptyset}_{\sigma}) \leq \min{\{0,w(\sigma)\}}$, and observing that removing vertices and/or edges from a shape $\sigma$ can only reduce $B_{adjust}(\sigma)$, we have that 
\begin{align*}
\frac{\lambda_{\sigma}||\mM_{\sigma}||}{B_{adjust}(\sigma)} &\leq \lambda_{\sigma}\sum_{j = 0}^{|E(\sigma)|}{\left(D^2\sqrt{\frac{p}{1-p}}\right)^{j}}n^{\frac{|V(\sigma)| + |I| - w(S')}{2}}\\
&\leq n^{(1 - \alpha)\dsos - \frac{|V(\sigma)|}{2} + \frac{1 - \alpha}{2}w(S^{\emptyset}_{\sigma}) + (\frac{\beta}{2} - {\alpha}\beta)|E(V_{\sigma})| - (\frac{1}{2} - \alpha)w(\sigma) - (\gamma - {\alpha}{\beta})|E(\sigma)|}n^{\frac{|V(\sigma)| + 2\log_n(D)|E(S^{\emptyset}_{\sigma})| + \eta}{2}}\\
&\sum_{j = 0}^{|E(\sigma)|}{\left(D^2\sqrt{\frac{p}{1-p}}\right)^{j}n^{\frac{(\beta - 2\log_n(D))j + w(S^{\emptyset}_{\sigma}) + 2\log_n(D)|E(S^{\emptyset}_{\sigma})| + \eta}{2}}} \\
&\leq n^{(1 - \alpha)\dsos + (\frac{\beta}{2} - {\alpha}\beta)|E(V_{\sigma})| - (\gamma - {\alpha}{\beta})|E(\sigma)|}\sum_{j = 0}^{|E(\sigma)|}{\left(D^2\sqrt{\frac{p}{1-p}}\right)^{j}}n^{\frac{({\beta} - 2\log_n(D))j}{2} + 2\log_n(D)|E(\sigma)| + \eta} \\
&\leq 2n^{(1 - \alpha)\dsos + \eta + (\frac{\beta}{2} - {\alpha}\beta)|E(V_{\sigma})| - (\gamma - {\alpha}{\beta} - 3\log_n(D))|E(\sigma)|}
\end{align*}
For the case when $w(V_{\sigma}) \geq 0$, we again observe that 
\[
\lambda_{\sigma} = n^{(1 - \alpha)\frac{|U_{\sigma}|}{2} - \frac{|V(\sigma)|}{2} + \frac{1 - \alpha}{2}w(V_{\sigma}) + (\frac{\beta}{2} - {\alpha}\beta)|E(V_{\sigma})| - (\frac{1}{2} - \alpha)w(\sigma) - (\gamma - {\alpha}{\beta})|E(\sigma)|}
\]
Since $w(\sigma) \geq w(V_{\sigma}) \geq 0$ and $w(V_{\sigma}) \leq |U_{\sigma}| \leq \dsos$ ,
\[
\lambda_{\sigma} \leq n^{(1 - \alpha)\dsos - \frac{|V(\sigma)|}{2} + (\frac{\beta}{2} - {\alpha}\beta)|E(V_{\sigma})| - (\gamma - {\alpha}{\beta})|E(\sigma)|}
\]
Since $||\mM_{\sigma}|| \leq B_{adjust}(\sigma)n^{\frac{|V(\sigma)| - w(S)}{2}} \leq B_{adjust}(\sigma)n^{\frac{|V(\sigma)|}{2}}$, we have that 
\[
\lambda_{\sigma}||\mM_{\sigma}|| \leq B_{adjust}(\sigma)n^{(1 - \alpha)\dsos + \eta + (\frac{\beta}{2} - {\alpha}\beta)|E(V_{\sigma})| - (\gamma - {\alpha}{\beta})|E(\sigma)|}\,.
\]
\end{proof}

Using the forthcoming notation of \cref{sec:starting-point}, we deduce
the following norm bounds for left shapes.
\begin{corollary}\label{cor: sigmaminusnormbondone}
    With the conditioning, for all $U \in \mathcal{I}_{\text{mid}}$ and all $\sigma \in \calL_{U,\leq D_V}$,
    \[
    \lambda_{\sigma^{-}}||\mM_{\sigma^{-}}|| \leq 2B_{adjust}(\sigma)n^{(1 - \alpha)\dsos + \eta - (\frac{\beta}{2} - 3\log_n(D_V))|E(U)| - (\gamma - {\alpha}{\beta} - 3\log_n(D_V))|E(\sigma) \setminus E(U)|}\,.
    \]
\end{corollary}
\begin{proof}
    By \cref{thm:conditionednormbounds},
    \[
        \lambda_{\sigma}||\mM_{\sigma}|| \leq 2B_{adjust}(\sigma)n^{(1 - \alpha)\dsos + \eta + (\frac{\beta}{2} - {\alpha}\beta)|E(U)| - (\gamma - {\alpha}{\beta} - 3\log_n(D))|E(\sigma)|}
    \]
    We now observe that 
    \[
    \lambda_{\sigma^{-}}||\mM_{\sigma^{-}}|| \leq \left(\sqrt{\frac{p}{1-p}}\right)^{|E(U)|}n^{(\gamma - \frac{\beta}{2})|E(U)|}\lambda_{\sigma}||\mM_{\sigma}|| = n^{(\gamma - \beta)|E(U)|}\lambda_{\sigma}||\mM_{\sigma}||
    \]
    where we remove the indicators from $V_{\sigma}$ for the purpose of bounding $||\mM_{\sigma}||$. 
    Combining these bounds gives the result.
\end{proof}
\begin{corollary}\label{cor:sigmaminusnormboundtwo}
With the conditioning, for all $U \in \mathcal{I}_{mid}$ and all $\sigma \in \calL_{U,\leq D_V}$,
\[
\lambda_{\sigma^{-}}||\mM_{\sigma^{-}}||\sqrt{\lambda_{U}||\mM_{U}||} \leq 2B_{adjust}(\sigma)n^{\dsos + \eta - \frac{\epsilon}{2}\dsos - \frac{\epsilon}{8}|E(\sigma)| - \frac{\epsilon}{8}|V(\sigma)|}
\]
\end{corollary}
\begin{proof}
By \cref{cor: sigmaminusnormbondone},
    \[
    \lambda_{\sigma^{-}}||\mM_{\sigma^{-}}|| \leq 2B_{adjust}(\sigma)n^{(1 - \alpha)\dsos + \eta - (\frac{\beta}{2} - 3\log_n(D_V))|E(U)| - (\gamma - {\alpha}{\beta} - 3\log_n(D_V))|E(\sigma)\setminus E(U)|}
    \]
Since $\sqrt{\lambda_{U}||\mM_{U}||} = n^{(\frac{\beta}{2} - \frac{\gamma}{2})|E(U)|}$, using the fact that $|E(\sigma)| \geq |V(\sigma) \setminus U_{\sigma}|$, we have that 
\begin{align*}  
    \lambda_{\sigma^{-}}||\mM_{\sigma^{-}}||\sqrt{\lambda_{U}||\mM_{U}||} &\leq 2B_{adjust}(\sigma)n^{(1 - \alpha)\dsos + \eta - (\frac{\gamma}{2} - 3\log_n(D_V))|E(U)| - (\gamma - {\alpha}{\beta} - 3\log_n(D_V))|E(\sigma)\setminus E(U)|}\\
    &\leq 2B_{adjust}(\sigma)n^{\dsos + \eta - \frac{\epsilon}{2}\dsos - \frac{\epsilon}{8}|E(\sigma)| - \frac{\epsilon}{8}|V(\sigma)|}
\end{align*}
    as needed.
\end{proof}

\section{Formal Approximate PSD Decomposition}
\label{sec:formal-details}

\subsection{Starting point for the approximate PSD decomposition}
\label{sec:starting-point}

In this section, we show that with high probability, the pseudo-calibrated moment matrix $\matLam$ (formally defined in \cref{sec:pseudocalibration}) is PSD. We do this by giving an approximate PSD factorization for $\matLam$. We will then show that the error is PSD dominated by the terms in this approximate PSD factorization. For this approximate PSD factorization, we use much of the technical framework of \cite{potechin2020machinery} (although we cannot formally apply the machinery there because it does not work well for \emph{sparse} random graphs, nor does it incorporate the PMVS idea).

\begin{definition}[$\calI_{mid}$]
Let $\calI_{mid}$ be the set of shapes of separators of $\calS$. In other words, $\calI_{mid}$ is the set of diagonal shapes $\alpha$ such that $V(\alpha) = U_{\alpha} = V_{\alpha}$ and $|V(\alpha)| \leq D_V$.
\end{definition}

\begin{definition}[$\calL$ and $\calL_U$ and $\calL_{U, {\scriptscriptstyle\leq} D}$]
\label{def:calL}
Let $\calL$ be the set of left shapes in $\calS$.
Given $U\in\calI_{mid}$ and $D \in \mathbb{N}$, we define $\calL_U$ to be the set of all left shapes $\sigma \in \calS$
such that $V_\sig = U$.
The set
$\mathcal{L}_{U, {\ssleq} D}$ also requires $|V(\sigma)| \leq D$.
\end{definition}

\begin{definition}[$\calM$ and $\calM_{U,V}$ and $\calM_{U, V, \ssleq D}$]
\label{def:calM}
Let $\calM$ be the set of middle shapes in $\calS$.
Given $D \in \mathbb{N}$ and $U,V \in \calI_{mid}$ (which may intersect), we define $\mathcal{M}_{U,V}$ to be the set of middle shapes $\tau \in \calS$ such that $U_{\tau} = U$ and $V_{\tau} = V$. The set $\mathcal{M}_{U,V,\ssleq D}$ also requires $|V(\tau)| \leq D$.
\end{definition}
\begin{remark}
Due to the size constraints on shapes in $\calS$, we only have shapes with at most $D_V$ vertices. We will sometimes write $\calL_{U,\leq D_V}$ instead of the equivalent $\calL_U$ when it is relevant to the current section. Since all of the shapes we will consider have an SMVS with weight at most $\dsos$, we only consider $U \in \calI_{mid}$ such that $w(U) \leq \dsos$.
\end{remark}

\begin{definition}[$\calT_{\tau}$ and $\calT_{\tau, \leq D_1, \leq D_2}$]
Given a shape $\tau$ (which may or may not be proper), we define $\mathcal{T}_{\tau}$ to be the set of triples of ribbons $(R_1,R_2,R_3)$ such that 
\begin{enumerate}[(i)]
    \item $R_2$ has shape $\tau$.
    \item $R_1$ is a left ribbon and $R_3$ is a right ribbon.
    \item $R_1,R_2,R_3$ are properly composable.
    \item The edges and edge indicators agree on $B_{R_1} = A_{R_2}$ and $B_{R_2} = A_{R_3}$.
    \item $R_1, R_3$ have no edge indicators outside of $B_{R_1} = A_{R_2}$ and $B_{R_2} = A_{R_3}$.
\end{enumerate}
$\calT_{\tau, \leq D_1, \leq D_2}$ additionally requires that $|V(R_1)| \leq D_1, |V(R_3)| \leq D_2$.
\end{definition}

\begin{definition}[$\calR(\al)$]
Given a shape $\alpha$, we define $\mathcal{R}(\alpha)$ to be the set of ribbons $R$ which have shape $\alpha$. 
\end{definition}

\begin{definition}[Minus abbreviation]
    Given a left ribbon $L$, let $L^- = L \setminus E(B_L)$.
    Given a right ribbon $R$, let $R^- = R \setminus E(A_R)$.
    The notation is defined similarly for shapes.
\end{definition}

\begin{definition}[$U \sim V$]
Given $U, V \in \calI_{mid}$, we write $U \sim V$
if $|U| = |V|$, $|E(U)| = |E(V)|$, $U$ and $V$ have the same edges on the vertex set $U \cap V$, and $U$ and $V$ have the same order on $U \cap V$.
\end{definition}

Starting from the pseudocalibrated formula for $\Lam$
and incorporating this notation, we have the following lemma.
\begin{lemma}
\begin{align*}
\matLam &= \sum_{\substack{U, V \in \calI_{mid}:\\ U \sim V}}
\frac{|U \cap V|!}{(|U|!)^2}
\sum_{\tau \in \calM_{U, V}}\sum_{R_1,R_2,R_3 \in \calT_{\tau, \leq D_V, \leq D_V}}
    \lam_{R_1^- \circ R_2 \circ R_3^-}{\mM_{R_1^- \circ R_2 \circ R_3^-}}\\
    & \quad - \trunc_1
\end{align*}
where $\trunc_1$ is defined in \cref{def:truncation1}.
\end{lemma}
\begin{proof}
Starting from $\matLam = \sum_{R \in \calS} \lam_R \mM_R$,
we apply \cref{prop:decomposition-uniqueness} to factor
$R$ into a left, middle, and right ribbon.
We symmetrize over all choices of the order for the leftmost
SMVS $U$ and the rightmost SMVS $V$ such that the permutation on
$U \cap V$ is the identity permutation.
There are exactly 
$\frac{|U|!|V|!}{|U \cap V|!}$ such choices for the orders.
By \cref{prop:smvs-isomorphic}, $U \sim V$ and $\frac{|U|!|V|!}{|U \cap V|!} = \frac{(|U|!)^2}{|U \cap V|!}$.
Therefore,

\begin{align*}
    \matLam &= 
    \sum_{\substack{U, V \in \calI_{mid}:\\ U \sim V}} \sum_{\substack{\sigma \in \calL_U\\\tau \in \calM_{U,V} \\ \sigma' \in \calL_V:\\ |V(\sig^- \circ \tau \circ (\sig'^\T)^-)| \leq D_V}}
    \sum_{\substack{R_1 \in \calR(\sigma)\\ R_2 \in \calR(\tau)\\R_3 \in \calR(\sigma'^\T):\\R_1, R_2, R_3 \text{ properly composable}}}
    \frac{|U \cap V|!}{(|U|!)^2} \lam_{R_1^- \circ R_2 \circ R_3^-} \mM_{R_1^- \circ R_2 \circ R_3^-}
\end{align*}

The condition $|V(\sig^- \circ \tau \circ (\sig'^\T)^-)| \leq D_V$ arises because the size of the entire ribbon $R = R_1^- \circ R_2 \circ R_3^-$ is
bounded by $D_V$. We would like instead that the sizes of the individual
pieces $R_1, R_2, R_3$ are bounded separately by $D_V$.
The difference between these two consists only of very large ribbons,
which we will bound as a truncation error.

\begin{definition}[$\trunc_1$]\label{def:truncation1}
    \begin{align*}
    \trunc_1 = \sum_{\substack{U, V \in \calI_{mid}:\\ U \sim V}} 
    \frac{|U\cap V|!}{(|U|!)^2}
    \sum_{\substack{\sig \in \calL_{U,\leq D_V}\\ \tau \in \calM_{U, V, \leq D_V}\\ \sig' \in \calL_{V,\leq D_V}:\\ |V(\sig^- \circ \tau \circ (\sig'^\T)^-)| > D_V}}
    \frac{\lam_{\sig^- \circ \tau \circ (\sig'^\T)^-} \mM_{\sig^- \circ \tau \circ (\sig'^\T)^-}}{|\Aut(\sig^- \circ \tau \circ (\sig'^\T)^-)|}
    \end{align*}
\end{definition}
We have grouped $\trunc_1$ into shapes. In terms of ribbons, by \cref{prop:automorphism-sum},
\begin{align*}
    \trunc_1 &= \sum_{\substack{U, V \in \calI_{mid}:\\ U \sim V}} \frac{|U\cap V|!}{(|U|!)^2} \sum_{\substack{\sig \in \calL_U\\ \tau \in \calM_{U, V}\\ \sig' \in \calL_V:\\ |V(\sig^- \circ \tau \circ (\sig'^\T)^-)| > D_V}}
    \sum_{R \in \calR(\sig^- \circ \tau \circ (\sig'^\T)^-)}
    \lam_{R} \mM_{R}\\
    &= \sum_{\substack{U, V \in \calI_{mid}:\\ U \sim V}} \frac{|U\cap V|!}{(|U|!)^2} \sum_{\substack{\sig \in \calL_U\\ \tau \in \calM_{U, V}\\ \sig' \in \calL_V:\\ |V(\sig^- \circ \tau \circ (\sig'^\T)^-)| > D_V}}
    \sum_{\substack{R_1 \in \calR(\sigma)\\ R_2 \in \calR(\tau)\\R_3 \in \calR(\sigma'^\T):\\R_1, R_2, R_3 \text{ properly composable}}}
    \lam_{R_1^- \circ R_2 \circ R_3^-} \mM_{R_1^- \circ R_2 \circ R_3^-}\\
\end{align*}
Continuing the calculation,
\begingroup
\allowdisplaybreaks
\begin{align*}
    \matLam + \trunc_1
    =\;& \sum_{\substack{U, V \in \calI_{mid}:\\ U \sim V}} 
    \frac{|U \cap V|!}{(|U|!)^2}
    \sum_{\substack{\sigma \in \calL_U\\\tau \in \calM_{U,V} \\ \sigma' \in \calL_V}}
    \sum_{\substack{R_1 \in \calR(\sigma)\\ R_2 \in \calR(\tau)\\R_3 \in \calR(\sigma'^\T):\\R_1, R_2, R_3 \text{ properly composable}}}
    \lam_{R_1^- \circ R_2 \circ R_3^-} \mM_{R_1^- \circ R_2 \circ R_3^-}\\
    =\;& \sum_{\substack{U, V \in \calI_{mid}:\\ U \sim V}}
\frac{|U \cap V|!}{(|U|!)^2}
\sum_{\tau \in \calM_{U, V}}\sum_{R_1,R_2,R_3 \in \calT_{\tau, \leq D_V, \leq D_V}}
    \lam_{R_1^- \circ R_2 \circ R_3^-}{\mM_{R_1^- \circ R_2 \circ R_3^-}}
\end{align*}
\endgroup
as desired.
\end{proof}

\subsection{Interaction patterns}
In order to analyze $\matLam$, we use the procedure in \cref{sec:pmvs} as described by \emph{interaction patterns}.
These generalize intersection patterns (\cref{def:intersection-pattern}) to account for the additional ways that shapes can interact in the recursion.

\begin{definition}[PMVS interaction, implicit definition]
    Given composable shapes $\gam, \tau, \gam'^\T$ such that $E(V_\gam) = E(U_\tau)$ and $E(V_\tau) = E(U_{\gam'^\T})$,
    let $\calP^{PMVS}_{\gam, \tau, \gam'^\T}$ be the set of possible choices for one iteration of the
    \textbf{Finding PMVS subroutine} (\cref{sec:ribbonoperations}, note that this includes the \textbf{Removing middle edge indicators operation})
    run on ribbons $R_1, R_2, R_3$ of shapes $\gam, \tau, \gam'^\T$.
\end{definition}

\begin{definition}[Intersection term interaction, implicit definition]
    Given composable shapes $\gam, \tau, \gam'^\T$ such that $E(V_\gam) = E(U_\tau)$ and $E(V_\tau) = E(U_{\gam'^\T})$,
    let $\calP^{intersect}_{\gam, \tau, \gam'^\T}$ be the set of possible choices for
    the \textbf{Intersection term decomposition operation} followed
    by the \textbf{Removing middle edge indicators operation} (\cref{sec:intersectionoperation})
    run on ribbons $R_1, R_2, R_3$ of shapes $\gam, \tau, \gam'^\T$.
\end{definition}

\begin{definition}[Interaction pattern]
    Let $\calP^{interact}_{\gam, \tau, \gam'^\T} = \calP^{PMVS}_{\gam, \tau, \gam'^\T} \cup \calP^{intersect}_{\gam, \tau, \gam'^\T}$.
\end{definition}

\begin{definition}[PMVS interaction, explicit definition]
Given composable shapes $\gamma,\tau,\gamma'^\T$ such that $E(V_\gam) = E(U_\tau)$ and $E(V_\tau) = E(U_{\gam'^\T})$,
a PMVS interaction pattern $P \in \mathcal{P}^{PMVS}_{\gamma,\tau,{\gamma'}^\T}$ consists of:

\begin{enumerate}[(i)]
    \item For each edge in $E(U_{\tau}) \cup E(V_{\tau})$
    which does not yet have an edge indicator, we specify whether the
    edge is given an indicator or is removed,
    \item For each edge in $E(U_{\tau}) \cup E(V_{\tau})$
    which is now in the middle of $\gamma^- \circ \tau \circ ({\gamma'}^\T)^-$, we specify whether the edge is kept or removed when the indicator is removed.
\end{enumerate}
    We furthermore have the structural property of $P$ with $\gamma,\tau,\gamma'^\T$ that after the edges have been removed from $V_{\gamma} = U_{\tau}$ and $V_{\tau} = U_{{\gamma'}^\T}$ in the first step, $U_{\gamma}$ is the leftmost SMVS for $\gamma$ and $V_{{\gamma'}^\T}$ is the rightmost SMVS for $\gamma'^\T$.
\end{definition}

\begin{definition}[Intersection term interaction, explicit definition]
Given composable shapes $\gamma,\tau,\gamma'^\T$ such that $E(V_\gam) = E(U_\tau)$ and $E(V_\tau) = E(U_{\gam'^\T})$,
an intersection term interaction pattern $P \in \mathcal{P}^{intersect}_{\gamma,\tau,{\gamma'}^\T}$ consists of:
    \begin{enumerate}[(i)]
        \item An intersection pattern between $\gamma^-,\tau,(\gamma'^\T)^-$
        \item After these intersections, some edges appear with multiplicity greater than $1$. For each such edge, we linearize it and specify whether we are taking the term with an edge, or the constant term.
        \item For each edge indicator in $E(U_{\tau}) \cup E(V_{\tau})$
        which is now in the middle of $\gamma^- \circ \tau \circ ({\gamma'}^\T)^-$, we specify whether the edge is kept or removed when the indicator is removed.
    \end{enumerate}
    We furthermore have the structural property of $P$ with $\gamma,\tau,\gamma'^\T$ that $U_{\gamma}$ is the leftmost SMVS in $\gamma$
    of $U_{\gam}$ and $V_\gam \cup V_{intersected}(\gamma)$,
    and $V_{{\gamma'}^\T}$ is the rightmost SMVS in $\gamma'^\T$ of $U_{\gamma'^\T} \cup V_{intersected}(\gamma')$ and $V_{\gam'^\T}$.
\end{definition}

\begin{definition}[$\tau_P$]
    Given an interaction pattern $P \in \calP^{interact}_{\gam, \tau, \gam'^\T}$, let $\tau_P$ be the resulting shape.
\end{definition}

\begin{definition}[$N(P)$, implicit definition]
\label{def:np}
Given $\gamma,\tau,{\gamma'}^\T$ and an interaction pattern $P \in \mathcal{P}^{interact}_{\gamma,\tau,{\gamma'}^\T}$, we define $N(P)$ so that for each ribbon $R$ of shape $\tau_P$, $N(P)$ is the number of triples of ribbons $R_1,R_2,R_3$ of shapes $\gamma,\tau,{\gamma'}^\T$ which result in the ribbon $R$ through interaction pattern $P$.
\end{definition}
$N(P) > 1$ holds when there is an increase in symmetry in an
interaction term. 

\begin{definition}[$c_P$, implicit definition]
    Given an interaction pattern $P \in \mathcal{P}^{interact}_{\gamma,\tau,{\gamma'}^\T}$, we define $c_P$ such that $N(P)c_P\lam_{\tau_P}$ is the coefficient on a resulting ribbon $R$ of shape $\tau_P$.
\end{definition}

Following the analysis in \cref{claim:coefficient-edge-removal}, \cref{claim:coefficient-intersection-term}, we have the following explicit formulas for $c_P$.
\begin{definition}[$c_P$, explicit definition]
    Given an interaction pattern $P \in \mathcal{P}^{interact}_{\gamma,\tau,{\gamma'}^\T}$, 
    \begin{enumerate}
        \item If $P$ is a PMVS interaction,
        \begin{align*}
            c_P = &\left(n^{-\gamma}\right)^{(\text{total }\# \text{ of removed edges})}(-1)^{\# \text{ of edges removed from } U_{\tau} \cup V_{\tau} \text{ for not being in } E(G)}\\
            &\left(\frac{1}{1-p}\right)^{(\# \text{ of new edge indicators})}
            (1-p)^{\# \text{ of edges and edge indicators removed from the middle}}
        \end{align*}
    \item If $P$ is an intersection term interaction,
    \begin{align*}
        c_P &=\left(\frac{k}{n}\right)^{(\# \text{ of intersections})}\Bigg(\prod_{e \in E_{tot}(\tau_P)}\left(\frac{q-p}{\sqrt{p(1-p)}}\right)^{\mult(e) - 1 + \one_{e \text{ vanishes}}}\left(\frac{1-2p}{\sqrt{p(1-p)}}\right)^{(\mult(e) - 2) \one_{e \text{ vanishes}}}\\
        &\left(\frac{1-2p}{\sqrt{p(1-p)}}\right)^{1_{\mult(e) = 2, e \text{ does not vanish or have an indicator}}}\left(\sqrt{\frac{1-p}{p}}\right)^{1_{\mult(e) = 2, e \text{ has an indicator}}}\\
        &\left(\frac{1-3p+3p^2}{p(1-p)}\right)^{1_{\mult(e) = 3, e \text{ does not vanish}}}\Bigg)\\
        &\quad\cdot\left(1-p\right)^{(\# \text{ of edges and edge indicators removed from the middle})} \left(q-p\right)^{(\# \text{ of edges removed from the middle})}
        \end{align*}
    \end{enumerate}
\end{definition}

\begin{lemma}\label{lem:cP-adjustment}
        \begin{align*}
        |c_P| &\leq \left(n^{\alpha-1}\right)^{(\# \text{ of intersections})}\left(\prod_{e \in E_{tot}(\tau_P)}{\left(n^{\frac{\beta}{2} 
         - \gamma}\right)^{\mult(e) - 1 + \one_{e \text{ vanishes}}}\left(n^{\frac{\beta}{2}}\right)^{\mult(e) - 1 - \one_{e \text{ vanishes}}}}\right)\\
        &\quad \cdot (1-p)^{\# \text{ of edges and edge indicators removed from the middle}} \left(n^{-\gamma}\right)^{(\# \text{ of edges removed from the middle})}
    \end{align*}
\end{lemma}
\begin{proof}
    This may be proven either directly from the exact formula above, or by the analysis already presented in \cref{claim:coefficient-intersection-term}.
    From the above, note that $\frac{1-2p}{\sqrt{p(1-p)}} \leq \sqrt{\frac{1-p}{p}} = n^{-\frac{\beta}{2}}$.
    Note that $\frac{1-3p+3p^2}{p(1-p)} \leq \frac{1-p}{p} = n^{-\frac{\beta}{2}}$ as $p \leq \frac{1}{2}$.
\end{proof}

\begin{definition}[$c_P^\approx$]
Define $c_P^\approx$ to be the polynomial factors in $c_P$, namely
\begin{align*}
{c_P^\approx} &= \left(n^{\alpha-1}\right)^{(\# \text{ of intersections})}
\left(\prod_{e \in E_{tot}(\tau_P)}\left(n^{\frac{\beta}{2} 
         - \gamma}\right)^{\mult(e) - 1 + \one_{e \text{ vanishes}}}\left(n^{\frac{\beta}{2}}\right)^{\mult(e) - 1 - \one_{e \text{ vanishes}}}\right)\\
    &\cdot\left(n^{-\gamma}\right)^{(\# \text{ of edges removed from the middle})}
\end{align*}
\end{definition}

\subsection{The approximate PSD decomposition}

Applying the recursion as described in \cref{sec:pmvs}
with the definitions from the previous section, we obtain the approximate PSD decomposition as follows.
The proofs essentially follow from the definitions and are
delayed to \cref{subsubsec:apx-decomp-proofs}.

\begin{definition}[Plus abbreviation]
    Given a shape $\tau$, let $\tau^+$ denote the improper shape with
    left and right edge indicators added to $U_\tau$ and $V_\tau$.
\end{definition}

\begin{definition}[Terminal interaction]
    $P \in \calP^{PMVS}_{\gam, \tau, \gam'^\T}$ is a \emph{terminal PMVS interaction} if all edges are given indicators (none are removed).
    $P \in \calP^{intersection}_{\gam, \tau, \gam'^\T}$ is a \emph{terminal intersection interaction} if there are no intersections.
\end{definition}

\begin{lemma}[One iteration, PMVS operation]
\label{lem:oneround-pmvs}
For all shapes $\tau$ and $D_1, D_2 \in \N$,
\begin{align*}
    &\sum_{(R_1,R_2,R_3) \in \mathcal{T}_{\tau,\leq D_1, \leq D_2}}{\lambda_{R^{-}_1 \circ R_2 \circ R^{-}_3}\mM_{R_1^{-} \circ R_2 \circ R_3^{-}}} \\
    = \; &\sum_{(R_1,R_2,R_3) \in \mathcal{T}_{\tau,\leq D_1, \leq D_2}}
    \left(\frac{1}{1-p}\right)^{(\text{\# of new edge indicators})}
    {\lambda_{R^{-}_1 \circ R_2 \circ R^{-}_3}\mM_{R_1^{-} \circ R_2^+ \circ R_3^{-}}} \\
    &+ \sum_{\gamma \in \mathcal{L}_{U_{\tau},\leq D_1}, \gamma' \in \mathcal{L}_{V_{\tau},\leq D_2}}
    \sum_{\substack{\text{non-terminal}\\P \in \mathcal{P}^{PMVS}_{\gamma,\tau,{\gamma'}^\T}}}
    \frac{N(P)c_P}{|U_\gam|! |V_{\gam'^\T}|!}
    \left(\sum_{(R'_1,R'_2,R'_3) \in \mathcal{T}_{\tau_P,\leq D'_1, \leq D'_2}}{\lambda_{{R'_1}^{-} \circ {R'_2} \circ {R'_3}^{-}}\mM_{{{R'_1}^{-} \circ {R'_2} \circ {R'_3}^{-}}}}\right)
\end{align*}
where $D'_1 = D_1 - |V(\gamma) \setminus U_{\gamma}|$ and $D'_2 = D_2 - |V({\gamma'}^\T) \setminus V_{{\gamma'}^\T}|$
\end{lemma}

\begin{definition}[$\matL$ and $\matL_U$ and $\matL_{U, \ssleq D}$] Let $\matL = \sum_{L \in \calL} \lam_{L^-} \mM_{L^-}$ and $\matL_{U} = \sum_{L \in \calL_{U}} \lam_{L^-} \mM_{L^-}$ and $\matL_{U, \ssleq D} = \sum_{L \in \calL_{U, \ssleq D}} \lam_{L^-} \mM_{L^-}$.
\end{definition}

\begin{lemma}[One iteration, intersection term operation]
\label{lem:oneround-intersection}
For all shapes $\tau$ and $D_1, D_2 \in \N$,
\begin{align*}
    &\sum_{(R_1,R_2,R_3) \in \mathcal{T}_{\tau,\leq D_1, \leq D_2}}{\lambda_{R^{-}_1 \circ R_2 \circ R^{-}_3}\mM_{R_1^{-} \circ R_2^+ \circ R_3^{-}}} = \matL_{U_{\tau},\leq D_1}{\left(\frac{\lambda_{\tau}\mM_{\tau^+}}{|\Aut(\tau)|}\right)}\matL^\T_{V_{\tau}, \leq D_2} \\
    &- \sum_{\gamma \in \mathcal{L}_{U_{\tau},\leq D_1}, \gamma' \in \mathcal{L}_{V_{\tau},\leq D_2}}
    \sum_{\substack{\text{non-terminal}\\P \in \mathcal{P}^{intersect}_{\gamma,\tau^+,{\gamma'}^\T}}}{\frac{N(P){c_P}}{|U_\gam|!|V_{\gam'^\T}|!}\left(\sum_{(R'_1,R'_2,R'_3) \in \mathcal{T}_{\tau_P,\leq D'_1, \leq D'_2}}{\lambda_{{R'_1}^{-} \circ {R'_2} \circ {R'_3}^{-}}\mM_{{{R'_1}^{-} \circ {{R'_2}} \circ {R'_3}^{-}}}}\right)}
\end{align*}
where $D'_1 = D_1 - |V(\gamma) \setminus U_{\gamma}|$ and $D'_2 = D_2 - |V({\gamma'}^\T) \setminus V_{{\gamma'}^\T}|$
\end{lemma}

We repeatedly apply \cref{lem:oneround-pmvs} until we have only terms with PMVS identified (i.e. having both left and right indicators),
then we apply \cref{lem:oneround-intersection}, then we repeat
\cref{lem:oneround-pmvs}, and so forth.
The next lemma gives the formal statement of the final result of the iteration on $\Lam$.

We define iterated interaction patterns,
which are the combinatorial
objects describing the branches of the full recursion.

\begin{definition}[Iterated interaction pattern, $\calP_j(\tau)$]
    Given a shape $\tau$ and $j \in \mathbb{N}$, define $\mathcal{P}_j(\tau)$ to be the set of tuples $(\Gamma, \Gamma'^\T,P)$ such that 
    \begin{enumerate}
    \item $\Gamma$ is a tuple of $j$ composable left shapes $(\gamma_j,\ldots,\gamma_1)$.
    Let $\gam = \gam_j^- \circ \cdots \circ \gam_2^- \circ \gam_1$.
    \item $\Gamma'^{\T}$ is a tuple of $j$ composable right shapes $({\gamma'}^\T_1,\ldots,{\gamma'}^\T_j)$.
    Let $\gam'^\T = \gam'^\T_1 \circ (\gam_2'^\T)^-\cdots \circ (\gam_j'^\T)^-$.
    \item $P$ is a tuple of $j$ interaction patterns $(P_1,\ldots,P_j)$ such that for each $i \in [j]$, $P_i \in \mathcal{P}^{interact}_{\gamma_i,\tau_{P_{i-1}},{\gamma'}^\T_i}$. For $i = 1$, we take $\tau_{P_0} = \tau$.
    \item $P$ consists of sequences of non-terminal PMVS interactions ending with a terminal PMVS interaction (there is at least one sequence), with a non-terminal intersection interaction in between consecutive sequences, then finally ending with a terminal intersection interaction.
    \item $P$ has at least one non-terminal interaction (we will explicitly separate the middle shapes i.e. the terms with a terminal PMVS interaction followed by a terminal intersection interaction because they are a good warm-up for the analysis of the other terms).
    \end{enumerate}
    We use $\tau_P = \tau_{P_j}$ to denote the final resulting shape.
\end{definition}

\begin{definition}[$D_L(P)$ and $D_R(P)$]
    Define $D_L(P) = D_V - |V(\gamma) \setminus U_\gam|$ and \\$D_R(P) = D_V - |V(\gamma'^\T) \setminus V_{\gam'^\T}|$.
\end{definition}

\begin{lemma}[Result of full iteration]\label{lem:decomposition}
\begin{align*}
    \matLam =\;& \sum_{U,V \in \calI_{mid}: U \sim V}{\frac{|U \cap V|!}{(|U|!)^2}\matL_{U,\leq D_V}\left(\sum_{\tau \in \mathcal{M}_{U,V}}{\left(\frac{1}{1-p}\right)^{|E(U_\tau) \cup E(V_\tau)|}\frac{\lambda_{\tau}\mM_{\tau^+}}{|\Aut(\tau)|}}\right)\matL_{V,\leq D_V}^\T} \\
    +\;& \sum_{U,V \in \calI_{mid}: U \sim V}{\frac{|U \cap V|!}{(|U|!)^2}\sum_{\tau \in \mathcal{M}_{U,V}}{\sum_{j=1}^{\infty}{\sum_{(\Gamma,\Gamma'^\T,P) \in \mathcal{P}_j(\tau)}{\frac{(-1)^{\# \text{ of intersection indices in } [j]}}{\prod_{i=1}^{j}{|U_{\gamma_i}|!|V_{{\gamma'_i}^\T}|!}}}}}} \\
    &\matL_{U_{\tau_{P}},\leq D_L(P)}\left(\frac{\left(\prod_{i=1}^{j}{c_{P_i}}N(P_i)\right)\lambda_{\tau_{P}}\mM_{\tau_{P}^+}}{|\Aut(\tau_{P})|}\right)\matL^{\T}_{V_{\tau_{P}},\leq D_R(P)}\\
    -\;& \trunc_1
\end{align*}
\end{lemma}

\begin{remark}
    The dominant terms in the decomposition are
    \begin{align*}
    &\sum_{U, V \in \calI_{mid}: U \sim V} \frac{|U \cap V|!}{(|U|!)^2} \matL_{U, \leq D_V} \left(\sum_{\substack{\tau \in \calM_{U, V}:\\\tau \text{ diagonal}}} \left(\frac{1}{1-p}\right)^{|E(U_\tau) \cup E(V_\tau)|}\frac{\lam_\tau \mM_{\tau^+}}{|\Aut(\tau)|}\right) \matL_{V, \leq D_V}^\T\\
    =\;&\sum_{U \in \calI_{mid}}  \left(\frac{1}{1-p}\right)^{|E(U)|}  \frac{\lam_{U}}{|U|!} \matL_{U, \leq D_V} \mM_{U^+}\matL_{U, \leq D_V}^\T\\
    \psdgeq \;& \sum_{U \in \calI_{mid}} \frac{\lam_{U}}{|U|!} \matL_{U, \leq D_V} \mM_{U^+}\matL_{U, \leq D_V}^\T\,.
    \end{align*}
    We will show that all of the remaining terms are PSD dominated by these terms.
\end{remark}

\subsubsection{Proofs of \cref{lem:oneround-pmvs} and \cref{lem:oneround-intersection}}
\label{subsubsec:apx-decomp-proofs}

\begin{proof}[Proof of \cref{lem:oneround-pmvs}]
    For a ribbon $R_1^- \circ R_2 \circ R_3^-$ on
    the left-hand side, let $E_{missing}(R_2)$ be
    the set of edges in $E(A_{R_2}) \cup E(B_{R_2})$ which do not yet have indicators.
    We update the single nonzero entry of $\mM_{R_1^- \circ R_2 \circ R_3^-}$ using the
    identity
    \[
        \prod_{\substack{e \in E_{missing}(R_2)}}{\left(\frac{q-p}{\sqrt{p(1-p)}}\chi_{e}\right)} 
        = \prod_{\substack{e \in E_{missing}(R_2)}} \left(\frac{1}{1-p} \one_{e \in E(G)}\frac{q-p}{\sqrt{p(1-p)}}\chi_e - \frac{q-p}{1-p}\right)\,.
    \]
    Note that $n^{-\gamma} = \frac{q-p}{1-p}$. Expanding this identity,
\begin{align*}
    &\sum_{(R_1,R_2,R_3) \in \mathcal{T}_{\tau,\leq D_1, \leq D_2}}{\lambda_{R^{-}_1 \circ R_2 \circ R^{-}_3}\mM_{R_1^{-} \circ R_2 \circ R_3^{-}}} \\
    = \; &\sum_{(R_1,R_2,R_3) \in \mathcal{T}_{\tau,\leq D_1, \leq D_2}} \sum_{E_1 \subseteq E_{missing}(R_2)}
    \left(\frac{1}{1-p}\right)^{|E_{missing}(R_2)| - |E_1|} \left(-n^{-\gamma}\right)^{|E_1|} \lambda_{R^{-}_1 \circ (R_2 \setminus E_1)^+ \circ R^{-}_3}\mM_{R_1^- \circ (R_2 \setminus E_1)^+ \circ R_3^-}
\end{align*}
    
    The term choosing $\frac{1}{1-p}\one_{e \in E(G)}\chi_e$ for all $e$ (i.e. $E_1 = \emptyset$) is the terminal interaction and yields $\mM_{R_1^- \circ R_2^+ \circ R_3^-}$.

    For the other $E_1$ terms, we refactor using the new SMVS,
    \[
    R_1^- \circ (R_2 \setminus S)^+ \circ R_3^- = R_1'^-\circ R_2'\circ R_3'^-
    \]
    Define $G_1$ and $G_3^\T$ by $R_1 = R_1'^- \circ G_1$ and $R_3 = G_3^\T \circ R_3'^-$.
    Recall that we need to specify the orders
    of $A_{G_1}$ and $B_{G_3^\T}$.
    We symmetrize over all $|A_{G_1}|!|B_{G_3^\T}|!$
    possible choices.

    At this point, we may have edge indicators in the middle of $G_1^-\circ (R_2 \setminus E_1)^+ \circ (G_3^\T)^-$. Letting $E_{extra}$ be the set of edges in the middle of $G_1^-\circ (R_2 \setminus E_1)^+ \circ (G_3^\T)^-$,
    \[
    \prod_{e \in E_{extra}}{\left(1_{e \in E(G)}\frac{q-p}{\sqrt{p(1-p)}}\chi_{e}\right)} = \prod_{e \in E_{extra}}{\left((1-p)\frac{q-p}{\sqrt{p(1-p)}}\chi_{e} + (q-p)\right)}
    \]
    Note that $(q-p) = (1-p)n^{-\gamma}$. Expanding this product out, we let $E_2$ be the set of edges where we choose the $(1-p)n^{-\gamma}$ term. After doing this, we set
    \[
    R_2' = G_1^-\circ (R_2 \setminus (E_1 \cup E_2))^+ \circ (G_3^\T)^-
    \]

    We replace the summation over ${(R_1,R_2,R_3) \in \mathcal{T}_{\tau,\leq D_1, \leq D_2}}$ by a summation
    over $(R_1', G_1, R_2, G_3^\T, R_3')$.
    We crucially have that $R_2'$ does not depend on $R_1'$ or $R_3'$, \cref{rmk:left-oblivious}.
    Therefore, we first sum over $(G_1, R_2, G_3^\T)$, $E_1 \subseteq E_{missing}(R_2)$, and $E_2 \subseteq E_{extra}(G_1^-\circ (R_2 \setminus E_1)^+ \circ (G_3^\T)^-)$ and
    then sum over $R_1'$ and $R_3'$.
    $R_2$ is a ribbon of shape $\tau$,
    $G_1$ and $G_3^\T$ are arbitrary
    left and right ribbons which match the left and right sides of $R_2$, and 
    together there is the additional structural
    property that
    \begin{align*}
    \text{$A_{G_1}$ is the leftmost SMVS of $G_1$,}\\
    \text{$B_{G_3^\T}$ is the rightmost SMVS of $G_3^\T$,}&& (*)
    \end{align*}
    In summary,
    \begin{align*}
    &\sum_{(R_1,R_2,R_3) \in \mathcal{T}_{\tau,\leq D_1, \leq D_2}} \sum_{E_1 \subseteq E_{missing}(R_2)}
    \left(\frac{1}{1-p}\right)^{|E_{missing}| - |E_1|} \left(-n^{-\gamma}\right)^{|E_1|} \lambda_{R^{-}_1 \circ (R_2 \setminus E_1)^+ \circ R^{-}_3}\mM_{R_1^- \circ (R_2 \setminus E_1)^+ \circ R_3^-}\\
=\;& \sum_{R_2 \in \calR_\tau}\sum_{S \subseteq E_{missing}(R_2)}
\sum_{\substack{G_1 \in \calL_{A_{R_2}, \leq D_1},\\
G_3 \in \calL_{B_{R_2}, \leq D_2}:\\
\text{(*) holds}}}\sum_{E_2 \subseteq E_{extra}(G_1^-\circ (R_2 \setminus E_1)^+ \circ (G_3^\T)^-)}
    \frac{1}{|A_{G_1}|!|B_{G_3^\T}|!}
    \left(\frac{1}{1-p}\right)^{|E_{missing}| - |E_1|}\\
    &(1-p)^{|E_{extra}|}(-1)^{|E_1|}\left(n^{-\gamma}\right)^{|E_1| + |E_2|}
    \sum_{\substack{R_1' \in \calL_{A_{G_1}, \leq D_1 - |V(G_1) \setminus A_{G_1}|}, \\
    R_3' \in \calL_{B_{G_3^\T}, \leq D_2 - |V(G_3^\T) \setminus B_{G_3^\T}|}}}
    \lambda_{R_1'^{-} \circ R_2' \circ  R_3'^{-}}\mM_{R_1'^{-} \circ R_2' \circ  R_3'^{-}}
\end{align*}
Finally, we prepare to shift from ribbons to shapes.
Fixing the shape $\gam$ of $G_1$ and $\gam'^\T$ of $G_3^\T$, the sum over $S \subseteq E_{missing}(R_2)$ and the condition (*)
is equivalent to summing over $P \in \calP^{PMVS}_{\gam, \tau, \gam'^\T}$.
Summing over $R_2'$ with the final shape $\tau_P$ specified by $P$, the sum over $R_1', R_2', R_3'$ is equivalent to summing over $\calT_{\tau_P, D_1', D_2'}$.
The coefficients are gathered into $N(P)c_P$.
In summary, the above is equal to
\begin{align*}
    \sum_{\substack{\gamma \in \mathcal{L}_{U_{\tau},\leq D_1}, \\\gamma' \in \mathcal{L}_{V_{\tau},\leq D_2}}}
    \sum_{\substack{P \in \mathcal{P}^{PMVS}_{\gamma,\tau,{\gamma'}^\T}}}
    \frac{N(P)c_P}{|U_\gam|! |V_{\gam'^\T}|!}
    \left(\sum_{(R'_1,R'_2,R'_3) \in \mathcal{T}_{\tau_P,\leq D_1 - |V(\gam) \setminus U_\gam|, \leq D_2 - |V(\gam'^\T) \setminus V_{\gam'^\T}|}}{\lambda_{{R'_1}^{-} \circ {R'_2} \circ {R'_3}^{-}}\mM_{{{R'_1}^{-} \circ {R'_2} \circ {R'_3}^{-}}}}\right)
\end{align*}
as needed.
\end{proof}

\begin{proof}[Proof sketch of \cref{lem:oneround-intersection}]
    This follows in the same way as the previous lemma.

    We have that $\matL_{U_{\tau},\leq D_1}{\left(\frac{\lambda_{\tau}\mM_{\tau}}{|\Aut(\tau)|}\right)}\matL^\T_{V_{\tau}, \leq D_2}$ is a sum over left, middle, and right ribbons respectively,
    where the sum over middle ribbons is over distinct ribbons due to normalization by $|\Aut(\tau)|$, \cref{prop:automorphism-sum}.
    We will argue that each $R = R_1^- \circ R_2 \circ R_3^-$ has the same
    coefficient on both sides of the equality.

    Suppose $R = R_1^- \circ R_2 \circ R_3^-$ where $R_1, R_2, R_3$ are properly composable. This term appears on the left-hand side
    with coefficient $\lam_{R_1^- \circ R_2 \circ R_3^-}$.
    The coefficient factors into $\lam_{R_1^-}\lam_{R_2}\lam_{R_3^-}$
    by \cref{lem:coefficients-factor}.
    Hence these terms match up.

    We now consider $R_1^- \circ R_2 \circ R_3^-$ which is an improper composition. These don't occur on the left, hence
    we would like to show that they cancel on the right.
    $R_1, R_2, R_3$ give rise to an intersection pattern $P \in \calP^{intersect}_{\gam, \tau, \gam'^\T}$ where $\gam, \gam'^\T$ are defined in \cref{def:gam}.

    In order to specify $R_1, R_2, R_3$ and the interaction pattern $P$, the latter sum instead specifies $R_1', R_2', R_3'$ in $\calT_{\tau_P, \le D_1', \le D_2}$ along with the orderings of $U_{\gam}, V_{\gam'}^\T$.
    There are $|U_\gam|!$ choices for the order of $U_\gam$ and $|V_{\gam'^\T}|!$ choices for the order of $V_{\gam'^\T}$,
    and we symmetrize over all choices.
    There may be multiple ribbons $R_1, R_2, R_3$ leading to the same ribbon $R_1^- \circ R_2 \circ R_3^-$ even with the same intersection pattern, and this is accounted for by $N(P)$.
    
    The change in coefficient can be analyzed as follows:
    \begin{enumerate}
        \item For each intersection, there is a factor of $n^{(\alpha-1)}$ coming from the $\lambda$ coefficients which needs to be shifted to $c_P$.
        \item For each edge $e \in E_{tot}(\tau_P)$, there is a factor of $\left(\frac{q-p}{\sqrt{p(1-p)}}\right)^{\mult(e)- 1 + 1_{e \text{ vanishes}}}$ coming from the $\lambda$ coefficients which needs to be shifted to $c_P$. 
        \item ${\chi_e}^2 = 1 + \frac{1-2p}{\sqrt{p(1-p)}}\chi_e$.
        \item ${\chi_e}^3 = \left(1 + \frac{1-2p}{\sqrt{p(1-p)}}\chi_e\right)\chi_e = \frac{1-2p}{\sqrt{p(1-p)}} + \frac{1-3p+3p^2}{p(1-p)}\chi_{e}$. Note that because only three ribbons are being composed, the maximum multiplicity of a multiedge is 3.
        \item The edge indicators and edges which are removed from the middle can be analyzed in the same way as before.
    \end{enumerate}
    This change is accounted for by $c_P$. This completes the proof.
\end{proof}

\subsection{Analyzing \texorpdfstring{$\matLam$}{Lambda}}
First we factor out the truncation error from $\matLam$.

\begin{lemma}\label{lem:factor-out-truncation}
\begin{align*}
    \matLam =\;& \sum_{U,V \in \calI_{mid}: U \sim V}{\frac{|U \cap V|!}{(|U|!)^2}\matL_{U}\left(\sum_{\tau \in \mathcal{M}_{U,V}}{\left(\frac{1}{1-p}\right)^{|E(U_\tau) \cup E(V_\tau)|}\frac{\lambda_{\tau}\mM_{\tau^+}}{|\Aut(\tau)|}}\right)\matL_{V}^\T} \\
    +\;& \sum_{U,V \in \calI_{mid}: U \sim V}{\frac{|U \cap V|!}{(|U|!)^2}\sum_{\tau \in \mathcal{M}_{U, V}}{\sum_{j=1}^{\infty}{\sum_{(\Gamma,\Gamma'^\T,P) \in \mathcal{P}_j(\tau)}{\frac{(-1)^{\# \text{ of intersection indices in } [j]}}{\prod_{i=1}^{j}{|U_{\gamma_i}|! |V_{{\gamma'_i}^\T}|!}}}}}} \\
    &\matL_{U_{\tau_{P}}}\left(\frac{\left(\prod_{i=1}^{j}{c_{P_i}}N(P_i)\right)\lambda_{\tau_{P}}\mM_{\tau_{P}^+}}{|\Aut(\tau_{P})|}\right)\matL^{\T}_{V_{\tau_{P}}} \\
    -\;& \trunc_1 +\trunc_2
\end{align*}
where $\trunc_2$ is defined in \cref{def:truncation2}.
\end{lemma}

\begin{proof}
Starting from \cref{lem:decomposition}, for each interaction term $\tau, j, (\Gam, \Gam'^\T, P)$, we use
\begin{align*}
    &\matL_{U_{\tau_{P}},\leq D_L(P)}\left(\frac{\left(\prod_{i=1}^{j}{c_{P_i}}N(P_i)\right)\lambda_{\tau_{P}}\mM_{\tau_{P}^+}}{|\Aut(\tau_{P})|}\right)\matL^{\T}_{V_{\tau_{P}},\leq D_R(P)} \\
    =\;&\matL_{U_{\tau_{P}},\leq D_V}\left(\frac{\left(\prod_{i=1}^{j}{c_{P_i}}N(P_i)\right)\lambda_{\tau_{P}}\mM_{\tau_{P}^+}}{|\Aut(\tau_{P})|}\right)\matL^{\T}_{V_{\tau_{P}},\leq D_V} \\
    +\;&\left(\matL_{U_{\tau_{P}},\leq D_L(P)} - \matL_{U_{\tau_{P}},\leq D_V}\right)\left(\frac{\left(\prod_{i=1}^{j}{c_{P_i}}N(P_i)\right)\lambda_{\tau_{P}}\mM_{\tau_{P}^+}}{|\Aut(\tau_{P})|}\right)\matL^{\T}_{V_{\tau_{P}},\leq D_V} \\
    +\;&\matL_{U_{\tau_{P}},\leq D_L(P)}\left(\frac{\left(\prod_{i=1}^{j}{c_{P_i}}N(P_i)\right)\lambda_{\tau_{P}}\mM_{\tau_{P}^+}}{|\Aut(\tau_{P})|}\right)\left(\matL^{\T}_{V_{\tau_{P}},\leq D_R(P)} - \matL^{\T}_{V_{\tau_{P}},\leq D_V}\right)
\end{align*}
The first term has $\matL_{U_{\tau_P}, \leq D_V} = \matL_{U_{\tau_P}}$
and $\matL_{V_{\tau_P}, \leq D_V} = \matL_{V_{\tau_P}}$ as needed.
The second and third terms are collected into $\trunc_2$. The following definition completes
the proof of the lemma.
\begin{definition}[$\trunc_2$]\label{def:truncation2}
    \begin{align*}
    \trunc_2 =& \sum_{U,V \in \calI_{mid}: U \sim V}{\frac{|U \cap V|!}{(|U|!)^2}\sum_{\tau \in \mathcal{M}_{U, V}}{\sum_{j=1}^{\infty}{\sum_{(\Gamma,\Gamma'^\T,P) \in \mathcal{P}_j(\tau)}{\frac{(-1)^{\# \text{ of intersection indices in } [j]}}{\prod_{i=1}^{j}{|U_{\gamma_i}|! |V_{{\gamma'_i}^\T}|!}}}}}}\\
    &\left(\left(\matL_{U_{\tau_{P}},\leq D_L(P)} - \matL_{U_{\tau_{P}},\leq D_V}\right)\left(\frac{\left(\prod_{i=1}^{j}{c_{P_i}}N(P_i)\right)\lambda_{\tau_{P}}\mM_{\tau_{P}^+}}{|\Aut(\tau_{P})|}\right)\matL^{\T}_{V_{\tau_{P}},\leq D_V}\right. \\
    +\;&\left.\matL_{U_{\tau_{P}},\leq D_L(P)}\left(\frac{\left(\prod_{i=1}^{j}{c_{P_i}}N(P_i)\right)\lambda_{\tau_{P}}\mM_{\tau_{P}^+}}{|\Aut(\tau_{P})|}\right)\left(\matL^{\T}_{V_{\tau_{P}},\leq D_R(P)} - \matL^{\T}_{V_{\tau_{P}},\leq D_V}\right)\right)
    \end{align*}
\end{definition}
\end{proof}

We would now like to analyze the non-truncation terms.
As we proved in \cref{sec:psdness},
the norm of each individual shape $\norm{\mM_\tau}$ and $\norm{\mM_{\tau_P}}$ is under control.
In order to sum over all the shapes,
we use combinatorial functions to convert the sum into a max.

The idea is as follows. If we have a sum of the form $\sum_{\alpha}{B(\alpha)}$ where $B(\alpha)$ is non-negative then instead of bounding the sum directly, we can choose a relatively simple function $c(\alpha)$ such that $\sum_{\alpha}{\frac{1}{c(\alpha)}} \leq 1$ and observe that 
\[
\sum_{\alpha}{B(\alpha)} = \sum_{\alpha}{\frac{1}{c(\alpha)}c(\alpha)B(\alpha)} \leq \sum_{\alpha}{\frac{1}{c(\alpha)}\max_{\alpha}{\{c(\alpha)B(\alpha)\}}} \leq \max_{\alpha}{\{c(\alpha)B(\alpha)\}}
\]
This allows us to use our bound on the individual terms.

\begin{definition}[$c(\al)$ and $c(P)$, informal version of \cref{def:c-functions}]\ 
\begin{enumerate}
    \item For shapes $\alpha$, $c(\alpha)$ is used to control the number of shapes we are summing over.    
    In particular, we have that for all $U \in \calI_{mid}$, 
    $\sum_{\text{shapes }\alpha: U_{\alpha} = U, \alpha \text{ is non-trivial}}{\frac{1}{|U_{\alpha} \cap V_{\alpha}|!c(\alpha)}} \leq 1$. Similarly, for all $V \in \calI_{mid}$, 
    $\sum_{\text{shapes }\alpha: V_{\alpha} = V, \alpha \text{ is non-trivial}}{\frac{1}{|U_{\alpha} \cap V_{\alpha}|!c(\alpha)}} \leq 1$.
    \item For intersection patterns $P$, $c(P)$ is used to control the number of interaction patterns we are summing over. In particular, for all $\gamma$, $\tau$, and ${\gamma'}^\T$, $\sum_{P \in \mathcal{P}^{interact}_{\gamma,\tau,{\gamma'}^\T}}{\frac{1}{c(P)}} \leq 1$.
\end{enumerate}
\end{definition}

The formal definitions are given in \cref{sec:c-functions},
where we will verify that they satisfy the stated summation property and also \cref{assumptions}.

With the combinatorial functions in hand, we can bound non-square terms by the square terms $\left\{\matL_{U}\mM_{U^+}\matL_{U}^\T: U \in \calI_{mid}\right\}$ as follows.
\cref{cor:middle-term-bound} applies to middle shapes and
\cref{cor:intersection-term-bound} applies to intersection terms.
The proofs are delayed
to the next subsection.

\begin{restatable}{corollary}{middleTermBound}\label{cor:middle-term-bound}
For $n$ sufficiently large,
\begin{align*}
    &\sum_{U,V \in \calI_{mid}: U \sim V}{\frac{|U \cap V|!}{(|U|!)^2}\matL_{U}\left(\sum_{\tau \in \mathcal{M}_{U,V}: \tau \text{ is non-diagonal}}\left(\frac{1}{1-p}\right)^{|E(U_\tau) \cup E(V_\tau)|}{\frac{\lambda_{\tau}\mM_{\tau^+}}{|\Aut(\tau)|}}\right)\matL^\T_{V}} \\
    &\succeq -\frac{1}{4}\sum_{U \in \calI_{mid}}{\frac{\lambda_{U}}{|U|!}\matL_{U}{\mM_{U^+}}\matL^{\T}_{U}}
\end{align*}
\end{restatable}

\begin{restatable}{corollary}{intersectionTermBound}\label{cor:intersection-term-bound}
For $n$ sufficiently large,
\begin{align*}
    &\sum_{U,V \in \calI_{mid}: U \sim V}{\frac{|U \cap V|!}{(|U|!)^2}\sum_{\tau \in \mathcal{M}_{U, V}}{\sum_{j=1}^{\infty}{\sum_{(\Gamma,\Gamma'^\T,P) \in \mathcal{P}_j(\tau)}{\frac{(-1)^{\# \text{ of intersection indices in } [j]}}{\prod_{i=1}^{j}{|U_{\gamma_i}|! |V_{{\gamma'_i}^\T}|!}}}}}} \\
    &\matL_{U_{\tau_{P}}}\left(\frac{\left(\prod_{i=1}^{j}{c_{P_i}}N(P_i)\right)\lambda_{\tau_{P}}\mM_{\tau_{P}}}{|\Aut(\tau_{P})|}\right)\matL^{\T}_{V_{\tau_{P}}} \\
    &\succeq -\frac{1}{4}\sum_{U \in \mathcal{I}_{mid}}\frac{\lam_U}{|U|!}{\matL_{U}\mM_{U^+}\matL^\T_{U}}
\end{align*}
\end{restatable}

Putting together \cref{cor:middle-term-bound} and \cref{cor:intersection-term-bound} with \cref{lem:factor-out-truncation},
we have
\begin{lemma}
For $n$ sufficiently large,
\begin{align*}
    \Lam \psdgeq \frac{1}{2}\sum_{U \in \calI_{mid}} \frac{1}{|U|!} \lam_U \matL_U\mM_{U^+}\matL_{U}^\T - \trunc_1 + \trunc_2\,.
\end{align*}
\end{lemma}

The remaining tasks to prove $\Lam \psdgeq 0$ are to verify \cref{assumptions}, and to analyze the truncation error, which we carry out in the following sections.

\subsubsection{Proofs of \cref{cor:middle-term-bound} and \cref{cor:intersection-term-bound}}

The building block that allows us to formally charge these shapes is the following lemma,
which lower bounds the negative impact of each individual term by ``square terms''.
\begin{lemma}\label{lem:comparingtosquares}
    For all shapes $\tau$, all $D_1,D_2 \in \mathbb{N}$, and all $b \in \{-1,1\}$,
    \begin{align*}
        &b(\matL_{U_{\tau}, \leq D_1}\lam_\tau\mM_{\tau^+}\matL^\T_{V_{\tau}, \leq D_2} + \matL_{V_{\tau}, \leq D_2}{\lam_{\tau}\mM_{\tau^+}^\T}\matL^\T_{U_{\tau}, \leq D_1}) \\
        \succeq\;& \frac{-\lam_\tau||\mM_{\tau}||}{\sqrt{\lam_{U_\tau}\norm{\mM_{U_\tau}} \lam_{V_\tau} \norm{\mM_{V_\tau}}}}(\matL_{U_{\tau}, \leq D_1}\lam_{U_\tau}\mM_{U_\tau^+}\matL^{\T}_{U_{\tau}, \leq D_1} + \matL_{V_{\tau}, \leq D_2}\lam_{V_\tau}\mM_{V_\tau^+}\matL^{\T}_{V_{\tau}, \leq D_2})
    \end{align*}
\end{lemma}
\begin{proof}
    We claim that for all $s>0$,
    \begin{align*}
    &b(\matL_{U_{\tau}, \leq D_1}\mM_{\tau^+}\matL^\T_{V_{\tau}, \leq D_2} + \matL_{V_{\tau}, \leq D_2}{\mM_{\tau^+}^\T}\matL^\T_{U_{\tau}, \leq D_1})\\
    \succeq\; & -s\matL_{U_{\tau}, \leq D_1}(\mM_{\tau^+} \mM_{\tau^+}^\T)^{1/2}\matL^{\T}_{U_{\tau}, \leq D_1} - \frac{1}{s}\matL_{V_{\tau}, \leq D_2}(\mM_{\tau^+}^\T \mM_{\tau^+})^{1/2}\matL^{\T}_{V_{\tau}, \leq D_2}
    \end{align*}
    Writing $\mM_{\tau^+} = \mathbf{X} \mathbf{\Sigma} \mathbf{Y}^\T$ for the singular value
    decomposition of $\mM_{\tau^+}$,
    observe that
        \begin{align*}
        0 & \preceq \left(\sqrt{s}\matL_{U_{\tau}, \leq D_1}\mathbf{X}\mathbf{\Sigma}^{1/2} + \frac{b}{\sqrt{s}} \matL_{V_{\tau}, \leq D_2}\mathbf{Y}\mathbf{\Sigma}^{1/2}\right)
        \left(\sqrt{s}\mathbf{\Sigma}^{1/2}\mathbf{X^\T}\matL^{\T}_{U_{\tau}, \leq D_1} + \frac{b}{\sqrt{s}} \mathbf{\Sigma}^{1/2}\mathbf{Y}^\T\matL^{\T}_{V_{\tau}, \leq D_2}\right) \\
        &= s\matL_{U_{\tau}, \leq D_1}\mathbf{X}\mathbf{\Sigma}\mathbf{X}^\T\matL^{\T}_{U_{\tau}, \leq D_1} + \frac{1}{s}\matL_{V_{\tau}, \leq D_2}\mathbf{Y}\mathbf{\Sigma}\mathbf{Y}^\T\matL^{\T}_{V_{\tau}, \leq D_2} \\
        &\quad + b(\matL_{U_{\tau}, \leq D_1}\mathbf{X}\mathbf{\Sigma}\mathbf{Y}^\T\matL^\T_{V_{\tau}, \leq D_2} + \matL_{V_{\tau}, \leq D_2}\mathbf{Y}\mathbf{\Sigma}\mathbf{X}^\T\matL^\T_{U_{\tau}, \leq D_1})\\
        &= s\matL_{U_{\tau}, \leq D_1}(\mM_{\tau^+} \mM_{\tau^+}^\T)^{1/2}\matL^{\T}_{U_{\tau}, \leq D_1} + \frac{1}{s}\matL_{V_{\tau}, \leq D_2}(\mM_{\tau^+}^\T \mM_{\tau^+})^{1/2}\matL^{\T}_{V_{\tau}, \leq D_2} \\
        &\quad+ b(\matL_{U_{\tau}, \leq D_1}\mM_{\tau^+}\matL^\T_{V_{\tau}, \leq D_2} + \matL_{V_{\tau}, \leq D_2}\mM_{\tau^+}^\T\matL^\T_{U_{\tau}, \leq D_1})        
        \end{align*}
    which implies the claim.
    
    We claim $(\mM_{\tau^+} \mM_{\tau^+}^\T)^{1/2} \psdleq \norm{\mM_\tau}\frac{\mM_{U_\tau^+}}{\norm{\mM_{U_\tau}}}$.
    To see this, note that $\mM_{U_\tau^+}$ is a diagonal
    matrix with nonnegative entries, therefore $\frac{\mM_{U_\tau^+}}{\norm{\mM_{U_\tau}}}$ has diagonal entries which are 0 and 1.
    The supported rows are the same as $\mM_{\tau^+} \mM_{\tau^+}^\T$
    hence the claim.

    Similarly, $(\mM_{\tau^+}^\T \mM_{\tau^+})^{1/2} \psdleq \norm{\mM_\tau}\frac{\mM_{V_\tau^+}}{\norm{\mM_{V_\tau}}}$.
    Using these claims with $s = \sqrt{\frac{\lam_{U_\tau}\norm{\mM_{U_\tau}}}{\lam_{V_\tau}\norm{\mM_{V_\tau}}}}$ completes the proof.
\end{proof}

\begin{lemma}\label{lem:nonsquare-middle-shapes}
\begin{align*}
&\sum_{U,V \in \calI_{mid}: U \sim V}{\frac{|U \cap V|!}{(|U|!)^2}\matL_{U}\left(\sum_{\tau \in \mathcal{M}_{U,V}: \tau \text{ is nontrivial}}{\frac{\lambda_{\tau}\mM_{\tau^+}}{|\Aut(\tau)|}}\right)\matL^\T_{V}} \\
&\succeq -\sum_{U \in \calI_{mid}}{\left(\max_{V, \tau: U \sim V, \tau \in \mathcal{M}_{U,V}, \tau \text{ is nontrivial}}{\left\{c(\tau)\frac{\lambda_{\tau}||\mM_{\tau}||}{\lambda_{U}||\mM_U||}\right\}}\right)\frac{\lambda_{U}}{|U|!}\matL_{U}{\mM_{U^+}}\matL^{\T}_{U}}
\end{align*}
\end{lemma}
\begin{proof}
Applying \cref{lem:comparingtosquares} and using the trivial bound $|\Aut(\tau)| \geq 1$,
\begingroup
\allowdisplaybreaks
\begin{align*}
&\sum_{U,V \in \calI_{mid}: U \sim V}{\frac{|U \cap V|!}{(|U|!)^2}\matL_{U}\left(\sum_{\text{nontrivial } \tau \in \mathcal{M}_{U,V}}{\frac{\lambda_{\tau}\mM_{\tau}}{|\Aut(\tau)|}}\right)\matL^\T_{V}} \\
&\succeq -\sum_{U,V \in \calI_{mid}: U \sim V}{\frac{|U \cap V|!}{(|U|!)^2}\left(\sum_{\text{nontrivial } \tau \in \mathcal{M}_{U,V}}{\frac{\lambda_{\tau}||\mM_{\tau}||}{\lam_{U}\norm{\mM_U}}
\left(\frac{1}{2}\matL_{U}\mM_{U^+}\matL^{\T}_{U} + \frac{1}{2}\matL_{V}\mM_{V^+}\matL^{\T}_{V}\right)}\right)} \\
&= -\sum_{U \in \calI_{mid}} \frac{1}{2|U|!}{\left(\sum_{V, \tau: U \sim V, \tau \in \mathcal{M}_{U,V},  \tau \text{ is nontrivial}}{\frac{|U \cap V|!}{c(\tau)|U|!}c(\tau)\frac{\lambda_{\tau}||\mM_{\tau}||}{\lambda_{U}||\mM_U||}}\right)\matL_{U}\mM_{U^+}\matL^{\T}_{U}}\\
&\quad- \sum_{V \in \calI_{mid}}\frac{1}{2|V|!}{\left(\sum_{U, \tau: U \sim V, \tau \in \mathcal{M}_{U,V},  \tau \text{ is nontrivial}}{\frac{|U \cap V|!}{c(\tau)|U|!}c(\tau)\frac{\lambda_{\tau}||\mM_{\tau}||}{\lambda_{U}||\mM_U||}}\right)\matL_{V}\mM_{V^+}\matL^{\T}_{V}} \\
&\succeq -\sum_{U \in \calI_{mid}}{\left(\max_{V, \tau: U \sim V, \tau \in \mathcal{M}_{U,V},  \tau \text{ is nontrivial}}{\left\{c(\tau)\frac{\lambda_{\tau}||\mM_{\tau}||}{\lambda_{U}||\mM_U||}\right\}}\right)\frac{1}{|U|!}\matL_{U}\mM_{U^+}\matL^{\T}_{U}}
\end{align*}
where the last line uses the following facts:
\begin{enumerate}
\item For all $U \in \mathcal{I}_{mid}$, $\sum_{V, \tau: U \sim V, \tau \in \mathcal{M}_{U,V},  \tau \text{ is nontrivial}}{\frac{|U \cap V|!}{c(\tau)|U|!}} \leq 1$ 
\item For all $V \in \mathcal{I}_{mid}$, $\sum_{U, \tau: U \sim V, \tau \in \mathcal{M}_{U,V},  \tau \text{ is nontrivial}}{\frac{|U \cap V|!}{c(\tau)|U|!}} \leq 1$.
\end{enumerate}
\end{proof}
\endgroup

\middleTermBound*

\begin{proof}
    We observe that 
    \begin{align*}
        &\max_{V, \tau: U \sim V, \tau \in \mathcal{M}_{U,V},  \tau \text{ is nontrivial}}{\left\{c(\tau)\frac{\lambda_{\tau}||\mM_{\tau}||}{\lambda_{U}||\mM_U||}\right\}} \\
        &\leq \max_{V, \tau: U \sim V, \tau \in \mathcal{M}_{U,V},  \tau \text{ is nontrivial}}{\left\{c(\tau)B_{adjust}(\tau)n^{-\slack(\tau)}\right\}} \leq \frac{1}{4}
    \end{align*}
    where the last inequality follows from the facts that for all $\tau$,
    \begin{enumerate}
        \item $\slack(\tau) \geq \frac{\epsilon}{4}(|E(\tau)| - \frac{|E(U_{\tau})| + |E(V_{\tau})|}{2} + |V(\tau)| - \frac{|U_{\tau}| + |V_{\tau}|}{2})$
        \item $c(\tau) \leq n^{\frac{\epsilon}{16}(|E(\tau)| - \frac{|E(U_{\tau})| + |E(V_{\tau})|}{2} + |V(\tau)| - \frac{|U_{\tau}| + |V_{\tau}|}{2})}$ \item $B_{adjust}(\tau) \leq n^{\frac{\epsilon}{16}(|E(\tau)| - \frac{|E(U_{\tau})| + |E(V_{\tau})|}{2} + |V(\tau)| - \frac{|U_{\tau}| + |V_{\tau}|}{2})}$
    \end{enumerate}
\end{proof}

\begin{lemma}\label{lem:interactiontermfullbound}
\begin{align*}
    &\sum_{U,V \in \calI_{mid}: U \sim V}{\frac{|U \cap V|!}{(|U|!)^2}\sum_{\tau \in \mathcal{M}_{U, V}}{\sum_{j=1}^{\infty}{\sum_{(\Gamma,\Gamma'^\T,P) \in \mathcal{P}_j(\tau)}{\frac{(-1)^{\# \text{ of intersection indices in } [j]}}{\prod_{i=1}^{j}{|U_{\gamma_i}|! |V_{{\gamma'_i}^\T}|!}}}}}} \\
    &\matL_{U_{\tau_{P}}}\left(\frac{\left(\prod_{i=1}^{j}{c_{P_i}}N(P_i)\right)\lambda_{\tau_{P}}\mM_{\tau_{P}^+}}{|\Aut(\tau_{P})|}\right)\matL^{\T}_{V_{\tau_{P}}}\\
    &\succeq -\sum_{U \in \mathcal{I}_{mid}}
    \left(
    \max_{\substack{\tau \in \calM\\
    j \in \N^+\\
    (\Gam, \Gam', P) \in \calP_{j}(\tau):\\
    U_{\tau_P} = U}}
    \left\{100^j c(\tau) \left(\prod_{i=1}^j c(\gam_i)c(\gam_i')c(P_i)c_{P_i}N(P_i)\right)\frac{\lam_{\tau_P}\norm{\mM_{\tau_P}}}{\sqrt{\lam_{U_{\tau_P}} \norm{\mM_{U_{\tau_P}}} \lam_{V_{\tau_P}} \norm{\mM_{V_{\tau_P}}}}}\right\}\right)\\
    &\qquad\frac{\lambda_{U}}{|U|!}
    {\matL_{U}\mM_{U^+}\matL^\T_{U}}
\end{align*}
\end{lemma}

\begin{proof}
Applying \cref{lem:comparingtosquares} and using the trivial bound $|\Aut(\tau_P)|~\geq~1$,
\begin{align*}
    &\sum_{U,V \in \calI_{mid}: U \sim V}{\frac{|U \cap V|!}{(|U|!)^2}\sum_{\tau \in \mathcal{M}_{U,V}}{\sum_{j=1}^{\infty}{\sum_{(\Gamma,\Gamma'^\T,P) \in \mathcal{P}_j(\tau)}{\frac{(-1)^{\# \text{ of intersection indices in } [j]}}{\prod_{i=1}^{j}{|U_{\gamma_i}|! |V_{{\gamma'_i}^\T}|!}}}}}} \\
    &\matL_{U_{\tau_P}}\left(\frac{\left(\prod_{i=1}^{j}{c_{P_i}}N(P_i)\right)\lambda_{\tau_P}\mM_{\tau_P^+}}{|\Aut(\tau_P)|}\right)\matL^{\T}_{V_{\tau_P}} \\
    &\succeq -\sum_{U,V \in \calI_{mid}: U \sim V}{\frac{|U \cap V|!}{(|U|!)^2}\sum_{\tau \in \mathcal{M}_{U,V}}{\sum_{j=1}^{\infty}{\sum_{(\Gamma,\Gamma'^\T,P) \in \mathcal{P}_j(\tau)}{\frac{\left(\prod_{i=1}^{j}{c_{P_i}}N(P_i)\right)\frac{\lambda_{\tau_P}\norm{\mM_{\tau_P}}}{\sqrt{\lambda_{U_{\tau_P}}\norm{\mM_{U_{\tau_P}}}\lambda_{V_{\tau_P}}\norm{\mM_{V_{\tau_P}}}}}}{\prod_{i=1}^{j}{|U_{\gamma_i}|! |V_{{\gamma'_i}^\T}|!}}}}}} \\
    &\left(\frac{1}{2}\matL_{U_{\tau_P}}\lambda_{U_{\tau_P}}\mM_{U_{\tau_P}^+}\matL^{\T}_{U_{\tau_P}} +\frac{1}{2}\matL_{V_{\tau_P}}\lambda_{V_{\tau_P}}\mM_{V_{\tau_P}^+}\matL^{\T}_{V_{\tau_P}}\right) 
    \end{align*}
    We now show how to bound the $\matL_{U_{\tau_P}}\lambda_{U_{\tau_P}}\mM_{U_{\tau_P}^+}\matL^{\T}_{U_{\tau_P}}$ terms. The
    $\matL_{V_{\tau_P}}\lambda_{V_{\tau_P}}\mM_{V_{\tau_P}^+}\matL^{\T}_{V_{\tau_P}}$ terms can be bounded by a symmetrical argument.

    Grouping all of the terms where $U_{\tau_P} = U$ together, we obtain that
    \begin{align*}
    &\sum_{U,V \in \calI_{mid}: U \sim V}{\frac{|U \cap V|!}{(|U|!)^2}\sum_{\tau \in \mathcal{M}_{U,V}}{\sum_{j=1}^{\infty}{\sum_{(\Gamma,\Gamma'^\T,P) \in \mathcal{P}_j(\tau)}{\frac{1}{\prod_{i=1}^{j}{|U_{\gamma_i}|! |V_{{\gamma'_i}^\T}|!}}}}}} \\
    &\left(\prod_{i=1}^{j}{c_{P_i}}N(P_i)\right)\frac{\lambda_{\tau_P}\norm{\mM_{\tau_P}}}{\sqrt{\lambda_{U_{\tau_P}}\norm{\mM_{U_{\tau_P}^+}}\lambda_{V_{\tau_P}}\norm{\mM_{V_{\tau_P}}}}}\matL_{U_{\tau_P^+}}\lambda_{U_{\tau_P}}\mM_{U_{\tau_P}}\matL^{\T}_{U_{\tau_P}} \\
    &= \sum_{U \in \calI_{mid}}\frac{1}{|U|!}\sum_{\tau \in \calM}\sum_{j=1}^{\infty}
    \sum_{\substack{(\Gamma,\Gamma'^\T,P) \in \mathcal{P}_j(\tau):\\ U_{\tau_P} = U}}\frac{|U_\tau \cap V_\tau|!}{|V_{\tau}|!} \frac{1}{\prod_{i=1}^j|V_{\gamma_i}|!|V_{{\gamma'_i}^\T}|!}\\
    &\left(2^jc(\tau)\left(\prod_{i=1}^{j}{c(\gamma_i)c(\gamma'_i)c(P_i)c_{P_i}N(P_i)}\right)\frac{\lambda_{\tau_P}\norm{\mM_{\tau_P}}}{\sqrt{\lambda_{U_{\tau_P}}\norm{\mM_{U_{\tau_P}}}\lambda_{V_{\tau_P}}\norm{\mM_{V_{\tau_P}}}}}\right)\frac{\matL_{U}\lambda_{U}\mM_{U^+}\matL^{\T}_{U}}{2^jc(\tau)\prod_{i=1}^{j}{c(\gamma_i)c(\gamma'_i)c(P_i)}} \\
    &\preceq 
    \sum_{U \in \mathcal{I}_{mid}}
    \left(
    \max_{\substack{\tau \in \calM\\
    j \in \N^+\\
    (\Gam, \Gam', P) \in \calP_{j}(\tau):\\
    U_{\tau_P} = U}}
    \left\{100^j c(\tau) \left(\prod_{i=1}^j c(\gam_i)c(\gam_i')c(P_i)c_{P_i}N(P_i)\right)\frac{\lam_{\tau_P}\norm{\mM_{\tau_P}}}{\sqrt{\lam_{U_{\tau_P}} \norm{\mM_{U_{\tau_P}}} \lam_{V_{\tau_P}} \norm{\mM_{V_{\tau_P}}}}}\right\}\right)\\
    &\qquad{\frac{\lambda_{U}}{|U|!}\matL_{U}\mM_{U^+}\matL^{\T}_{U}}
    \end{align*}
    where the last inequality uses the following facts to convert
    the sums into a maximization:
    \begin{enumerate}
        \item For all $i \in [j]$, $\sum_{\gam_i \in \calL: U_{\gam_i} = V_{\gam_{i+1}}}\frac{1}{|V_{\gam_i}|!c(\gam_i)} \leq 2$ where we set $V_{\gam_{j+1}} = U$. Across all $i$, this multiplies the total by $2^j$.
        \item $\sum_{\tau \in \calM: U_\tau = V_{\gam_1}} \frac{|U_\tau \cap V_\tau|!}{|V_\tau|!c(\tau)} \leq 2$
        \item For all $i \in [j]$, $\sum_{\gam_i' \in \calL_{V_{\gam'^\T_{i-1}}}}\frac{1}{|V_{\gam'^\T_i}|!c(\gam'^\T_i)} \leq 2$. Across all $i$, this multiplies the total by $2^j$. 
        \item For all $i \in [j]$, $\sum_{P_i \in \mathcal{P}^{interact}_{\gamma_i,\tau_{P_{i-1}}, {\gamma'}^\T_i}}{\frac{1}{c(P_i)}} \leq 1$. 
        \item $\sum_{j=1}^{\infty}{\frac{1}{2^{j}}} \leq 1$.
    \end{enumerate}
\end{proof}

\intersectionTermBound*
\begin{proof}
    We need to show that 
    \[
    \max_{\substack{\tau \in \calM\\
    j \in \N^+\\
    (\Gam, \Gam', P) \in \calP_{j}(\tau):\\
    U_{\tau_P} = U}}
    \left\{100^j c(\tau) \left(\prod_{i=1}^j c(\gam_i)c(\gam_i')c(P_i)c_{P_i}N(P_i)\right)\frac{\lam_{\tau_P}\norm{\mM_{\tau_P}}}{\sqrt{\lam_{U_{\tau_P}} \norm{\mM_{U_{\tau_P}}} \lam_{V_{\tau_P}} \norm{\mM_{V_{\tau_P}}}}}\right\} \leq \frac{1}{4}
    \]
    This follows from the following observations:
    \begin{enumerate}
        \item $\frac{\prod_{i=1}^{j}c^\approx_{P_i}\lam_{\tau_P}\norm{\mM_{\tau_P}}}{\sqrt{\lam_{U_{\tau_P}} \norm{\mM_{U_{\tau_P}}} \lam_{V_{\tau_P}} \norm{\mM_{V_{\tau_P}}}}} = n^{-\slack(\tau_P)}$
        \item By the slack lower bound in \cref{assumptions},
        \begin{align*}
        \slack(\tau_P) &\geq \eps\left(E_{tot}(\tau_P) - \frac{|E(U_{\tau_P})| + |E(V_{\tau_P})|}{2} + |V_{tot}(\tau_P)| - \frac{|U_{\tau_P}| + |V_{\tau_P}|}{2}\right)\\
        &= \epsilon\left(|E(\tau)| - \frac{|E(U_{\tau})| + |E(V_{\tau})|}{2} + |V(\tau)| - \frac{|U_{\tau}| + |V_{\tau}|}{2}\right) + \\
        &\epsilon\sum_{i \in [j]}{}
        \Bigg(|E(\gam_i)| - \frac{|E(U_{\gam_i})| + |E(V_{\gam_i})|}{2} + (\# \text{ of edges removed from } \gam_i) \\
        &+ |V(\gam_i)| - \frac{|U_{\gam_i}| + |V_{\gam_i}|}{2} +|E({\gam'}^\T_i)| - \frac{|E(U_{{\gam'}^\T_i})| + |E(U_{{\gam'}^\T_i})|}{2} \\
        &+ (\# \text{ of edges removed from } {\gam'}^\T_i) + |V({\gam'}^\T_i)| - \frac{|U_{{\gam'}^\T_i}| + |V_{{\gam'}^\T_i}|}{2}\Bigg)
        \end{align*}        
        \item $c(\tau) \leq n^{\frac{\epsilon}{32}(|E(\tau)|- \frac{|E(U_{\tau})| + |E(V_{\tau})|}{2} + |V(\tau)| - \frac{|U_{\tau}| + |V_{\tau}|}{2})}$
        \item $B_{adjust}(\tau_P) \leq n^{\frac{\epsilon}{32}(|E(\tau_P)| - \frac{|E(U_{\tau_P})| + |E(V_{\tau_P})|}{2} + |V(\tau_P)| - \frac{|U_{\tau_P}| + |V_{\tau_P}|}{2})}$
        \item For all $i \in [j]$, $c(\gam_i)c(\gam_i')$, $c(P_i)$, and $N(P_i)$ are all at most $n$ raised to the power 
        \begin{align*}
        &\frac{\epsilon}{32}
        \Bigg(|E(\gam_i)| - \frac{|E(U_{\gam_i})| + |E(V_{\gam_i})|}{2} + (\# \text{ of edges removed from } \gam_i) \\
        &+ |V(\gam_i)| - \frac{|U_{\gam_i}| + |V_{\gam_i}|}{2} +|E({\gam'}^\T_i)| - \frac{|E(U_{{\gam'}^\T_i})| + |E(U_{{\gam'}^\T_i})|}{2} \\
        &+ (\# \text{ of edges removed from } {\gam'}^\T_i) + |V({\gam'}^\T_i)| - \frac{|U_{{\gam'}^\T_i}| + |V_{{\gam'}^\T_i}|}{2}\Bigg)
        \end{align*}
    \item $\abs{\tfrac{c_P}{c_P^\approx}} \leq 2$
    \end{enumerate}
\end{proof}

\subsection{\texorpdfstring{$c$}{c}-function bounds}
\label{sec:c-functions}

In this section we bound the various combinatorial functions.

\begin{definition}[$N^{shape}(U, e)$]
    Given a diagonal shape $U$ and $V \in \N$,
    let $N^{shape}(U, e)$ be the number of shapes $\alpha$ with
    $U_\alpha = U$ and $e$ edges outside of $U_\al$.
\end{definition}
\begin{remark}
    Since the permutation of $V_\al$ can be arbitrary,
    $N^{shape}(U, e)$ is a multiple of $|V_\al|!$.
\end{remark}

\begin{definition}[$c(\al)$, formal]
\label{def:c-functions}
    \begin{align*}
    c(\al) =& 2^{|E(\al) \setminus E(U_\al \cap V_\al)|} \cdot \frac{1}{|U_\al \cap V_\al|!}\max\left\{N^{shape}(U_\al, |E(\al)\setminus E(U_\al)|), N^{shape}(V_\al, |E(\al)\setminus E(V_\al)|)\right\}
    \end{align*}
\end{definition}

\begin{lemma}
    For all diagonal shapes $U$, 
    \begin{align*}
    \sum_{\text{shapes }\al: U_\al = U, \al\text{ non-trivial}}\frac{1}{|U \cap V_{\alpha}|!c(\al)} &\leq 1\\
    \sum_{\text{shapes }\al: U_\al = U}\frac{1}{|U \cap V_{\alpha}|!c(\al)} &\leq 2
    \end{align*}
    By symmetry the same holds for the sum over $\al : V_\al = V$.
\end{lemma}
\begin{proof}
\begin{align*}
   &\sum_{\text{shapes }\al: U_\al = U}\frac{1}{|U \cap V_{\alpha}|!c(\al)}\\
   & = \sum_{e=0}^\infty \sum_{\text{shapes }W} \sum_{\substack{\text{shapes }\al: U_\al = U,\\
   |E(\al) \setminus E(U_\al)| = e}} \frac{1}{2^eN^{shape}(U, e)}\\
   &= \sum_{e = 0}^\infty \frac{1}{2^e}\\
   &= 2\,.
\end{align*}
To derive the first statement, note that the trivial shapes
contribute exactly 1 to the sum.
\end{proof}

\begin{lemma}[Bound for $c(\al)$]
\label{lem:c-shape}
    For all shapes $\al$ with at most $D_V$ vertices,
    \[c(\al) \leq 2 (4D_V)^{2|E(\al) \setminus E(U_\al \cap V_\al)|} (2D_V)^{2|(U_\al \cup V_\al) \setminus (U_\al \cap V_\al)|}\]
\end{lemma}
\begin{proof}
    The shapes counted by $N^{shape}(U, e)$ can be generated by the following process.
    \begin{enumerate}
        \item Start from $V(\alpha) = U_\alpha = U$.
        \item Run the following process to select a subset of $U_\al$ to be in $U_\al \cap V_\al$. Use a label in [2]
        to decide whether or not at least one vertex is in $U_\al \cap V_\al$. If so, use a label in $|U_\al|$
        to choose the vertex, and then use a label in $[2]$
        to decide whether or not another vertex is in $U_\al \cap V_\al$, and so forth.
        \item For each edge outside $E(U_\al)$, identify each endpoint using a label in $[D_V]$, and additionally use a label in $[2]$ to identify whether each endpoint is in $V_\al$.
        \item Specify the permutation of $V_\al$ in $|V_\al|!$ ways.
    \end{enumerate}
    In total, this is at most
    $2(2D_V)^{ 2|E(\al) \setminus E(U_\al)|} (2|U_\al|)^{|(U_\al \cup V_\al) \setminus (U_\al \cap V_\al)|}|V_\al|!.$
    A symmetric bound applies to $N^{shape}(V, e)$.
    Therefore,
    \begin{align*}
        c(\al) &\leq 2^{|E(\al) \setminus E(U_\al \cap V_\al)|} \cdot \frac{\max\{|U_\al|!, |V_\al|!\}}{|U_\al \cap V_\al|!} \cdot 2(2D_V)^{ 2|E(\al) \setminus E(U_\al \cap V_\al)|}(2D_V)^{|(U_\al \cup V_\al) \setminus (U_\al \cap V_\al)|}\\
        &\leq 2^{|E(\al) \setminus E(U_\al \cap V_\al)|} \cdot D_V^{|(U_\al \cup V_\al) \setminus (U_\al \cap V_\al)|} \cdot 2(2D_V)^{ 2|E(\al) \setminus E(U_\al \cap V_\al)|}(2D_V)^{|(U_\al \cup V_\al) \setminus (U_\al \cap V_\al)|}\\
        &= 2 (4D_V)^{2|E(\al) \setminus E(U_\al \cap V_\al)|} (2D_V)^{2|(U_\al \cup V_\al) \setminus (U_\al \cap V_\al)|}
    \end{align*}
    as needed.
\end{proof}

\begin{definition}[$N^{PMVS}(\gam, \tau, \gam'^\T, e)$]
Given shapes $\gam, \tau, \gam'^\T$ and $e \in \N$, let $N^{PMVS}(\gam, \tau, \gam'^\T, e)$ be the number of PMVS interaction patterns such that $e$ edges are removed from $\tau_P$, either because of the adding indicators step or the removing middle edge indicators step.

Similarly, let $N^{intersect}(\gam, \tau, \gam'^\T, e)$ be the number of intersection interaction patterns such that $e$ edges are removed from $\tau_P$ in the removing middle edge indicators step.
\end{definition}

\begin{definition}[$c(P)$, formal]
    For an interaction pattern $P \in \calP^{interact}_{\gam, \tau, \gam'^\T}$, let $c(P) = 2^{e+1}N^{PMVS}(\gam, \tau, \gam'^\T, e)$ if $P$ is a PMVS interaction pattern and $c(P) = 2^{e+2}N^{intersect}(\gam, \tau, \gam'^\T, e)$ if $P$ is an intersection interaction pattern, where $e$ is the number of edges removed from $\tau_P$.
\end{definition}
\begin{lemma}
    For all $\gam, \tau, \gam'$, $\sum_{P \in \mathcal{P}^{interact}_{\gamma,\tau, {\gamma'}^\T}}{\frac{1}{c(P)}} \leq 1$.
\end{lemma}
\begin{proof}
    \begin{align*}
    \sum_{P \in \mathcal{P}^{interact}_{\gamma,\tau, {\gamma'^\T}}}{\frac{1}{c(P)}} &=  \sum_{e=1}^{\infty}\sum_{\substack{P \in \mathcal{P}^{PMVS}_{\gamma,\tau, {\gamma'^\T}}:\\
    {e \text{ edges are removed from } \tau_P}}}\frac{1}{2^{e+1}N^{PMVS}(\gam, \tau, \gam'^\T, e)}\\
    &+\sum_{e=0}^{\infty}\sum_{\substack{P \in \mathcal{P}^{intersect}_{\gamma,\tau, {\gamma'^\T}}:\\
    {e \text{ edges are removed from } \tau_P}}}\frac{1}{2^{e+2}N^{intersect}(\gam, \tau, \gam'^\T, e)}\\
    &\leq 1
    \end{align*}
\end{proof}

\begin{lemma}[Bound for $c(P)$] For all shapes $\gam, \tau, \gam'^\T$ such that $|V(\gam)| \leq D_V$, $|V(\tau)| \leq 3D_V$, and $|V(\gam'^\T)| \leq D_V$, 
\begin{enumerate}
    \item For all PMVS interaction patterns such that $e$ edges are removed from $\tau_P$, $c(P) \leq 2(4D_V^2)^{e}$
    \item For all intersection interaction patterns such that $e$ edges are removed from $\tau_P$, 
    \[
        c(P) \leq 4(3D_V)^{|V(\gam)\setminus V_\gam| + |V(\gam'^\T)\setminus U_{\gam'^\T}|}2^{e + |E(\gam) \setminus E(U_{\gam} \cap V_\gam)| + |E({\gam'}^\T) \setminus E(U_{{\gam'}^\T} \cap V_{{\gam'}^\T})|}
    \]
\end{enumerate}
\end{lemma}
\begin{proof}
For the case of a PMVS interaction pattern, we do the following.
\begin{enumerate}
    \item We know that at least one edge must be removed in the adding edge indicators step for non-terminal $P$. For each such edge, we specify the endpoints for a cost of at most $D_V^2$.
    \item For the removing middle edge indicators step, we can specify each edge which is removed by specifying 
    its two endpoints at a cost of $D_v^2$ per edge.
\end{enumerate}

For the case of an intersection interaction pattern, we do the following.
    \begin{enumerate}
        \item Go through each vertex in $V(\gam)\setminus V_\gam$ and $V(\gam'^\T)\setminus U_{\gam'^\T}$ and indicate which vertex they intersect with, if any. This has a cost of $(3D_V)^{|V(\gam)\setminus V_\gam| + |V(\gam'^\T)\setminus U_{\gam'^\T}|}$
        \item For each edge that  intersected, use a label in $[2]$ to denote its multiplicity after linearization. This has a cost of at most $2^{|E(\gam)\setminus E(V_\gam)| + |E(\gam'^\T)\setminus E(U_{\gam'^\T} )|}$.
        \item For each edge in $V_\gam \cup U_{{\gam'}^{\T}}$ which is not in $U_{\gam} \cup V_{{\gam'}^\T}$, use a label
        in $[2]$ to decide whether it is removed in the \textbf{Remove middle edge indicators operation}.
        This has a cost of at most $2^{|E(V_\gam) \setminus E(U_{\gam})| + |E(U_{{\gam'}^\T}) \setminus E(V_{{\gam'}^\T})|}$.
    \end{enumerate}
\end{proof}

\begin{lemma}[Bound for $c_P$] 
\label{lem:cp-bound}
    The excess in $c_P$ over what goes into the slack is 
    \[\abs{\tfrac{c_P}{c_P^\approx}} \leq 2\,.\]
\end{lemma}
\begin{proof}
    By \cref{lem:cP-adjustment}, 
    \begin{align*}
    \abs{\tfrac{c_P}{c_P^\approx}} \leq \left(\frac{1}{1-p}\right)^{\#\text{ of indicators}} \leq 
    \left(1 + O(p)\right)^{D_V^2} \leq 2
    \end{align*}
    provided $n$ is sufficiently large.
\end{proof}

\begin{lemma}[Bound for $N(P)$]
For all $\gam, \tau, \gam'^\T$ with size at most $D_V$ and $P \in 
\calP^{interact}_{\gam, \tau, \gam'^\T}$
\[
N(P) \leq (3D_V)^{|V(\gam) \setminus U_{\gam}| + |V({\gam'}^\T) \setminus V_{{\gam'}^\T}|}
\]
\end{lemma}
\begin{proof}
    We are given $\gam,\tau,{\gam'}^\T$, the interaction pattern $P$, and the resulting ribbon $R'_2$ and we need to specify the ribbons $G, R_2, {G'}^{\T}$ which have shapes $\gam,\tau,{\gam'}^\T$, have interaction pattern $P$, and result in the ribbon $R'_2$.

    Suppose that $P$ is a PMVS interaction.
    To do this, it is sufficient to specify how the vertices in $\gam$ and ${\gam'}^{\T}$ are mapped to in $R'_2$. This specifies the ribbons $G$ and ${G'}^{\T}$.
    Either $A_{R_2} = B_{G}$ and $B_{R_2} = A_{G'^\T}$ together with the remaining unmapped vertices of $R'_2$ have shape $\tau$, in which case
    this is a possible ribbon $R_2$, or they do not, in which case
    this is merely overcounting.

    We automatically have that $A_{G} = A_{R'_2}$ and $B_{{G'}^{\T}} = B_{R'_2}$ so we do not need to specify where the vertices $U_{\gam}$ and $V_{{\gam'}^{\T}}$ are mapped to. For each of the remaining vertices, the number of choices is at most $3D_V$ so the total number of choices is $(3D_V)^{|V(\gam) \setminus U_{\gam}| + |V({\gam'}^\T) \setminus V_{{\gam'}^\T}|}$, as needed.

    For an intersection term interaction, the same analysis goes
    through, with the added constraint that the intersection pattern
    along with the mappings of $\gam$ and $\gam'^\T$ fix additional labels of $R_2$.
\end{proof}

\subsection{Truncation error}
\label{sec:truncation}

\begin{definition}[$\idsym$]
    \begin{align*}
    \idsym[I, J] = \begin{cases}
        1 & I = J\text{ as unordered sets}\\
        0 & \text{otherwise}
    \end{cases}    
    \end{align*}
\end{definition}

\begin{lemma}
$\trunc_1 \psdleq n^{\dsos + \eta - \frac{\epsilon}{16}(D_V - 2\dsos)}\Id_{sym}$
\end{lemma}
\begin{proof}
Applying \cref{thm:conditionednormbounds} with $D = 3D_V$, for all shapes $\alpha$ such that $D_V \leq |V(\alpha)| \leq 3D_V$, $|U_{\alpha}| \leq \dsos$, $|V_{\alpha}| \leq \dsos$, and $\alpha$ has no isolated vertices outside of $U_{\alpha} \cup V_{\alpha}$,
\begin{align*}
    \lambda_{\alpha}||\mM_{\alpha}|| &\leq 2B_{adjust}(\alpha)n^{(1 - \alpha)\frac{|U_{\alpha}| + |V_{\alpha}|}{2} + \eta - (\gamma - {\alpha}{\beta} - 3\log_n(3D_V))|E(\alpha)|} \\
    &\leq 2B_{adjust}(\alpha)n^{(1 - \epsilon)\dsos - \frac{\epsilon}{8}(D_V - 2\dsos) - \frac{\epsilon}{4}(|E(U_{\alpha})| + |E(V_{\alpha})|)}
\end{align*}
as $|E(\alpha)| \geq |E(U_{\alpha})| + |E(V_{\alpha})| + \frac{|V(\alpha)| - |U_{\alpha}| - |V_{\alpha}|}{2} \geq |E(U_{\alpha})| + |E(V_{\alpha})| + \frac{D_V - 2\dsos}{2}$ and $\gamma - {\alpha}{\beta} - 3\log_n(3D_V) \geq \frac{\epsilon}{4}$.

We now observe that
    \begin{align*}
    \trunc_1 &= \sum_{\substack{U, V \in \calI_{mid}:\\ U \sim V}} 
    \frac{|U\cap V|!}{(|U|!)^2}
    \sum_{\substack{\sig \in \calL_{U,\leq D_V}\\ \tau \in \calM_{U, V, \leq D_V}\\ \sig' \in \calL_{V,\leq D_V}:\\ |V(\sig^- \circ \tau \circ (\sig'^\T)^-)| > D_V}}
    \frac{\lam_{\sig^- \circ \tau \circ (\sig'^\T)^-} \mM_{\sig^- \circ \tau \circ (\sig'^\T)^-}}{|\Aut(\sig^- \circ \tau \circ (\sig'^\T)^-)|} \\
    &\psdleq \left(\sum_{\substack{\text{shape }\alpha: \\
    D_V < |V(\alpha)| \leq 3D_V,\\ |U_{\alpha}| \leq \dsos, |V_{\alpha}| \leq \dsos}}{\lambda_{\alpha}||\mM_{\alpha}||}\right)\Id_{Sym} \\
    &\psdleq \left(\sum_{\substack{\text{shape }\alpha: \\
    D_V < |V(\alpha)| \leq 3D_V,\\ |U_{\alpha}| \leq \dsos}}{\lambda_{\alpha}||\mM_{\alpha}||}\right)\Id_{Sym} \\
    &\psdleq \sum_{U \in \mathcal{I}_{mid}: |U| \leq \dsos}{\left(\sum_{\substack{\text{shape }\alpha:\\
    D_V < |V(\alpha)| \leq 3D_V,\\ U_{\alpha} = U}}{\frac{\dsos!}{|U_{\alpha}|!c(\alpha)}c(\alpha)\lambda_{\alpha}||\mM_{\alpha}||}\right)}\Id_{Sym} \\
    &\psdleq \sum_{U \in \mathcal{I}_{mid}: |U| \leq \dsos}{\left(\max_{\substack{\text{shape }\alpha:\\ D_V < |V(\alpha)| \leq 3D_V,\\ U_{\alpha} = U}}{\left\{c(\alpha)B_{adjust}(\alpha)\right\}}\right)}\\
    &\quad{\dsos}!n^{\dsos + \eta - \frac{\epsilon}{8}(D_V - 2\dsos) - \frac{\epsilon}{4}|E(U)|}\Id_{Sym} \\
    &\psdleq n^{\dsos + \eta - \frac{\epsilon}{16}(D_V - 2\dsos)}\Id_{sym}
    \end{align*}
\end{proof}

    We now analyze the second part of the truncation error. 
    \begin{lemma}
        $\trunc_2 \succeq -n^{2\dsos + 2\eta - \frac{\epsilon}{32}D_V}\Id_{Sym}$
    \end{lemma}
    \begin{proof}
    Recall that 
    \begin{align*}
    \trunc_2 =& \sum_{U,V \in \calI_{mid}: U \sim V}{\frac{|U \cap V|!}{(|U|!)^2}\sum_{\tau \in \mathcal{M}_{U, V}}{\sum_{j=1}^{\infty}{\sum_{(\Gamma,\Gamma'^\T,P) \in \mathcal{P}_j(\tau)}{\frac{(-1)^{\# \text{ of intersection indices in } [j]}}{\prod_{i=1}^{j}{|U_{\gamma_i}|! |V_{{\gamma'_i}^\T}|!}}}}}}\\
    &\left(\left(\matL_{U_{\tau_{P}},\leq D_L(P)} - \matL_{U_{\tau_{P}},\leq D_V}\right)\left(\frac{\left(\prod_{i=1}^{j}{c_{P_i}}N(P_i)\right)\lambda_{\tau_{P}}\mM_{\tau_{P}^+}}{|\Aut(\tau_{P})|}\right)\matL^{\T}_{V_{\tau_{P}},\leq D_V}\right. \\
    +\;&\left.\matL_{U_{\tau_{P}},\leq D_L(P)}\left(\frac{\left(\prod_{i=1}^{j}{c_{P_i}}N(P_i)\right)\lambda_{\tau_{P}}\mM_{\tau_{P}^+}}{|\Aut(\tau_{P})|}\right)\left(\matL^{\T}_{V_{\tau_{P}},\leq D_R(P)} - \matL^{\T}_{V_{\tau_{P}},\leq D_V}\right)\right)
    \end{align*}
    so 
    \begin{align*}
    \trunc_2 \succeq& -\sum_{U,V \in \calI_{mid}: U \sim V}{\frac{|U \cap V|!}{(|U|!)^2}\sum_{\tau \in \mathcal{M}_{U, V}}{\sum_{j=1}^{\infty}{\sum_{(\Gamma,\Gamma'^\T,P) \in \mathcal{P}_j(\tau)}{\frac{1}{\prod_{i=1}^{j}{|U_{\gamma_i}|! |V_{{\gamma'_i}^\T}|!}}}}}}\\
    &\left(\norm{\matL_{U_{\tau_{P}},\leq D_L(P)} - \matL_{U_{\tau_{P}},\leq D_V}}\left(\frac{\left(\prod_{i=1}^{j}{c_{P_i}}N(P_i)\right)\lambda_{\tau_{P}}\norm{\mM_{\tau_{P}}}}{|\Aut(\tau_{P})|}\right)\norm{\matL^{\T}_{V_{\tau_{P}},\leq D_V}}\right. \\
    +\;&\left.\norm{\matL_{U_{\tau_{P}},\leq D_L(P)}}\left(\frac{\left(\prod_{i=1}^{j}{c_{P_i}}N(P_i)\right)\lambda_{\tau_{P}}\norm{\mM_{\tau_{P}}}}{|\Aut(\tau_{P})|}\right)\norm{\matL^{\T}_{V_{\tau_{P}},\leq D_R(P)} - \matL^{\T}_{V_{\tau_{P}},\leq D_V}}\right)\Id_{Sym}
    \end{align*}
    We can analyze this using the following claims.
    \begin{claim}
        For all $U \in \mathcal{I}_{mid}$, 
        \[
            \max_{D: D \leq D_V}{\left\{\norm{\matL_{U,\leq D}}\right\}}\sqrt{\lambda_{U}\norm{M_U}} \leq n^{\dsos + \eta - \frac{\epsilon}{16}|E(U)|}
        \]
    \end{claim}
    \begin{proof}
        By \cref{cor:sigmaminusnormboundtwo}, for all $\sigma \in \matL_{U,\leq D_V}$
            \[
                \lambda_{\sigma^{-}}||\mM_{\sigma^{-}}||\sqrt{\lambda_{U}||M_{U}||} \leq 2B_{adjust}(\sigma)n^{\dsos + \eta - \frac{\epsilon}{2}\dsos - \frac{\epsilon}{8}|E(\sigma)| - \frac{\epsilon}{8}|V(\sigma)|}
            \]
        We now observe that 
        \begin{align*}
            &\max_{D: D \leq D_V}{\left\{\norm{\matL_{U,\leq D}}\right\}}\sqrt{\lambda_{U}\norm{M_U}} \leq \sum_{\sigma \in \matL_{U,\leq D_V}}{\lambda_{\sigma^{-}}||\mM_{\sigma^{-}}||\sqrt{\lambda_{U}||M_{U}||}} \\
            &\leq {\dsos}!\left(\sum_{\sigma \in \matL_{U,\leq D_V}}{\frac{1}{c(\sigma)|U_{\sigma}|!}}\right)\max_{\sigma \in \matL_{U,\leq D_V}}{\left\{c(\sigma)\lambda_{\sigma^{-}}||\mM_{\sigma^{-}}||\sqrt{\lambda_{U}||M_{U}||}\right\}}\\
            &\leq 4{\dsos}!\max_{\sigma \in \matL_{U,\leq D_V}}{\left\{c(\sigma)B_{adjust}(\sigma)n^{\dsos + \eta - \frac{\epsilon}{2}\dsos - \frac{\epsilon}{8}|E(\sigma)| - \frac{\epsilon}{8}|V(\sigma)|}\right\}}\\
            &\leq n^{\dsos + \eta - \frac{\epsilon}{16}|E(U)|}
        \end{align*}
    \end{proof}

    \begin{claim}
        For all $U \in \mathcal{I}_{mid}$ and all $D \leq D_V$,  
        \[
        \norm{\matL_{U,\leq D} - \matL_{D,\leq D_V}}\sqrt{\lambda_{U}\norm{\mM_{U}}} \leq n^{\dsos + \eta - \frac{\epsilon}{16}D - \frac{\epsilon}{16}|E(U)|}
        \]
    \end{claim}
    \begin{proof}
        By \cref{cor:sigmaminusnormboundtwo}, for all $\sigma \in \matL_{U,\leq D_V}$
            \[
                \lambda_{\sigma^{-}}||\mM_{\sigma^{-}}||\sqrt{\lambda_{U}||\mM_{U}||} \leq 2B_{adjust}(\sigma)n^{\dsos + \eta - \frac{\epsilon}{2}\dsos - \frac{\epsilon}{8}|E(\sigma)| - \frac{\epsilon}{8}|V(\sigma)|}
            \]
        We now observe that 
        \begin{align*}
            &\norm{\matL_{U,\leq D} - \matL_{U,\leq D_V}}\sqrt{\lambda_{U}\norm{\mM_{U}}} \leq \sum_{\sigma \in \matL_{U,\leq D_V}: |V(\sigma)| > D}{\lambda_{\sigma^{-}}||\mM_{\sigma^{-}}||\sqrt{\lambda_{U}||\mM_{U}||}} \\
            &\leq {\dsos}!\left(\sum_{\sigma \in \matL_{U,\leq D_V}}{\frac{1}{c(\sigma)|U_{\sigma}|!}}\right)\max_{\sigma \in \matL_{U,\leq D_V}: |V(\sigma)| > D}{\left\{c(\sigma)\lambda_{\sigma^{-}}||\mM_{\sigma^{-}}||\sqrt{\lambda_{U}||\mM_{U}||}\right\}}\\
            &\leq 4{\dsos}!\max_{\sigma \in \matL_{U,\leq D_V}: |V(\sigma)| > D}{\left\{c(\sigma)B_{adjust}(\sigma)n^{\dsos + \eta - \frac{\epsilon}{2}\dsos - \frac{\epsilon}{8}|E(\sigma)| - \frac{\epsilon}{8}|V(\sigma)|}\right\}} \\
            &\leq n^{\dsos + \eta - \frac{\epsilon}{16}D - \frac{\epsilon}{16}|E(U)|}
        \end{align*}
    \end{proof}
    Using these claims and grouping all of the terms where $U_{\tau_P} = U$ together in the same way as in the proof of \cref{lem:interactiontermfullbound}, we obtain that
    \begin{align*}
    &\sum_{U,V \in \calI_{mid}: U \sim V}{\frac{|U \cap V|!}{(|U|!)^2}\sum_{\tau \in \mathcal{M}_{U, V}}{\sum_{j=1}^{\infty}{\sum_{(\Gamma,\Gamma'^\T,P) \in \mathcal{P}_j(\tau)}{\frac{1}{\prod_{i=1}^{j}{|U_{\gamma_i}|! |V_{{\gamma'_i}^\T}|!}}}}}}\\
    &\left(\norm{\matL_{U_{\tau_{P}},\leq D_L(P)} - \matL_{U_{\tau_{P}},\leq D_V}}\left(\frac{\left(\prod_{i=1}^{j}{c_{P_i}}N(P_i)\right)\lambda_{\tau_{P}}\norm{\mM_{\tau_{P}}}}{|\Aut(\tau_{P})|}\right)\norm{\matL^{\T}_{V_{\tau_{P}},\leq D_V}}\right. \\
    &+\left.\norm{\matL_{U_{\tau_{P}},\leq D_L(P)}}\left(\frac{\left(\prod_{i=1}^{j}{c_{P_i}}N(P_i)\right)\lambda_{\tau_{P}}\norm{\mM_{\tau_{P}}}}{|\Aut(\tau_{P})|}\right)\norm{\matL^{\T}_{V_{\tau_{P}},\leq D_R(P)} - \matL^{\T}_{V_{\tau_{P}},\leq D_V}}\right)\Id_{Sym} \\
    \preceq& 2\sum_{U \in \calI_{mid}}\frac{1}{|U|!}\sum_{\tau \in \calM}\sum_{j=1}^{\infty}
    \sum_{\substack{(\Gamma,\Gamma'^\T,P) \in \mathcal{P}_j(\tau):\\ U_{\tau_P} = U}}\frac{|U_\tau \cap V_\tau|!}{|V_{\tau}|!} \frac{1}{\prod_{i=1}^j|V_{\gamma_i}|!|V_{{\gamma'_i}^\T}|!}\\
    &\left(2^jc(\tau)\left(\prod_{i=1}^{j}{c(\gamma_i)c(\gamma'_i)c(P_i)c_{P_i}N(P_i)}\right)\frac{\lambda_{\tau_P}\norm{\mM_{\tau_P}}n^{-\frac{\epsilon}{16}\min{\{D_L(P),D_R(P)\}}}}{\sqrt{\lambda_{U_{\tau_P}}\norm{\mM_{U_{\tau_P}}}\lambda_{V_{\tau_P}}\norm{\mM_{V_{\tau_P}}}}}\right)\frac{n^{2\dsos + 2\eta - \frac{\epsilon}{16}|E(U)|}\Id_{Sym}}{2^jc(\tau)\prod_{i=1}^{j}{c(\gamma_i)c(\gamma'_i)c(P_i)}} \\
    \preceq &
    \sum_{U \in \mathcal{I}_{mid}}
    \left(
    \max_{\substack{\tau \in \calM\\
    j \in \N^+\\
    (\Gam, \Gam', P) \in \calP_{j}(\tau):\\
    U_{\tau_P} = U}}
    \left\{100^j c(\tau) \left(\prod_{i=1}^j c(\gam_i)c(\gam_i')c(P_i)c_{P_i}N(P_i)\right)\frac{\lam_{\tau_P}\norm{\mM_{\tau_P}}n^{-\frac{\epsilon}{16}\min{\left\{D_L(P),D_R(P)\right\}}}}{\sqrt{\lam_{U_{\tau_P}} \norm{\mM_{U_{\tau_P}}} \lam_{V_{\tau_P}} \norm{\mM_{V_{\tau_P}}}}}\right\}\right)\\
    &n^{2\dsos + 2\eta - \frac{\epsilon}{16}|E(U)|}\Id_{Sym}
    \end{align*}
    We now make the same observations as before together with an observation about $D_L(P)$ and $D_R(P)$:
    \begin{enumerate}
        \item $\frac{\prod_{i=1}^{j}c^\approx_{P_i}\lam_{\tau_P}\norm{\mM_{\tau_P}}}{\sqrt{\lam_{U_{\tau_P}} \norm{\mM_{U_{\tau_P}}} \lam_{V_{\tau_P}} \norm{\mM_{V_{\tau_P}}}}} = n^{-\slack(\tau_P)}$
        \item By the slack lower bound in \cref{assumptions},
        \begin{align*}
        \slack(\tau_P) &\geq \eps\left(E_{tot}(\tau_P) - \frac{|E(U_{\tau_P})| + |E(V_{\tau_P})|}{2} + |V_{tot}(\tau_P)| - \frac{|U_{\tau_P}| + |V_{\tau_P}|}{2}\right)\\
        &= \eps\left(|E(\tau)| - \frac{|E(U_{\tau})| + |E(V_{\tau})|}{2} + |V(\tau)| - \frac{|U_{\tau}| + |V_{\tau}|}{2}\right) + \\
        &\epsilon\sum_{i \in [j]}{}
        \Bigg(|E(\gam_i)| - \frac{|E(U_{\gam_i})| + |E(V_{\gam_i})|}{2} + (\# \text{ of edges removed from } \gam_i) \\
        &+ |V(\gam_i)| - \frac{|U_{\gam_i}| + |V_{\gam_i}|}{2} +|E({\gam'}^\T_i)| - \frac{|E(U_{{\gam'}^\T_i})| + |E(U_{{\gam'}^\T_i})|}{2} \\
        &+ (\# \text{ of edges removed from } {\gam'}^\T_i) + |V({\gam'}^\T_i)| - \frac{|U_{{\gam'}^\T_i}| + |V_{{\gam'}^\T_i}|}{2}\Bigg)
        \end{align*}
        \item $D_L(P) \geq D_V - |V(\gamma) \setminus V_{\gam}|$ and $D_R(P) \geq D_V - |V({\gam'}^\T) \setminus U_{V({\gam'}^\T)}|$ so 
        \[
        \min{\{D_L(P),D_R(P)\}} \geq D_V - 2\sum_{i=1}^{j}{\left(|V(\gam_i)| - \frac{|U_{\gam_i}| + |V_{\gam_i}|}{2} + |V({\gam'}^\T_i)| - \frac{|U_{{\gam'}^\T_i}| + |V_{{\gam'}^\T_i}|}{2}\right)}
        \]
        \item $c(\tau) \leq n^{\frac{\epsilon}{32}(|E(\tau)|- \frac{|E(U_{\tau})| + |E(V_{\tau})|}{2} + |V(\tau)| - \frac{|U_{\tau}| + |V_{\tau}|}{2})}$
        \item $B_{adjust}(\tau_P) \leq n^{\frac{\epsilon}{32}(|E(\tau_P)| - \frac{|E(U_{\tau_P})| + |E(V_{\tau_P})|}{2} + |V(\tau_P)| - \frac{|U_{\tau_P}| + |V_{\tau_P}|}{2})}$
        \item For all $i \in [j]$, $c(\gam_i)c(\gam_i')$, $c(P_i)$, and $N(P_i)$ are all at most $n$ raised to the power
        \begin{align*}
        &\frac{\epsilon}{32}\sum_{i \in [j]}{}
        \Bigg(|E(\gam_i)| - \frac{|E(U_{\gam_i})| + |E(V_{\gam_i})|}{2} + (\# \text{ of edges removed from } \gam_i) \\
        &+ |V(\gam_i)| - \frac{|U_{\gam_i}| + |V_{\gam_i}|}{2} +|E({\gam'}^\T_i)| - \frac{|E(U_{{\gam'}^\T_i})| + |E(U_{{\gam'}^\T_i})|}{2} \\
        &+ (\# \text{ of edges removed from } {\gam'}^\T_i) + |V({\gam'}^\T_i)| - \frac{|U_{{\gam'}^\T_i}| + |V_{{\gam'}^\T_i}|}{2}\Bigg)
        \end{align*}
    \item $\abs{\tfrac{c_P}{c_P^\approx}} \leq 2$
    \end{enumerate}
\end{proof}

\subsection{Well-conditionedness of \texorpdfstring{$\matL$}{L}}

The goal of this section is to prove a lower bound on the minimum nonzero eigenvalue of $\matL \matL^\T$. More specifically we will prove the following lemma:
\begin{lemma}[Well-conditionedness of $\matL$]
\label{lem:well-conditionedness}
    \[\sum_{V \in \calI_{mid}} \frac{\lam_V}{|V|!}  \matL_{V} \mM_{V^+}\matL_{V}^\T \psdgeq \Omega(n^{-\dsos})\matId_{sym}\]
\end{lemma}

    The approach we take to \cref{lem:well-conditionedness} is as follows. If we can find nonnegative weights $\{w_V: V \in \calI_{mid}\}$ such that
    \[\sum_{\substack{V \in \calI_{mid}}}  \frac{w_V\lam_V}{|V|!}\matL_{V}\mM_{V^+}\matL_{V}^\T \psdgeq \idsym\]
    then since each term is individually PSD, the left-hand side is PSD-dominated
    by 
    \[\left(\max_{V \in \calI_{mid}} w_V\right)\cdot \sum_{\substack{V \in \calI_{mid}}}  \frac{\lam_V}{|V|!}\matL_{V}\mM_{V^+}\matL_{V}^\T\,.\]
    This implies the lower bound 
    \[ \sum_{\substack{V \in \calI_{mid}}}  \frac{\lam_V}{|V|!}\matL_{V}\mM_{V^+}\matL_{V}^\T \psdgeq \frac{1}{\displaystyle\max_{V \in \calI_{mid}} w_V} \idsym\,.\]
We will therefore seek an appropriate choice of $w_V$ that does not grow too quickly.

We will choose $w_U = 0$ unless $E(U) = \emptyset$.
The following lemma shows a growth condition which is sufficient.
Since $|U_\sig| > |V_\sig|$ for all non-diagonal left shapes $\sig$, we can use this lemma to define $w_U$ in order of increasing size $|U|$.

\begin{lemma}\label{lem:well-conditionedness-weights}
If we have nonnegative weights $\{w_V: V \in \mathcal{I}_{mid}, E(V) = \emptyset\}$ such that for all $U,V \in \mathcal{I}_{mid}$ with $E(V) = \emptyset$,
\[
w_V\max_{\substack{\text{nontrivial}\\\sigma \in \mathcal{L}_{V}:U_{\sigma} = U}}{\left\{2c(\sigma)\lambda_{\sigma}||\mM_{\sigma}||\right\}} \leq \frac{w_{U}}{2}
\]
then
\[
\sum_{\substack{V \in \mathcal{I}_{mid}:\\ E(V) = \emptyset}}{\frac{w_V}{|V|!}\matL_{V}\matL_{V}^\T} \succeq \frac{1}{2}\idsym\,.
\]
\end{lemma}
\begin{proof}
    
$\matL_{V}$ consists of the trivial shape with $V(\sig) = V$, as well as larger off-diagonal shapes. We use the following definition and claim to bound the off-diagonal shapes.

\begin{definition}[$\matId_{Sym, V}$]
    For $V \in \calI_{mid}$, let $\matId_{Sym, V}$
    be the restriction of $\idsym$ to the degree $|V|$-by-$|V|$ block.
\end{definition}
\begin{claim}\label{lem:boundingoffdiagonalL}
    For all $V \in \mathcal{I}_{mid}$ with $E(V) = \emptyset$,
    \begin{align*}
    &\frac{1}{|V|!}\matL_{V}\matL_{V}^\T \succeq \matId_{Sym,V} - \sum_{U \in \mathcal{I}_{mid}}\frac{1}{|U|!}{\sum_{\substack{\text{nontrivial }\\\sigma \in \mathcal{L}_{V}:U_{\sigma} = U}}{2\lambda_{\sigma}||\mM_{\sigma}||}\matId_{Sym,U}}
    \end{align*}
\end{claim}
\begin{proof}[Proof of claim.]
We have that
\begin{align*}
    \frac{1}{|V|!}\matL_{V}\matL_{V}^\T &= \matId_{Sym,V}
    +
    \sum_{\substack{\text{non-trivial }\sigma \in \mathcal{L}_{V}}}{\lambda_{\sigma}(\mM_{\sigma}+ \mM_{{\sigma}}^\T}) 
    + \frac{1}{|V|!}
    \sum_{\text{non-trivial }\substack{\sigma,\sigma' \in \mathcal{L}_{V}}}
    {\lambda_{\sigma}\lambda_{{\sigma'}}\mM_{\sigma}\mM_{{\sigma'}}^\T}
\end{align*}
The last term is PSD. Hence it remains to bound the middle term.

\begin{align*}
    \sum_{\substack{\text{non-trivial }\sigma \in \mathcal{L}_{V}}}{\lambda_{\sigma}(\mM_{\sigma}+ \mM_{{\sigma}}^\T}) &= \sum_{U \in \calI_{mid}: U \neq V}\sum_{\substack{\sigma \in \mathcal{L}_{V}:\\ U_\sig = U}}{\lambda_{\sigma}(\mM_{\sigma}+ \mM_{{\sigma}}^\T}) \\
    &= \sum_{U \in \calI_{mid}}\frac{1}{|U|!}\sum_{\substack{\text{non-trivial }\sigma \in \mathcal{L}_{V}:\\ U_\sig = U}}{\lambda_{\sigma}\Id_{Sym, U}^{1/2}(\mM_{\sigma}+ \mM_{{\sigma}}^\T})\Id_{Sym, U}^{1/2}\\
    &\psdgeq -\sum_{U \in \calI_{mid}}\frac{1}{|U|!}\sum_{\substack{\text{non-trivial }\sigma \in \mathcal{L}_{V}:\\ U_\sig = U}}{\lambda_{\sigma}\Id_{Sym, U}^{1/2}(\norm{\mM_{\sigma}}+ \norm{\mM_{{\sigma}}^\T}})\Id_{Sym, U}^{1/2}\\
    &= -\sum_{U \in \calI_{mid}}\frac{1}{|U|!}\sum_{\substack{\text{non-trivial }\sigma \in \mathcal{L}_{V}:\\ U_\sig = U}} 2\lambda_{\sigma}\norm{\mM_{\sigma}}\Id_{Sym, U}
\end{align*}
which completes the proof of the claim.
\end{proof}
Using the claim,
    \begin{align*}
    \sum_{V \in \mathcal{I}_{mid}}{\frac{w_V}{|V|!}\matL_{V}\matL_{V}^\T} &\succeq \sum_{V \in \mathcal{I}_{mid}}{{w_V}\matId_{Sym,V}} - \sum_{U,V \in \mathcal{I}_{mid}}{\frac{w_V}{|U|!}\sum_{\substack{\text{nontrivial}\\\sigma \in \mathcal{L}_{V}:U_{\sigma} = U}}{2\lambda_{\sigma}||\mM_{\sigma}||}\matId_{Sym,U}} \\
    &= \sum_{V \in \mathcal{I}_{mid}}{{w_V}\matId_{Sym,V}} - \sum_{U \in \mathcal{I}_{mid}}\left(\sum_{V \in \calI_{mid}}\sum_{\substack{\text{nontrivial}\\\sigma \in \mathcal{L}_{V}:U_{\sigma} = U}}{\frac{w_V}{c(\sig)|U|!} 2c(\sig)\lambda_{\sigma}||\mM_{\sigma}||}\right) \matId_{Sym,U} \\
    &\psdgeq \sum_{V \in \mathcal{I}_{mid}}{{w_V}\matId_{Sym,V}} - \sum_{U \in \mathcal{I}_{mid}}{\left(\max_{V \in \calI_{mid}}w_V\max_{\substack{\text{nontrivial}\\\sigma \in \mathcal{L}_{V}:U_{\sigma} = U}}{\{2c(\sigma)\lambda_{\sigma}||\mM_{\sigma}||\}}\right)\matId_{Sym,U}} \\
    &\succeq \sum_{V \in \mathcal{I}_{mid}}{{w_V}\matId_{Sym,V}} - \frac{1}{2}\sum_{U \in \mathcal{I}_{mid}}{w_U\matId_{Sym,U}} \qquad\qquad\qquad (\text{by assumption})\\
    &= \frac{1}{2}\sum_{V \in \mathcal{I}_{mid}}{{w_V}\matId_{Sym,V}}
    \end{align*}
\end{proof}

Now we calculate the bound on the weights to deduce \cref{lem:well-conditionedness}.

\begin{lemma}
For all left shapes $\sigma$ with $E(V_\sig) = \emptyset$,
    \[
    \lam_{\sigma} \normapx{\sigma} \leq n^{(1-\al)\left(\frac{|U_{\sigma}|-|V_{\sigma}|}{2} \right) -(\gam-\al\beta)|E(\sigma)|}
    \]
\end{lemma}
\begin{proof}
By \cref{lem:GeneralCoefficientTimesNorm} we have
\begin{align*}
    \lam_\sig \normapx{\sig} &= n^{(1-\al)\left(\frac{|U_{\sigma}|+|V_{\sigma}|}{2} \right) - \left(\frac{1}{2}-\al \right)w(\sigma) -\frac{w(V_{\sigma})}{2} -(\gam-\al\beta)|E(\sigma)|}
\end{align*}
Substituting $w(\sig) \geq w(V_\sig)$ since $\sig$ is a left shape, and $w(V_\sig) = |V_\sig|$,
\begin{align*}
    \lam_\sig \normapx{\sig} &\leq n^{(1-\al)\left(\frac{|U_{\sigma}|+|V_{\sigma}|}{2} \right) - \left(1-\al \right)|V_\sig| -(\gam-\al\beta)|E(\sigma)|}\\
    &= n^{(1-\al)\left(\frac{|U_{\sigma}|-|V_{\sigma}|}{2} \right) -(\gam-\al\beta)|E(\sigma)|}
\end{align*}
\end{proof}

\begin{corollary}
For all $U, V \in \calI_{mid}$ such that $E(V) = \emptyset$, \[\max_{\substack{\text{nontrivial}\\\sigma \in \mathcal{L}_{V}:U_{\sigma} = U}}{\left\{2c(\sigma)\lambda_{\sigma}||\mM_{\sigma}||\right\}} \leq n^{(1-\alpha)\left(\frac{|U_\sig| - |V_\sig|}{2}\right)}\]
\end{corollary}
\begin{proof}
    \begin{align*}
    2c(\sigma)\lambda_{\sigma}||\mM_{\sigma}|| &\leq 2c(\sig) B_{adjust}(\sig) n^{(1-\alpha)\left(\frac{|U_\sig| - |V_\sig|}{2}\right) - (\gam - \al\beta)|E(\sig)|}\\
    &\leq (16 D_V)^{|E(\sig)|} n^{(1-\alpha)\left(\frac{|U_\sig| - |V_\sig|}{2}\right) - (\gam - \al\beta)|E(\sig)|}\\
    &\leq n^{(1-\alpha)\left(\frac{|U_\sig| - |V_\sig|}{2}\right)}\,.
    \end{align*}
\end{proof}
\begin{corollary}
    Choosing $w_U = O\left(n^{(1-\al)\frac{|U|}{2}}\right)$ satisfies the assumption of \cref{lem:well-conditionedness-weights}.
\end{corollary}

Recall that because the size of the SMVS is at most $\dsos$,
then $|U| \leq \dsos$.
The maximum of $w_U$ is $w_{\dsos} \leq O(n^{\dsos})$,
therefore we conclude \cref{lem:well-conditionedness}.

\section{Computing \texorpdfstring{$\pE[1]$}{E[1]}}
\label{sec:formal-pseudocal}

\pseudoEone*

\begin{proof}
\begin{align*}
    \abs{\pE[1]-1} &= \big\lvert\sum_{\substack{\al \in \calS:\\ U_\al = V_\al =\emptyset,\\ E(\alpha) \neq \emptyset} }\lam_\al \mM_\al\big\rvert
        \leq  \max_{\al \neq \emptyset: U_\al = V_\al  =\emptyset}\{ c(\al) \lambda(\al) \cdot \|\mM_\al\|\}
\end{align*}
Letting $S$ be the SMVS for $\al$, observe that 
\begin{enumerate}
    \item $\lambda_\al \|\mM_\al\| \leq n^{ - (\frac{1}{2}-\al)w(\al) - \frac{w(S)}{2} -  (\gam-\al\beta)|E(\al)| } \cdot B_{adjust}(\al) $
    \item $w(S) \geq  -\eta - 2 \abs{E(S)} \log_n D_V  $
    \item $w(\al) \geq -\eta - 2|E(\al)|\log_n D_V$
    \item $c(\al) \cdot B_{adjust}(\al) \cdot n^{ 2\eta + 4 \abs{E(\al)} \log_n D_V  } \leq n^{2\eta}(16D_V)^{8|E(\al)|}$ which is less than $n^{\frac{\gam - \al\beta}{2}|E(\al)|}$ when $\eta \leq \frac{\gamma - {\al}{\beta}}{4} - 4\log_{n}(16D_V)$.
\end{enumerate}
Therefore, we can bound \[ 
 \max_{\al \neq \emptyset: U_\al = V_\al  =\emptyset}\{ c(\al) \lambda(\al) \cdot \|\mM_\al\| \} \leq n^{\frac{\al\beta-\gam}{2}}\leq o(1)
\]
as $|E(\al)|>0$, giving us that $\left|\pE[1] - 1\right| = o(1)$ with probability at least $1 - 2n^{4\log_{n}(16D_V) - \frac{\gamma - {\al}{\beta}}{4}}$.
\end{proof}

\pseudoEXv*

\begin{proof}
This can be shown by decomposing 
$\sum_{i=1}^n\pE[\xsos_i]$
and $\sum_{\{i,j\} \in E(G)}\pE[\xsos_i\xsos_j]$ in terms of ribbons.

For each ribbon $R$ with edges $E(R)$ and $A_R = B_R = \emptyset$, this ribbon appears in $\sum_{i=1}^n\pE[\xsos_i]$ in two ways.
\begin{enumerate}
\item For each $i \in V(R)$, the ribbon $R'$ with $E(R') = E(R)$, $V(R') = V(R)$, $A_{R'} = (i)$, and $B_{R'} = \emptyset$ appears in $\pE[\xsos_i]$ with coefficient $\left(\frac{k}{n}\right)^{|V(R)|}\left(\frac{q - p}{\sqrt{p(1-p)}}\right)^{|E(R)|}$. This happens for $|V(R)|$ different $i$.
\item For each $i$ which is not in $V(R)$, the ribbon $R'$ with $E(R') = E(R)$, $V(R') = V(R) \cup \{i\}$, $A_{R'} = (i)$, and $B_{R'} = \emptyset$ appears in $\pE[\xsos_i]$ with coefficient $\left(\frac{k}{n}\right)^{|V(R)| + 1}\left(\frac{q - p}{\sqrt{p(1-p)}}\right)^{|E(R)|}$. This happens for $n - |V(R)|$ different $i$.
\end{enumerate}
Thus, for each shape $\alpha$ with $U_{\alpha} = V_{\alpha} = \emptyset$ where $|V(\alpha)| \leq D_V$, $\mM_{\alpha}$ appears in $\sum_{i=1}^n\pE[\xsos_i]$ with coefficient 
\[
\frac{1}{|\Aut(\alpha)|}\left(\frac{k(n - |V(\alpha)|)}{n}\one_{|V(\alpha)| < D_V} + |V(\alpha)|\right)\left(\frac{k}{n}\right)^{|V(R)|}\left(\frac{q - p}{\sqrt{p(1-p)}}\right)^{|E(\alpha)|}
\]
The dominant term is the trivial shape $\alpha$ with no vertices or edges which gives a contribution of exactly $k$. Using a similar analysis as the analysis used to bound $\left|\pE[1] - 1\right|$, it is not hard to show that the remaining terms have magnitude $o(k)$ with high probability.

To analyze $\sum_{\{i,j\} \in E(G)}\pE[\xsos_i\xsos_j]$, we use the identities $\one_{e \in E(G)}\chi_{e} = \sqrt{p(1-p)} + (1-p)\chi_{e}$ and $\one_{e \in E(G)} = p + \sqrt{p(1-p)}\chi_{e}$.

For each ribbon $R$ with edges $E(R)$ and $A_R = B_R = \emptyset$, this ribbon appears in $\sum_{i < j}{\one_{\{i,j\} \in E(G)}\pE[\xsos_i\xsos_j]}$ in several ways.
\begin{enumerate}
\item For each $i < j \in V(R)$ such that $\{i,j\} \in E(R)$, the ribbon $R'_1$ with $E(R'_1) = E(R)$, $V(R') = V(R)$, $A_{R'} = (i)$, and $B_{R'} = (j)$ appears in $\pE[\xsos_i\xsos_j]$ with coefficient $\left(\frac{k}{n}\right)^{|V(R)|}\left(\frac{q - p}{\sqrt{p(1-p)}}\right)^{|E(R)|}$. This is then multiplied by $(1-p)$ because of the edge indicator $\one_{\{i,j\} \in E(G)}$.

Similarly, the ribbon $R'_2$ with $E(R'_2) = E(R) \setminus \{i,j\}$, $V(R') = V(R)$, $A_{R'} = (i)$, and $B_{R'} = (j)$ appears in $\pE[\xsos_i\xsos_j]$ with coefficient $\left(\frac{k}{n}\right)^{|V(R)|}\left(\frac{q - p}{\sqrt{p(1-p)}}\right)^{|E(R)| - 1}$. This is then multiplied by $\sqrt{p(1-p)}$ because of the edge indicator $\one_{\{i,j\} \in E(G)}$. 

This gives a total contribution of $\left(1 + \frac{p(1-p)}{q-p}\right)\left(\frac{k}{n}\right)^{|V(R)|}\left(\frac{q - p}{\sqrt{p(1-p)}}\right)^{|E(R)|}$. This happens for $|E(R)|$ different $i < j$.

\item For each $i < j \in V(R)$ such that $\{i,j\} \notin E(R)$, the ribbon $R'_1$ with $E(R'_1) = E(R)$, $V(R') = V(R)$, $A_{R'} = (i)$, and $B_{R'} = (j)$ appears in $\pE[\xsos_i\xsos_j]$ with coefficient $\left(\frac{k}{n}\right)^{|V(R)|}\left(\frac{q - p}{\sqrt{p(1-p)}}\right)^{|E(R)|}$. This is then multiplied by $p$ because of the edge indicator $\one_{\{i,j\} \in E(G)}$.

Similarly, the ribbon $R'_2$ with $E(R'_2) = E(R) \cup \{i,j\}$, $V(R') = V(R)$, $A_{R'} = (i)$, and $B_{R'} = (j)$ appears in $\pE[\xsos_i\xsos_j]$ with coefficient $\left(\frac{k}{n}\right)^{|V(R)|}\left(\frac{q - p}{\sqrt{p(1-p)}}\right)^{|E(R)| + 1}$. This is then multiplied by $\sqrt{p(1-p)}$ because of the edge indicator $\one_{\{i,j\} \in E(G)}$. 

This gives a total contribution of $q\left(\frac{k}{n}\right)^{|V(R)|}\left(\frac{q - p}{\sqrt{p(1-p)}}\right)^{|E(R)|}$. This happens for $\binom{V(R)}{2} - |E(R)|$ different $i < j$.

\item For each $i < j \in V(R)$ such that $i \in V(R)$ but $j \notin V(R)$, the ribbon $R'_1$ with $E(R'_1) = E(R)$, $V(R') = V(R) \cup \{j\}$, $A_{R'} = (i)$, and $B_{R'} = (j)$ appears in $\pE[\xsos_i\xsos_j]$ with coefficient $\left(\frac{k}{n}\right)^{|V(R)|+1}\left(\frac{q - p}{\sqrt{p(1-p)}}\right)^{|E(R)|}$. This is then multiplied by $p$ because of the edge indicator $\one_{\{i,j\} \in E(G)}$.

Similarly, the ribbon $R'_2$ with $E(R'_2) = E(R) \cup \{i,j\}$, $V(R') = V(R)  \cup \{j\}$, $A_{R'} = (i)$, and $B_{R'} = (j)$ appears in $\pE[\xsos_i\xsos_j]$ with coefficient $\left(\frac{k}{n}\right)^{|V(R)|+1}\left(\frac{q - p}{\sqrt{p(1-p)}}\right)^{|E(R)| + 1}$. This is then multiplied by $\sqrt{p(1-p)}$ because of the edge indicator $\one_{\{i,j\} \in E(G)}$. 

This gives a total contribution of $\frac{kq}{n}\left(\frac{k}{n}\right)^{|V(R)|}\left(\frac{q - p}{\sqrt{p(1-p)}}\right)^{|E(R)|}$.

The same analysis holds for the case when $i \notin V(R)$ and $j \in V(R)$. These two cases happen for $|V(R)|(n - |V(R)|)$ different $i < j$.

\item For each $i < j \in V(R)$ such that $i,j \notin V(R)$, the ribbon $R'_1$ with $E(R'_1) = E(R)$, $V(R') = V(R) \cup \{i,j\}$, $A_{R'} = (i)$, and $B_{R'} = (j)$ appears in $\pE[\xsos_i\xsos_j]$ with coefficient $\left(\frac{k}{n}\right)^{|V(R)|+2}\left(\frac{q - p}{\sqrt{p(1-p)}}\right)^{|E(R)|}$. This is then multiplied by $p$ because of the edge indicator $\one_{\{i,j\} \in E(G)}$.

Similarly, the ribbon $R'_2$ with $E(R'_2) = E(R) \cup \{i,j\}$, $V(R') = V(R)  \cup \{i,j\}$, $A_{R'} = (i)$, and $B_{R'} = (j)$ appears in $\pE[\xsos_i\xsos_j]$ with coefficient $\left(\frac{k}{n}\right)^{|V(R)|+2}\left(\frac{q - p}{\sqrt{p(1-p)}}\right)^{|E(R)| + 1}$. This is then multiplied by $\sqrt{p(1-p)}$ because of the edge indicator $\one_{\{i,j\} \in E(G)}$. 

This gives a total contribution of $\frac{{k^2}q}{n^2}\left(\frac{k}{n}\right)^{|V(R)|}\left(\frac{q - p}{\sqrt{p(1-p)}}\right)^{|E(R)|}$. This happens for $\binom{(n - |V(R)|)}{2}$ different $i < j$.
\end{enumerate}

Putting everything together, each ribbon $R$ with $A_R = B_R = \emptyset$ and $|V(R)| \leq D_V$ appears with coefficient
\begin{align*}
\Bigg(&\left(1 + \frac{p(1-p)}{q-p}\right)|E(R)|+ q\left(\binom{V(R)}{2} - |E(R)|\right) + \frac{kq}{n}|V(R)|(n - |V(R)|)\one_{|V(R)| < D_V} \\
&+\frac{{k^2}q}{n^2}\binom{(n - |V(R)|)}{2}\one_{|V(R)| < D_V - 1}\Bigg)\left(\frac{k}{n}\right)^{|V(R)|}\left(\frac{q - p}{\sqrt{p(1-p)}}\right)^{|E(R)|}
\end{align*}
in $\sum_{\{i,j\} \in E(G)}\pE[\xsos_i\xsos_j]$. Thus, for each shape $\alpha$ with $U_{\alpha} = V_{\alpha} = \emptyset$ and $|V(\alpha)| \leq D_V$, the graph matrix $M_{\alpha}$ appears with coefficient
\begin{align*}
\frac{1}{|\Aut(\alpha)|}\Bigg(&\left(1 + \frac{p(1-p)}{q-p}\right)|E(\alpha)|+ q\left(\binom{V(\alpha)}{2} - |E(\alpha)|\right) + \frac{kq}{n}|V(\alpha)|(n - |V(\alpha)|)\one_{|V(\alpha)| < D_V} \\
&+\frac{{k^2}q}{n^2}\binom{(n - |V(R)|)}{2}\one_{|V(R)| < D_V - 1}\Bigg)\left(\frac{k}{n}\right)^{|V(\alpha)|}\left(\frac{q - p}{\sqrt{p(1-p)}}\right)^{|E(\alpha)|}
\end{align*}
in $\sum_{\{i,j\} \in E(G)}\pE[\xsos_i\xsos_j]$.

The dominant term in $\sum_{\{i,j\} \in E(G)}\pE[\xsos_i\xsos_j]$ comes from the empty shape $\alpha$ with no vertices or edges. This gives $\frac{{k^2}q}{n^2}\binom{n}{2} \approx \frac{{k^2}q}{2}$. Using a similar analysis as the analysis used to bound $\left|\pE[1] - 1\right|$, it is not hard to show that the remaining terms have magnitude $o({k^2}q)$ with high probability.
\end{proof}

\end{document}